%% file: manuscript_6_14_2026.tex
\documentclass[a4paper,11pt, oneside]{article} 
\usepackage{palatino}
\usepackage{amssymb}
\usepackage{amsmath}
\usepackage{amsfonts}
\usepackage{amsthm}
\usepackage{changepage}
\newtheorem{proposition}{Proposition}[section]
\newtheorem{lemma}{Lemma}
\usepackage{anyfontsize}
\usepackage{lmodern}  

\usepackage[margin=0.9in]{geometry}
\usepackage[onehalfspacing]{setspace}
\usepackage{fancyhdr}
\usepackage{appendix} 
\usepackage{graphicx}
\usepackage{float}
\usepackage{fp}
\usepackage{rotating}
\usepackage[para,online,flushleft]{threeparttable}
\usepackage{tabulary}
\usepackage{tabularx}
\usepackage{natbib}
\setcitestyle{authoryear}
\usepackage[utf8]{inputenc}
\usepackage[english]{babel}
\usepackage{latexsym}
\usepackage{amsxtra}
\usepackage{dcolumn}
\usepackage[mathcal]{eucal}
\usepackage{enumerate}
\usepackage[labelfont=bf]{caption}
\usepackage{subcaption}
\usepackage{lscape}
\usepackage{multirow}
\usepackage{paralist}
\usepackage{indentfirst}
\usepackage[affil-it]{authblk}
\usepackage[dvipsnames,svgnames]{xcolor}
\usepackage{booktabs}
\usepackage{threeparttablex}
\setTableNoteFont{\fontsize{8}{9.6}}
\usepackage{titlesec}
\usepackage{bbm}
\usepackage{afterpage}
\usepackage{pdflscape}
\usepackage{url}
\usepackage{csquotes}
\usepackage{mathtools}
\usepackage{dsfont}
\usepackage{array}
\usepackage{accents}

\newcolumntype{x}[1]{>{\centering\let\newline\\\arraybackslash\hspace{0pt}}p{#1}}
\usepackage {tikz}
\usepackage{textcomp}
\usetikzlibrary{positioning}
\usetikzlibrary{shapes,arrows}
\usepackage{siunitx}
\newcolumntype{d}[1]{D{.}{.}{#1}}
\definecolor{blue}{rgb}{0,0.08,0.5}
\definecolor{red}{rgb}{.6,0,0}
\definecolor{green}{rgb}{0,0.376,0}

\usepackage{xparse}
\usepackage{adjustbox}
\usepackage{comment}
\usepackage{nicefrac}
\usepackage[most]{tcolorbox}
\usepackage{placeins}

\makeatletter
\renewcommand{\and}{\\[-0.2em]}
\makeatother

\ExplSyntaxOn
\NewDocumentCommand{\longdash}{ O{2} }
 {
  --\prg_replicate:nn { #1 - 1 } { \negthinspace -- }
 }
\ExplSyntaxOff 

\usepackage{import}
 \newcommand{\sym}[1]{{#1}} 
\newcommand{\Expect}{{\rm I\kern-.3em E}} 
\newsavebox\mytabularbox
\newcount\figwidthc
\newcount\textwidthc
\newcolumntype{Y}{>{\centering\arraybackslash}X}
\newcommand\totextwidth[1]{%
  \sbox{\mytabularbox}{#1}%
  \figwidthc=\wd\mytabularbox%
  \textwidthc=\textwidth%
  \FPdiv\scaleratio{\the\textwidthc}{\the\figwidthc}%
  \FPmin\scaleratio{\scaleratio}{1}%
  \scalebox{\scaleratio}{\usebox{\mytabularbox}}%
}
\usepackage[bookmarks, pdftitle={}, pdfauthor={}, colorlinks=true, linkcolor=red, citecolor=blue, urlcolor=blue]{hyperref} 
\begin{document}
\raggedbottom
\enlargethispage{\baselineskip}
\title{\textbf{Hiring Discrimination and the Task Content of Jobs: \\ Evidence from a Large-Scale R\'{e}sum\'{e} Audit} \thanks{We are thankful for helpful comments from Duha Altindag, Jaime Arellano-Bover, Nabamita Dutta, Marissa Eckrote-Nordland, David Jaeger, Joanna Lahey, James Murray, Isaac McFarlin, Sara Seals, Mike Stern, and Adam Stivers. We also thank seminar participants at Auburn University, the Southern Economic Association, and Yale University. Any errors are our own.} }

\author{
Sharon Braun\thanks{Braun, Auburn University, \href{mailto:szb0288@auburn.edu}{szb0288@auburn.edu}}, Jonathan Bushnell, Zachary Cowell, David Dowling
\and
Samuel Goldstein, Andrew Johnson, George Miller, John M. Nunley\thanks{Nunley, University of Wisconsin--La Crosse, \href{mailto:jnunley@uwlax.edu}{jnunley@uwlax.edu}}
\and
R. Alan Seals\thanks{Seals, Auburn University, \href{mailto:alan.seals@auburn.edu}{alan.seals@auburn.edu} [Corresponding Author]}, and Mingzhou Wang\thanks{Wang, University of Georgia, \href{mailto:mingzhou.wang@uga.edu}{mingzhou.wang@uga.edu}}
}

\date{\today}

\maketitle

\thispagestyle{empty} 

\begin{abstract}

\begin{singlespace}
    
\noindent We conducted a r\'{e}sum\'{e} audit of 36,880 applications to 9,220 job advertisements for new college graduates across the United States and linked each advertisement to occupation-level task measures. In management occupations, callbacks for Black men, Black women, White women, and Hispanic men are 18 to 28 percent lower than those of otherwise identical White men. Callback gaps are largest in jobs that demand analytical and interpersonal skills and involve little routine work. We develop a theoretical framework in which discrimination narrows when objective criteria constrain the hiring decision and widens when subjective judgment carries greater weight. In the framework's central test, randomly assigned credentials that raise callbacks narrow non-White-male gaps only where applicants are evaluated against structured criteria. Consistent with this experimental result, gaps are smaller in jobs with more codifiable task content and in occupations where advertisements more frequently mention tests, certifications, or background checks. Across these measures, callback gaps are narrower where screening can rely more heavily on verifiable criteria.

\vspace{0.75cm}
\noindent\textbf{JEL Categories:} J23, J71, C93 

\vspace{0.75cm}

R\'{e}sum\'{e} Audit, College Graduates

\vspace{0.75cm}

\end{singlespace}

\vfill 
\end{abstract}

\pagebreak
\setcounter{page}{1} 

\section{Introduction}
\label{section:intro}
  
  \noindent Racial and gender differences in the tasks workers perform are persistent features of the U.S. labor market. Black--White differences in routine, manual, and contact tasks have narrowed since the 1960s, but the gap in analytically demanding tasks persists \citep{hurst_task-based_2024}. Women's entry into high-status occupations such as medicine, law, and business was rapid through the 1980s but slowed during the 1990s \citep{goldin_career_2021,blau_gender_2024}. Because the pay attached to analytical tasks has risen, these differences account for a substantial share of the wage gap for both women \citep{blau_gender_2017} and Black workers \citep{gray_tasks_2025}. Whether hiring discrimination contributes to unequal access to these task bundles remains an open question.

  The scope for discrimination depends on how employers evaluate applicants. Some jobs can be screened against verifiable criteria, such as a pre-employment test, a certification or license, or a background check. Other jobs require employers to infer fit largely from the r\'{e}sum\'{e} itself. Subjective judgment of analytical and interpersonal fit creates greater scope for group-based beliefs to influence screening than do verifiable criteria. Routine-intensive jobs are more amenable to the first kind of screening and analytically and interpersonally demanding jobs rely more heavily on the second. The relevant margin is the job's evaluative discretion, the extent to which its screening must rest on subjective judgment rather than on verifiable criteria. We formalize this margin in a model of employer screening (Section \ref{section:theory}). The model predicts that callback gaps shrink where verifiable criteria constrain the evaluation. Discrimination at the screening stage is then not solely a fixed attribute of the employer but also varies with how the job is screened.

  We test this prediction with a randomized r\'{e}sum\'{e} audit of entry-level hiring across the United States. In 2016 and 2017, we sent 36,880 applications to 9,220 job advertisements posted by 4,969 firms. Each application represented a new college graduate, limiting variation in prior work histories. Applicant names and qualifications were randomly assigned within postings, so the design identifies callback differences within a job. We classify each advertisement into one of 175 detailed occupations using its title and text and link those occupations to task measures from the Occupational Information Network (O*NET) and the American Community Survey. Because task content is a feature of the vacancy rather than a randomized treatment, comparisons across task environments identify where callback discrimination concentrates but do not establish why it does so.

  We begin with broad occupation groups, which bundle tasks in different proportions. Management occupations make heavy analytical and interpersonal demands, and callback gaps are larger there than in the other groups. A Black man faces a callback gap of 5.1 percentage points in management compared with 1.7 points in office and administrative support, and he must send roughly seven applications for every five a White man sends to receive the same number of callbacks. Across the four non-reference groups with statistically significant management gaps, callback rates are 18 to 28 percent below the White-male rate. Across all jobs the pooled non-White-male gap implies that these applicants must submit about 15 percent more applications than White men to obtain the same number of callbacks, and the burden is concentrated in the analytically and interpersonally demanding jobs. Broad occupation categories do not establish whether particular tasks or other features of these jobs generate the pattern.

  This pattern extends beyond occupation labels. When we group jobs by task profile with K-means clustering, discrimination concentrates in jobs that combine high analytical and interpersonal intensity with low routine content, regardless of their major occupation. We summarize the extent to which a job's screening must rely on subjective judgment in a discretion index, and callback gaps are smaller where that index is low. Task content is not randomly assigned, so these comparisons remain descriptive. We therefore place greater weight on experimentally assigned variation within advertisements.

  The randomly assigned credentials provide the framework's central test, because they vary verifiable information within the advertisement while holding the job fixed. Two credentials, an internship emphasizing sales and client interaction and a combination of programming and data-analysis skills, substantially raise callback rates. In low-discretion jobs they narrow non-White-male callback gaps by 4.1 and 3.7 percentage points, both distinguishable from zero. These effects exceed the roughly two-percentage-point gap in low-discretion jobs, so they imply full closure, though the estimates are too imprecise to rule out a smaller effect. In high-discretion jobs the same credentials provide no differential benefit. Because the credentials are randomized within the advertisement, their differential effect on the callback gap is causal. However, the contrast between low- and high-discretion jobs still relies on cross-occupation variation. The result holds when r\'{e}sum\'{e} characteristics are interacted with occupation groups and when firm fixed effects replace advertisement fixed effects, and tests find no within-ad interference. Verifiable credentials thus improve non-White-male applicants' callback access to routine-intensive positions but do little to open the analytically and interpersonally demanding roles where discrimination is concentrated and wage returns are higher.

  The potential consequences of this exclusion depend on what an entry job carries forward. Early access to task bundles may shape longer-run skill and earnings trajectories \citep{arellano-bover_effect_2022}, and because our applicants are new college graduates, we observe discrimination at the point of labor-market entry. Initial placement has persistent effects on later earnings and career paths \citep{von_wachter_persistent_2020,kahn_long-term_2010,oreopoulos_short-_2012}, and the job-specific experience that workers build early accounts for a large share of later wage gaps \citep{stinebrickner_job_2026,deming_why_2023}. Discrimination in entry jobs that build this experience can compound over a career \citep{bohren_systemic_2025}. Heterogeneity in callback discrimination across job types may therefore matter for inequality well beyond the callback stage.

  The mid-career workforce is already sorted along the same gradient. Among mid-career college-educated workers in the American Community Survey, Black and Hispanic workers are underrepresented in the high-discretion occupations where callback gaps concentrate. Black men are 9.5 percentage points less likely than White men to hold these jobs, and 7.7 points less likely among workers who studied the same fields as our applicants. We read this sorting as descriptive rather than as a causal consequence of callback discrimination.

  \noindent\textit{Contribution to the Literature.} First, we recover a dimension of hiring discrimination that within-firm audit designs hold fixed. By comparing applicants within the same firm, firm-fixed-effects designs identify which firms discriminate through a contrast that pools across the different jobs a firm posts \citep{kline_systemic_2022}. A parallel literature documents demographic differences in the tasks workers perform and the wages attached to them \citep{gray_tasks_2025,hurst_task-based_2024}, and \citet{hurst_task-based_2024} decompose those differences structurally. Observational sorting, however, cannot reveal how employers treat identical applicants at the callback stage. Randomizing applicant characteristics within advertisements and linking each advertisement to occupation-level task measures allows us to estimate how callback discrimination varies with the task content of the job. The occupation-group comparison, the task-profile clusters, and the discretion index are three views of the same cross-occupation pattern in the O*NET measures, not three independent results.

  Second, the randomized credentials identify, within the advertisement, the causal effect of verifiable information on the callback gap, though whether that effect differs between high- and low-discretion jobs is an occupation-level comparison. Third, we build a measure of each occupation's screening environment from the advertisement text rather than from O*NET. Because it is only weakly correlated with the discretion index, it provides a separate check on the pattern without relying on the same task taxonomy, and callback gaps shrink where advertisements announce standardized screening instruments. Neither the cross-occupation pattern nor the screening-text evidence can independently separate the mechanism from job complexity. The evidentiary case instead rests primarily on the randomized credential experiment, supported by two directionally consistent cross-occupation patterns constructed from different data. Taken together, these results are difficult to reconcile with a simple account in which job complexity alone mechanically raises discrimination.


\section{Theoretical Framework}
\label{section:theory}

\subsection{Related Literature and Motivation}

\subsubsection{Task-Based Production and Discrimination}

\noindent Task-based production models treat occupations as bundles of tasks that workers or capital can perform \citep{autor_skill_2003,autor_growth_2013,acemoglu_skills_2011,autor_putting_2013}. This approach has organized research on technological change and wage polarization. \citet{hurst_task-based_2024} apply it to Black--White inequality in occupational sorting. Their Contact task measures interaction with coworkers and customers, while their Abstract task captures analytical work. Among Black and White men from 1960 through 2018, the occupational Contact gap nearly closed while the Abstract gap persisted. The price of Abstract tasks also rose after 1980. Their structural estimates attribute the Contact barrier mainly to nonpecuniary forces and the Abstract barrier largely to skill differences and wage discrimination. We use their Contact measure because it captures a dimension of job requirements that can matter for screening. Contact is not itself a measure of structured evaluation, and the \citet{hurst_task-based_2024} framework does not observe employer callback decisions.

Long-run transition evidence points to the same economic stakes. \citet{gray_tasks_2025} document that Black men were less likely to move into higher-return non-routine analytical work and more likely to remain in or move into physically intensive work.\footnote{\citet{dicandia_technological_2021} links persistent racial wage gaps in the United States to routine-biased technological change.} Analytical and interpersonal skills have become increasingly valuable in the labor market \citep{deming_growing_2017,deming_growing_2021,deming_why_2023}. These outcomes reflect education, worker choices, networks, and employer behavior operating over a career. They do not reveal how an employer evaluates otherwise identical applicants at the beginning of a vacancy. That distinction directs attention to audit evidence.

\subsubsection{Evidence from Audit Studies}

\noindent R\'{e}sum\'{e} audits randomize applicant attributes and estimate employer responses while holding submitted qualifications fixed \citep{bertrand_are_2004,lippens_state_2023}. Their results vary across job contexts.\footnote{Meta-analyses document persistent racial and ethnic callback gaps in the United States \citep{quillian_meta-analysis_2017} and substantial variation across countries and periods \citep{quillian_countries_2019,quillian_trends_2023}.} Gender callback gaps also change with job context. \citet{booth_employers_2010} find higher callback rates for women in female-dominated Australian occupations. \citet{ahmed_gender_2021} find that women's callback advantage in Sweden is concentrated in female-dominated occupations. \citet{kline_systemic_2022} document employer-level heterogeneity in gender contact gaps in the United States. Racial callback gaps can widen with credentials that signal productivity or match quality \citep{nunley_racial_2015}.

Several mechanisms can generate this heterogeneity before an employer makes a callback decision. \citet{bartos_attention_2016} show that group identity changes costly information acquisition, with lower minority attention in selective labor markets and greater attention in their rental-housing setting. \citet{kessler_incentivized_2019} study how hiring interest and beliefs about offer acceptance jointly enter r\'{e}sum\'{e} evaluation. \citet{kline_reasonable_2021} develop a method for locating discrimination at individual firms, and \citet{kline_systemic_2022} relate employer-level gaps to task measures. The recurring empirical object is the interaction between applicant group and the evaluative demands of a job.

\citet{bohren_systemic_2025} formalize why the placement of discrimination across jobs can matter beyond the initial response. In their iterated audit, an early callback gap for inexperienced Black applicants reduces experience acquisition and contributes to later hiring gaps. Their framework distinguishes the initial group-specific action from disparities generated through experience and other signals. Our study estimates the initial callback margin rather than the later sequence. Its consequences depend on the jobs in which it appears and on how workers sort into those jobs.

\subsubsection{Occupational Self-Selection}

\noindent Occupational allocation reflects employer demand and workers' job-search and career choices \citep{pager_race_2015,le_barbanchon_gender_2021}. Same-demographic role models can affect education outcomes and career choices \citep{kofoed_effect_2019,porter_gender_2020}, and beliefs about expected returns can shape college-major decisions \citep{wiswall_determinants_2015,zafar_college_2013}.

One account of the gender wage gap is that women often avoid \enquote{greedy jobs}, whose earnings premium reflects long or unpredictable hours \citep{goldin_most_2016,goldin_career_2021}.\footnote{\citet{bertrand_gender_2020} identifies two related sources of gender inequality. Educational choices can channel women toward lower-paying fields despite greater educational attainment, and earnings often fall after motherhood. The motherhood penalty operates in part through reduced hours and exits from full-time work, a pattern consistent with the appeal of jobs with more predictable schedules.} Using the Berea Panel Study, \citet{stinebrickner_job_2026} find that gender differences in task-specific work experience, especially in high-skilled information tasks, account for roughly 35 percent of the gender wage gap within ten years of college graduation. In the same study, precollege aptitude predicts the tasks graduates take up, a pattern consistent with comparative advantage.

Observed sorting therefore cannot identify demand-side discrimination. Every advertisement in our audit receives r\'{e}sum\'{e}s whose applicant characteristics are randomly assigned. Self-selection into occupations does not confound the callback comparison within a posting. We then examine whether the randomized callback gap varies with task-based proxies for evaluative discretion. The task measures remain job characteristics, so the model that follows makes explicit the screening mechanism they are intended to capture.

\subsection{A Model of Information-Based Discrimination}

\subsubsection{Overview}

\noindent The model organizes the cross-occupation variation described above into a framework in which the scope for discrimination depends on how a job is screened. It builds on two premises.

\begin{enumerate}
    \item Jobs differ in the extent to which hiring decisions rely on \emph{subjective} assessments (interviews, unstructured screening) versus \emph{objective} information (standardized tests, verifiable credentials).
    
    \item Subjective evaluation is noisier for minority than for majority applicants. We treat this differential precision as a reduced-form consequence of asymmetric attention, cultural distance, and communication frictions. \citet{bartos_attention_2016} model how decision-makers allocate effort during applicant evaluation, and our assumption captures one plausible downstream consequence of such attention discrimination in selective labor markets. By contrast, objective assessments are equally precise across groups.\footnote{Consistent with this premise, \citet{autor_does_2008} find that a national retail chain's standardized job test increased the tenure of its hires without reducing the minority share of hires, and they conclude that the test introduced no additional negative information about minority applicants.}
\end{enumerate}

\noindent These assumptions yield a compact summary of how discrimination varies with task content. Let \(D_{gj}\) denote the callback gap between majority and group-\(g\) minority applicants for job \(j\). To a first-order approximation,

\begin{equation}
\label{eq:main_result}
    D_{gj} \;\approx\; K_{gj}\, E_j^\ast, \qquad K_{gj} > 0,
\end{equation}

\noindent where \(E_j^\ast \in [0, 1)\) is the \emph{discretion index}, defined as the share of the hiring evaluation driven by subjective assessment, and \(K_{gj}\) is a positive scale factor. When objective criteria are strong, they dominate and \(E_j^\ast \approx 0\). As the job's reliance on subjective evaluation grows, group differences in evaluation noise have greater scope to affect outcomes, and discrimination widens. Unlike models based solely on tastes or beliefs, this framework predicts that discrimination varies systematically with the evaluability of job tasks. The rest of this section develops the model that generates this relationship and connects it to observable task measures.

\subsubsection{Environment}

\noindent Workers from two demographic groups \(g \in \{M, m\}\) (majority, minority) apply to jobs. Productivity \(\theta\) is normally distributed and identical across groups, so any callback differences arise from evaluation rather than from underlying productivity.

For each applicant, firms observe two conditionally independent screening signals prior to the callback decision.\footnote{In our correspondence audit, we observe only the callback decision. Interviews occur after callbacks and are therefore outside the information set relevant for our observed outcome.} The subjective signal reflects unstructured assessment of application materials (for example, perceived \enquote{fit} or judgment-based r\'{e}sum\'{e} review) and takes the form
\[
s^{s}_{ij} = \theta_i + \varepsilon^{s}_{ij}.
\]
The objective signal reflects verifiable credentials and rule-based screening and is
\[
s^{o}_{ij} = \theta_i + \varepsilon^{o}_{ij}.
\]

\noindent The central asymmetry is that subjective evaluation is less precise for minority than for majority applicants. Let \(B_j > 0\) denote subjective evaluation noise for job \(j\). Majority applicants face noise variance \(B_j\), while group-\(g\) minority applicants face noise variance \(B_j / \pi_g\), where \(\pi_g \in (0,1)\). Because \(\pi_g < 1\), the minority noise variance exceeds the majority noise variance for every job. The subscript allows the precision loss to differ across minority groups. Asymmetric attention \citep{bartos_attention_2016}, cultural distance \citep{lang_racial_2012}, and unfamiliarity with minority signal environments \citep{aigner_statistical_1977} can all reduce the precision of unstructured screening. The parameter \(\pi_g\) summarizes these mechanisms in reduced form.

Let \(U_j > 0\) denote objective evaluation noise for job \(j\), the same for both groups. We define objective evaluation precision as \(P_j = 1/U_j\). Task content affects signal quality in reduced form.

\begin{itemize}
    \item Analytical and interpersonal tasks increase subjective noise. Skills such as \enquote{creative problem-solving} or \enquote{motivating others} are difficult to assess in brief interactions, and verifiable credentials capture them only coarsely, so r\'{e}sum\'{e} screening relies on subjective judgment. \(B_j\) is therefore increasing in analytical (\(a_j\)) and interpersonal (\(p_j\)) task intensity.\footnote{Consistent with this premise, advertisements for analytical- and interpersonal-intensive jobs, and for management, express greater demand for judgment-based attributes such as leadership, decision-making, and problem-solving, net of ad length. See Appendix Table \ref{table:judgment-demand}.}

    \item Routine tasks increase objective precision. Standardized procedures and clear performance benchmarks make objective assessments easier to design. \(P_j\) is therefore increasing in routine task intensity (\(r_j\)).\footnote{Mentions of objective screening instruments in the advertisement text (pre-employment tests and assessments, certifications and licenses, background checks, and GPA requirements) rise with routine cognitive intensity, net of ad length. Management advertisements mention these instruments less often than the other three major occupation groups. See Appendix Table \ref{table:screening-instruments}.}
\end{itemize}

\noindent Management roles, for example, tend to have high subjective noise because they demand analytical and interpersonal skills that are hard to evaluate from a r\'{e}sum\'{e}. Routine office and administrative jobs tend to have high objective precision because performance criteria are well defined.\footnote{Routine jobs may also be jobs in which new hires are trained to standard quickly, which could lower the cost of a screening error independently of evaluation precision. The horse race in Section \ref{section:mechanism} conditions on the O*NET job zone, which encodes the preparation a job requires, including on-the-job training. The discretion contrast is slightly larger with the quality controls included (Appendix Table \ref{table:complexity-horserace}).} Appendix \ref{appendix:functional_forms} parameterizes these two task--noise relationships linearly. The results rely only on the sign conditions above.

The central object in the model is the weight the employer places on subjective assessment. The discretion index \(E_j^\ast \in [0,1)\) is that weight when the two signals are combined. It increases in analytical (\(a_j\)) and interpersonal (\(p_j\)) task intensity and decreases in routine (\(r_j\)) intensity, so it rises with a job's reliance on judgment that cannot be verified from a r\'{e}sum\'{e}. Its empirical counterpart is

\begin{equation}
\label{eq:Estar_def}
    E_j^\ast = \frac{\hat{B}_j}{\hat{B}_j + \hat{P}_j}, \qquad E_j^\ast \in [0,1),
\end{equation}

\noindent where \(\hat{B}_j\) and \(\hat{P}_j\) are rescaled, strictly positive composites (min-max plus a constant) of the analytical, interpersonal, and routine task intensities (Appendix \ref{appendix:functional_forms}). Both composites are dimensionless, so the index is a well-defined fraction of the job's evaluation that rests on subjective judgment. As before, strong objective criteria push \(E_j^\ast\) toward zero. When subjective evaluation difficulty is high, the subjective channel dominates and \(E_j^\ast\) approaches one. The model's predictions rest on the task comparative statics, which hold for any positive weights on the task intensities, and the discretion index orders jobs by their reliance on subjective assessment.

\subsubsection{From Signals to Discrimination}

\noindent How do firms combine subjective and objective information? A fully Bayesian employer would weight each signal inversely with its noise and place less weight on the subjective channel when it is unreliable. In practice, employers do not fully adjust for signal quality. Managers in low-skill service firms who overrode a newly adopted job test's recommendations hired workers with shorter tenures \citep{hoffman_discretion_2018}. Employers also read the verifiable credentials applicants present as signals of underlying traits. \citet{stans_impact_2026} find that human-resource managers value a finished graduate degree as evidence of stronger cognitive and non-cognitive traits and read an unfinished one as a weaker signal than a bachelor's degree. When hiring for a management position, employers emphasize subjective assessments of leadership and interpersonal skills because the job requires those skills, not because the assessments are precise. These skills could in principle be assessed with standardized instruments, but at the screening stage employers evaluate r\'{e}sum\'{e}s through subjective judgment rather than administering tests.\footnote{The advertisement text bears this out. Mentions of objective screening instruments do not rise with analytical task intensity, and the point estimate is negative, while they rise with routine cognitive intensity (Appendix Table \ref{table:screening-instruments}).} We therefore model employers as placing weight \(E_j^\ast\) on the subjective signal and weight \(1 - E_j^\ast\) on the objective signal, so that the job's task demands rather than signal precision govern the combination. This behavioral assumption departs from optimal updating. Under optimal updating, an employer would shrink the weight on the noisier subjective signal for minority applicants, and that adjustment would weaken or reverse the prediction that discrimination rises with \(E_j^\ast\). The evidence cited above on incomplete adjustment to signal quality motivates the departure.\footnote{The task-based weighting assumption can be relaxed within the behavioral class. Any signal-combination rule in which the weight on subjective assessment remains positive and increases with the job's reliance on subjective evaluation produces similar predictions. Fully Bayesian precision weighting lies outside this class because it ties the weight to signal quality rather than to the job's task demands. See Appendix \ref{appendix:discrim}.}

Because subjective evaluation is noisier for minority applicants, the composite assessment is also noisier. Let \(V_{gj}\) and \(V_{Mj}\) denote the composite evaluation variance for group-\(g\) minority and majority applicants, respectively. The minority--majority variance gap is
\begin{equation}
\label{eq:var_gap}
    V_{gj} - V_{Mj} = {E_j^\ast}^2 \, B_j \, \Delta_g,
\end{equation}

\noindent where \(\Delta_g = 1/\pi_g - 1 > 0\) is the group-specific noise penalty (Appendix \ref{appendix:discrim}). A smaller \(\pi_g\) (less precise evaluation of group \(g\)) implies a larger \(\Delta_g\). The squared discretion index \({E_j^\ast}^2\) reflects the employer's reliance on subjective evaluation. The subjective noise level \(B_j\) scales the absolute size of the noise differential.

Noisier composite assessments produce lower posterior expected productivity and therefore fewer callbacks, since employers call back applicants who clear a selective bar (Appendix \ref{appendix:callback}). A first-order approximation of the callback function yields equation \eqref{eq:main_result}, with the scale factor \(K_{gj}\) collecting \(E_j^\ast\), \(B_j\), \(\Delta_g\), and a Taylor coefficient that depends on the baseline majority evaluation variance. The approximation isolates \(E_j^\ast\) as the object governing the testable predictions (Appendix \ref{appendix:sufficient_stats}). Because \(K_{gj}\) itself varies with \(E_j^\ast\) and \(B_j\), equation \eqref{eq:main_result} is a monotonicity statement rather than a decomposition with a fixed scale factor, matching the caveats in Appendix \ref{appendix:sufficient_stats}. The variance gap is unambiguously increasing in \(E_j^\ast\). Callback gaps inherit this ordering when jobs are compared at similar levels of majority evaluation variance, which holds approximately constant the rate at which a variance gap translates into a callback gap.

The scale factor \(K_{gj}\) varies across demographic groups through \(\Delta_g\). Groups with less precise subjective evaluation (small \(\pi_g\), large \(\Delta_g\)) face larger callback gaps at every level of \(E_j^\ast\). Groups for which \(\pi_g \approx 1\) face little or no discrimination regardless of task content. The group-specific parameters \(\pi_g\) shift the level of each group's gap but not its cross-job gradient. The model's testable content is the sign of the task comparative statics, which holds for every \(\pi_g < 1\), so the framework is disciplined by its cross-job predictions rather than fit to each group's average gap.

The empirical decomposition in Section \ref{section:mechanism} takes these cross-job predictions to the data, and it does not estimate the two task channels with equal sharpness. The asymmetry is anticipated. Analytical and interpersonal intensity are fuzzier constructs than routine, codifiable content, so the subjective proxy \(\hat{B}_j\) carries more measurement error than \(\hat{P}_j\), which attenuates the coefficient on \(\hat{B}_j\) toward zero. The design therefore identifies the objective-evaluability channel more sharply than the subjective one. We read the compression of gaps where screening can rest on verifiable criteria as the model's robustly testable content.

\subsubsection{Testable Predictions}

\noindent Equation \eqref{eq:main_result} generates a direct prediction about how task bundles affect discrimination. The proposition below establishes that the minority--majority \emph{variance gap} varies systematically with task content. A second prediction follows when we extend the model to allow contact to modify subjective evaluation noise.

\begin{adjustwidth}{0.75cm}{0cm}
\begin{proposition}[Task Bundles and Discrimination]
\label{prop:task_main}
High-\(B_j\), low-\(P_j\) occupations (high analytical and interpersonal content, e.g., management) exhibit larger minority--majority evaluation variance gaps than low-\(B_j\), high-\(P_j\) occupations (routine-intensive, e.g., office and administrative support). Callback gaps inherit this ordering when jobs are compared at similar levels of majority evaluation variance, which holds approximately constant the rate at which a variance gap translates into a callback gap.
\end{proposition}
\end{adjustwidth}
\vspace{0.5cm}

\noindent Appendix \ref{appendix:proof_task} provides the formal derivation.

Proposition \ref{prop:task_main} concerns variation across occupation groups with different task profiles. We now ask how contact affects discrimination within the model. Our contact measure combines interaction with customers and with coworkers (Section \ref{section:tasks}). In the variance gap of equation \eqref{eq:var_gap}, contact can operate through subjective evaluation noise \(B_j\). Roles built around interaction introduce additional dimensions of unstructured assessment (appearance, demeanor, interpersonal manner) that are difficult to evaluate objectively. The question is whether contact raises \(B_j\) uniformly or only where \(B_j\) is already elevated. If contact multiplies the subjective channel rather than shifting it additively, its effect is concentrated in high-\(B_j\) jobs (those with high analytical and interpersonal content), while the objective channel \(P_j\) is unaffected throughout.

A large literature documents higher discrimination in customer-facing positions, typically attributed to customer prejudice or employer beliefs about customer preferences. The model offers a distinct prediction. Under the information-friction mechanism, contact amplifies discrimination only where the subjective channel \(B_j\) is active, and not where evaluation rests on objective precision \(P_j\).

\begin{adjustwidth}{0.75cm}{0cm}
\begin{proposition}[Contact--Non-routine Complementarity]
\label{prop:contact_main}
Suppose contact multiplies subjective evaluation noise \(B_j\) rather than acting as a generic, additive noise shifter. Then contact raises the minority--majority variance gap only in jobs where \(B_j\) is already elevated (high analytical and interpersonal content), and has no effect where evaluation rests on objective precision \(P_j\) (purely routine jobs).
\end{proposition}
\end{adjustwidth}
\vspace{0.5cm}

\noindent The multiplicative structure sharpens the contrast with a pure customer-prejudice account. A model in which contact raises the gap regardless of task content predicts contact effects at every \(B_j\). The information-friction mechanism predicts that contact matters only where \(B_j\) is already elevated, so that subjective evaluation has scope to differ across groups. A customer-prejudice channel operating through those same high-\(B_j\) jobs would generate a similar pattern, so the contrast is suggestive rather than decisive. Appendix \ref{appendix:proof_contact} provides the formal derivation. We examine both predictions in Section \ref{section:results}.

\section{The Experiment} \label{section:experiment}
\subsection{Design}\label{section:design}
  
  \noindent We conducted identical r\'{e}sum\'{e} audits in 2016 and 2017. The audits ran from April through July in each year. Using a large online job search board, we submitted randomly generated r\'{e}sum\'{e}s to job advertisements randomly drawn from a bank of postings compiled by our research team. We monitored employer responses through dedicated email and voice accounts to record callbacks.\footnote{The experiment was reviewed by the Institutional Review Boards at Auburn University and the University of Wisconsin--La Crosse. Both ruled that the experiment did not constitute human subjects research, although we agreed to maintain the anonymity of the audited organizations and the institutions and employers named on the fictitious r\'{e}sum\'{e}s (e.g., universities and firms associated with internship, work, and volunteer experiences).}
  
  The job bank covers six categories, including account executive, banking, customer service, finance, insurance, and marketing. We selected these categories because they generated a large and consistent flow of postings suitable for an audit study of entry-level labor markets. We excluded advertisements requiring extensive prior experience, specialized credentials, or possession of an occupational license.\footnote{A text audit of the archived advertisements flagged about 4 percent of advertisements (406 of 9,220) whose language could be interpreted as an in-hand license requirement. The flag is deliberately over-inclusive: it captures any advertisement titled \enquote{licensed}, a current license described as a plus, or prospective requirements sponsored by the employer after hire (a Series 7, for example, which a new graduate cannot hold without a sponsoring firm). The mean callback rate for these advertisements is roughly half the sample mean. Any screen they embed operates at the advertisement level and applies to all applicants equally. As a conservative check, excluding them leaves the screening results in Section \ref{section:screening} slightly stronger (Table \ref{table:screening-robustness}).} Because the fictitious applicants represent new college graduates, the r\'{e}sum\'{e}s primarily reflect internships and part-time employment during college. Restricting the job bank to positions compatible with this profile avoids design complications arising from licensing requirements and other specialized qualifications that vary across states.
  
  Unlike many audit studies that focus on specific cities or local labor markets \citep{bertrand_are_2004,kroft_duration_2013,lahey_age_2008,nunley_racial_2015}, we studied the national labor market and imposed no geographic restrictions on job advertisements. Research assistants submitted applications to randomly selected openings. We submitted four randomly generated r\'{e}sum\'{e}s to each advertisement, yielding 36,880 applications to 9,220 unique job postings.\footnote{We initially applied to 9,468 job postings but could not assign detailed occupation codes to 248 of them. We excluded these advertisements because our empirical framework relies on occupation-level task measures. They account for less than 3 percent of the sample, and including them does not materially affect our baseline estimates.}
  
  We randomly assigned r\'{e}sum\'{e} characteristics using the program developed by \citet{lahey_computerizing_2009}. Each r\'{e}sum\'{e} listed a name, address, university, major, and work experience obtained during college. Following the audit literature \citep[e.g.,][]{bertrand_are_2004}, we used names designed to signal applicants' race/ethnicity and gender to prospective employers. The names are Colin Schneider and Jack Schwartz (White men), Darius Jackson and Xavier Washington (Black men), Diego Martinez and Andres Flores (Hispanic men), Claire Haas and Madeline Krueger (White women), Kiara Banks and Jasmin Booker (Black women), and Adriana Hernandez and Gabriela Lopez (Hispanic women).\footnote{The \enquote{White} and \enquote{Black} first names come from \citet{levitt_freakonomics_2005}. We identified Hispanic names through internet sources. We then used the Social Security name database to select names with similar popularity among individuals born between 1994 and 1996, the cohorts most likely to enter the labor market in 2016--2017. The first-name popularity rankings were Colin (128th), Jack (109th), Darius (169th), Xavier (133rd), Diego (208th), Andres (165th), Claire (131st), Madeline (79th), Kiara (188th), Jasmin (182nd), Adriana (147th), and Gabriela (122nd). Surnames were selected using the 2000 Census surname distribution. Among individuals with the \enquote{White} surnames, approximately 96--97 percent  report their race/ethnicity as non-Hispanic White, while the comparable range for the Hispanic surnames is 91--94 percent. Among individuals with the \enquote{Black} surnames, the share reporting Black/African American ranges from 54 to 90 percent. Although these shares are lower than the corresponding shares associated with White and Hispanic surnames, the selected names are prevalent among Black respondents in the Census data. Survey-based perception data also support these choices. \citet{gaddis_how_2017} tests the four Black surnames used here directly and finds that pairing a Black first name with a Black surname increases the rate at which the name is perceived as Black from 75.0 to 82.5 percent on average. In his data, Darius is perceived as Black at high rates and Claire as White at very high rates. Jasmin and Kiara are perceived as Black less often from the first name alone. Pairing them with the strongly Black-associated surnames Booker and Banks offsets this attenuation. \citet{gaddis_racialethnic_2017} finds that Hispanic first names paired with Hispanic surnames, the combination used here, are recognized as Hispanic about 90 percent of the time. Any attenuation in a group's name signal biases its estimated callback gap toward zero. As such, cross-group comparisons are conservative for groups with weaker signals.}
  
  Our design does not constrain the combinations of racial/ethnic and gender-specific names submitted to a given job advertisement. Consequently, the demographic composition of applications varies across postings. For example, among the 9,220 job openings, some received no applications with Black-sounding names, while others received one or more such applications, with the gender composition also varying across postings. The distribution of these combinations is consistent with random assignment \citep{arellano-bover_unbundling_2026}. Patterns for White- and Hispanic-sounding names are similar. Because four r\'{e}sum\'{e}s were submitted to each advertisement, a potential concern is that employer responses depend on the demographic composition of co-applicants rather than solely on individual attributes. We test for within-ad interference in Section \ref{section:test_discrim}.
  
  We used twelve flagship public universities spanning the continental United States, with all broad Census regions represented. Each r\'{e}sum\'{e} was randomly assigned a major (one of eight fields, each with equal probability), up to one minor (history or mathematics), and a GPA (including a 25 percent probability of no listed GPA). Applicants could also be assigned an internship experience drawn from one of two categories. Quantitative internships emphasize data and research tasks (e.g., \enquote{Marketing Analyst Intern}, \enquote{Research Intern}), while social internships emphasize sales and client interaction (e.g., \enquote{Marketing Sales}, \enquote{General Sales}). In total, we constructed 60 internship experiences split evenly across the two categories. Additional credentials were randomized independently and included volunteer experience, Spanish language proficiency, study abroad, college work experience, and computer skills. Computer skills were none (25 percent), basic skills (25 percent), data analysis (25 percent), programming (12.5 percent), or both data analysis and programming (12.5 percent). Computer skills, internship type, and study abroad enter directly into the credential attenuation analysis in Section \ref{section:credentials}. A companion paper, \citet{arellano-bover_unbundling_2026}, uses the same audit to estimate the returns to r\'{e}sum\'{e} characteristics for new college graduates. That paper documents the experimental design, balance and randomization checks, and callback returns to all r\'{e}sum\'{e} attributes. Rather than duplicate those details, we focus on how discrimination varies with the task content of jobs.\footnote{Online Appendix Table O1 reports the full set of r\'{e}sum\'{e} characteristics, their assignment probabilities, and corresponding sample means.}

\subsection{Template and Detection Bias} \label{section:avoid-bias}
 
  \noindent Audit studies face two standard concerns: \enquote{template bias} and detection by audited firms.\footnote{See \citet{lahey_technical_2018} for a detailed discussion of template bias and other considerations in the design and implementation of an audit study.} Template bias arises when r\'{e}sum\'{e} characteristics are systematically linked, making it difficult to isolate the effect of individual attributes. To prevent this, we randomly assigned each r\'{e}sum\'{e} characteristic independently across applications. The sole exception was the joint assignment of university and residential address. This linkage does not affect identification because no two r\'{e}sum\'{e}s listing the same university were submitted to the same advertisement.
  
  Detection by audited firms cannot be fully ruled out and is unobserved by the researcher. However, evidence suggests that such detection attenuates measured discrimination. \citet{balfe_infrequent_2023} show that repeated exposure to infrequent identity signals can alter employer behavior and thereby reduce observed bias. Because we submitted four r\'{e}sum\'{e}s to each job opening, employers could potentially detect the audit. We took five steps to reduce this possibility.

  \begin{enumerate}
    \item Most firms were audited only once. Approximately 71 percent of firms in our sample received applications for a single job posting, with four r\'{e}sum\'{e}s submitted to each advertisement.
    \item Applications were staggered so that no advertisement received more than one r\'{e}sum\'{e} per day. Research assistants applied to each job over a 4--7-day period.
    \item Each application submitted to a job listed a different university. This reduced the likelihood that firms detected similarities across r\'{e}sum\'{e}s when applicants shared other characteristics, such as majors.
    \item Restricting the design to six job categories allowed us to construct realistic r\'{e}sum\'{e}s. Focusing on new college graduates limited work histories to short and plausible experiences.\footnote{Audit studies with varied and, in some cases, lengthy work histories include, for example, those focused on unemployment spells \citep[e.g.,][]{farber_determinants_2016, kroft_duration_2013, nunley_effects_2017} and age discrimination \citep{lahey_age_2008,neumark_is_2019}.}
    \item We excluded advertisements requiring company-specific application forms. Open-ended application questions are difficult to standardize across applicants and may increase the likelihood of detection.
  \end{enumerate}

  \noindent As a further check, we re-estimated the specifications using firm rather than advertisement fixed effects, because detection is more likely to occur across applications received by the same firm (see Section \ref{section:results-tasks}). Approximately 29 percent of firms appeared multiple times in the sample. The results are robust to this alternative specification.
    
\subsection{Advertisement Classification and Worker Tasks}\label{section:tasks}

\subsubsection{Advertisements} \label{section:ad_char}
  
  \noindent In total, we audited 4,969 firms. Most firms (71 percent) were audited once, and 95 percent were audited four or fewer times.

  Of the 9,220 job openings in our sample, firm names were identified in 96 percent of advertisements, and 99 percent include location information. Geographic coverage was uneven. Thirty-nine percent of job openings were located in five states (California, Texas, Florida, Illinois, and New York), but the sample covered all 50 states, with at least 40 job openings in 35 states. Figure \ref{fig:app-maps} presents a heat map of advertisements across commuting zones.\footnote{After obtaining the latitude and longitude of firm locations, we identified the corresponding county and used a crosswalk from the U.S. Department of Agriculture (USDA) to group counties into 1990 commuting zones.}
    
  We submitted the job title and description to the O*NET-SOC Autocoder, a machine learning algorithm developed by the Department of Labor and refined by R.M. Wilson Consulting, Inc. The algorithm assigned each advertisement an 8-digit O*NET-SOC code and produced an occupation-match score, ranging from 50 to 99, with higher values indicating a stronger match. In sensitivity analyses, we weight observations by these scores to assess sensitivity to occupational misclassification. Weighted and unweighted estimates are similar.

  The sample is unevenly distributed across major occupation groups. Approximately 94 percent of advertisements fall into four categories. These are management (10 percent, 21 detailed occupations),\footnote{Throughout, \enquote{management} refers to occupations that the O*NET-SOC Autocoder places in the major management group. In our entry-level sample, these consist of managerial and professional roles within the six sampled fields rather than executive-track positions.} business and financial operations (21 percent, 27 detailed occupations), sales (36 percent, 16 detailed occupations), and office and administrative support (27 percent, 35 detailed occupations). Altogether, the sample covers 175 detailed occupations.

\subsubsection{Task Intensities} \label{section:worker_tasks}

  \noindent We use the task taxonomy from \citet{acemoglu_skills_2011}, which classifies occupational content into analytical, interpersonal, routine cognitive, routine manual, and physical task dimensions. These measures are constructed from survey responses to items in the O*NET Abilities, Work Activities, and Work Context modules.

  \begin{itemize}
    \item \emph{Analytical} captures analyzing data, creative thinking, and interpreting information.
    \item \emph{Interpersonal} captures establishing relationships, guiding, motivating, and coaching others.
    \item \emph{Routine Cognitive} captures repetitive tasks, maintaining accuracy, and performing structured activities.
    \item \emph{Routine Manual} captures repetitive motions, machine-paced work, and process control.
    \item \emph{Physical} captures operating machinery, manual dexterity, spatial orientation, and object handling.
  \end{itemize}

  \noindent Following \citet{autor_skill_2003} and \citet{deming_growing_2017}, we rescale the underlying ordinal measures to a 0--10 range and average them to construct composite task indices. In the theoretical framework, analytical and interpersonal intensity serve as proxies for subjective evaluation noise ($B_j$), while routine cognitive intensity serves as a proxy for objective evaluation precision ($P_j$). Routine manual and physical task intensities enter as controls in the empirical specifications.

  We also incorporate a contact-task measure from \citet{hurst_task-based_2024}. This measure captures the importance of interactions with coworkers and customers, a central focus of the discrimination literature \citep[e.g.,][]{combes_customer_2016,giuliano_manager_2009,laouenan_hate_2017,nunley_racial_2015,hedegaard_price_2018}. It is constructed from two O*NET survey items measuring coworker interaction and customer engagement. In Proposition \ref{prop:contact_main}, contact scales subjective noise ($B_j$).
  
  We use the assigned O*NET-SOC codes to link each advertisement to O*NET and the ACS.\footnote{This process requires harmonizing O*NET-SOC codes with the ACS \emph{occsoc} variable using crosswalks. ACS codes for 2015--2017 are based on the 2010 SOC system, while the 2018 data use the 2018 SOC system. We applied a 2010--2019 O*NET-SOC crosswalk and then mapped the 2019 O*NET-SOC codes to 2018 SOC codes. In cases where the ACS aggregates occupations (e.g., postsecondary teachers coded as 25-1000 rather than the 25-1011 to 25-1199 classifications in the 2018 SOC system), we recoded them to align with ACS classifications.} We then construct an occupation-year panel of occupations for 2015--2018, following \citet{ross_routine-biased_2017, ross_effect_2021} and \citet{cole_distribution_2022}, and compute employment-weighted task intensity measures at the detailed occupation level.\footnote{For employed, college-educated individuals (ages 21--26), we compute employment weights following \citet{deming_growing_2017}. These weights are the product of person weight (\emph{perwt}), usual hours worked (\emph{uhrswork}), and weeks worked (\emph{wkswork2}). Because \emph{wkswork2} is reported in intervals, we use midpoint values (e.g., 50--52 weeks $\rightarrow$ 51 weeks).} Figure \ref{fig:task-distributions} presents kernel density estimates for the six task-intensity measures. Solid lines show young, college-educated workers in occupations linked to the audit sample. Dashed lines show the same group across all ACS occupations.

  Relative to the overall labor market, the audit sample is concentrated in occupations with lower routine manual and physical task intensities and slightly higher analytical, interpersonal, routine cognitive, and contact task intensities. These differences reflect the six job categories targeted in the audit.

  A potential concern with occupation-level task measures is that they assign the same task intensities to all job postings within a detailed occupation, even though task requirements may vary across postings. We address this concern in two ways. First, our classification distinguishes among 175 detailed occupations. Second, we complement the occupation-level analysis with advertisement-level text measures constructed from posting language that allow task emphasis to vary within detailed occupations. The advertisement-level results replicate the directional patterns from the occupation-level analysis, indicating that the findings do not depend solely on occupation-average task measures (see Section 1 of the Online Appendix).

  Table \ref{table:occlist-new} reports the most prevalent detailed occupations in the sample, organized by major occupation group (Panels A--E). For each occupation, we report the number of unique advertisements (column 1), its share of the sample (column 2), and percentile rankings for each of the six task intensities (columns 3--8). The breadth of occupations provides substantial variation in task intensity both within and across major occupation groups (2-digit SOC).

  Management occupations rank high in both analytical and interpersonal intensity, while sales and office and administrative support occupations exhibit high contact intensity but lower analytical and interpersonal intensity. More broadly, routine cognitive intensity is negatively associated with analytical and interpersonal intensity across occupations.

  Task correlations are often strong. Analytical and interpersonal tasks are positively correlated $\left( \hat{\rho} = 0.7 \right)$, while analytical tasks are negatively associated with routine manual and physical tasks at approximately $\hat{\rho} = -0.5$. Interpersonal tasks show similar negative correlations with routine manual and physical tasks. The correlation between interpersonal and contact intensity is small $\left( \hat{\rho} = 0.06 \right)$, indicating that the two measures capture distinct dimensions of job content. Routine cognitive and routine manual tasks are positively correlated $\left( \hat{\rho} = 0.4 \right)$, while routine manual and physical tasks are positively correlated $\left( \hat{\rho} = 0.8 \right)$.

  These correlations complicate the interpretation of task-specific estimates. For example, occupations with high analytical intensity also tend to require strong interpersonal skills. To assess overall collinearity, we compute the determinant of the task correlation matrix. Values close to one indicate weak dependence, while values near zero indicate strong collinearity. For the 175 occupations in the audit sample, the determinant is 0.053. Because interpreting individual task coefficients without accounting for this interdependence could be misleading, we focus on task bundles and the composite discretion index rather than individual task coefficients in Section \ref{section:results}.
  
\subsection{Employer Responses} \label{callbacks}

  \noindent Our outcome was a callback, defined as an employer response expressing interest in the applicant, including an interview invitation or a request for additional information (e.g., the applicant's availability to discuss the position). The overall callback rate is approximately 15 percent. Callback rates range from 5 percent in business and financial operations to 25 percent in sales, with intermediate rates in management (12 percent) and office and administrative support (10 percent).

  Across task dimensions, callback rates are similar for occupations above and below the median of analytical and interpersonal intensity. The corresponding differences are larger along the remaining task dimensions. Callback rates above versus below the median are 8 versus 22 percent for routine cognitive intensity, 10 versus 20 percent for routine manual intensity, 13 versus 18 percent for physical intensity, and 17 versus 12 percent for contact intensity. The absence of meaningful unconditional differences in analytical and interpersonal task intensity, despite the importance of these dimensions for the discrimination results below, illustrates that callback levels and callback gaps need not move together.

  A limitation of audit studies is that callbacks capture only the initial screening stage. Discrimination at this stage does not reveal the extent of discrimination later in the hiring process. \citet{jarosch_statistical_2019} show, in the context of unemployment-duration discrimination, that large callback gaps in audit studies need not translate one for one into differences in job-finding rates, because a foregone interview affects employment only if it would have led to an offer. Their extended model implies that downstream effects are larger when missed interviews interact with human-capital decay or congestion among applicants. Other evidence suggests the possibility of larger downstream gaps. \citet{quillian_evidence_2020} document larger majority--minority differences at the job-offer stage than at the callback stage, while \citet{lanning_opportunities_2013} embeds hiring discrimination in a search model calibrated to audit data and shows that such discrimination can generate wage gaps. The net effect of callback discrimination on employment outcomes therefore depends on how screening gaps propagate through the post-callback hiring process, which our design does not observe. We therefore interpret our estimates as evidence on where callback discrimination concentrates in the task space rather than as direct measures of its downstream employment consequences.

\subsection{Sample Representativeness}

  \noindent The job-category and credential restrictions in our design exclude many positions in fields such as medicine, law, engineering, and accounting. Because these fields employ a large share of college graduates, their exclusion limits the generalizability of our findings. Because we link each advertisement to a detailed occupation code, we can measure how much of the early-career labor market the sample covers. The 9,220 advertisements in our sample map to 175 detailed occupations in the ACS-O*NET data, representing roughly one-third of all detailed occupations.

  Using ACS data for 2015--2018, we assess the audit sample's coverage of early-career employment. Among college-educated individuals ages 21--26, 38 percent are employed when all occupations are included. This share rises to 46 percent for occupations represented by the audit sample, compared with 33 percent for occupations not included. Overall, approximately two-thirds of employed individuals in this age group work in occupations represented in the audit. Among those with majors included in our experiment, this share increases to 70 percent. These patterns indicate that the audit sample captures a substantial share of the early-career labor market for college graduates, particularly within the fields represented in the design.

  Because the design focuses on new college graduates, it limits variation arising from prior work histories and provides a clear view of the entry-level labor market for degree holders. This makes the data well suited to studying task-based discrimination, though less suited to firm-level analysis of the type conducted by \citet{kline_systemic_2022}. Our study audits 9,220 job advertisements from 4,969 firms, whereas \citet{kline_systemic_2022} audit 11,114 advertisements from 108 firms. Differences in applicant characteristics, firm composition, and job types further distinguish the two designs.

  \citet{kline_systemic_2022} also examine how callback gaps vary with job task content. In their bivariate comparisons, racial callback gaps are larger in jobs requiring customer interaction and manual skills. Conditional on firm fixed effects, however, the relationship between racial gaps and task content is jointly insignificant, and within-industry task gradients are weak. In their large-employer sample, the customer-interaction pattern operates mainly across industries, which they attribute to firm and sector recruiting cultures rather than to a job-level task mechanism. The contrast with our findings reflects differences in research design. Their sample is designed to detect firm-level discrimination, so firm fixed effects absorb much of the cross-job task variation retained in our broader sample of firms. The two sets of results are therefore complementary. \citet{kline_systemic_2022} show that discrimination concentrates across firms and sectors, while our within-ad design shows that, across the broad entry-level market, discrimination also concentrates in analytically and interpersonally intensive jobs with low routine content.\footnote{See Figure III and Section VIII.A in \citet{kline_systemic_2022}.}

\subsection{Pre-Registration, Pre-Analysis Plans, and Power Calculations} \label{prereg-preanal}

  \noindent We designed the experiment in 2014, before pre-analysis plans (PAPs) became widely used in the AEA RCT Registry (Online Appendix Figure O1). We registered the experiment on the AEA RCT Registry (\href{https://www.socialscienceregistry.org/trials/12914}{AEARCTR-0012914}) but did not submit a PAP. Three features of the design limit the scope for data mining. First, the experiment has a single outcome variable, employer callbacks. Second, the task taxonomy is standard in labor economics \citep{acemoglu_skills_2011} and was not developed by the research team. Third, we classify job advertisements into detailed occupations using a method established in prior work \citep{jaeger_demand_2023}. \citet{banerjee_praise_2020} argue that many of the principal benefits of pre-analysis planning can be achieved through the registration fields in the AEA RCT Registry, and that requiring detailed pre-specification can impede knowledge creation in real-world field experiments. Moreover, exhaustively specifying every hypothesis \emph{ex ante} is impractical for an audit design that randomizes many applicant attributes.

  Following the power-analysis protocol in \citet{lahey_technical_2018}, we used \textit{G\textsuperscript{*}Power} to compute the minimum number of applications required to detect the planned race/ethnicity--gender and r\'{e}sum\'{e}-credential effects \citep[see also][]{arellano-bover_unbundling_2026}. Because callbacks to r\'{e}sum\'{e}s submitted to the same vacancy are correlated, we adjusted these calculations for clustering using the standard design effect,
  
  \[
    \text{DE} \;=\; 1 + (k - 1)\times \text{ICC},
  \]
  
  \noindent where $k=4$ is the number of r\'{e}sum\'{e}s submitted per job advertisement and $\text{ICC}$ is the intraclass correlation coefficient, defined as the share of total callback variance attributable to differences between vacancies rather than to differences among applicants within the same vacancy. Higher values of the ICC imply stronger within-vacancy correlation. Failing to make this adjustment would overstate statistical power.

  Given the absence of pilot data, we follow the conservative upper-bound guidance in \citet{lahey_technical_2018} and set $\text{ICC} = 0.30$. Even under this assumption, the adjusted calculation requires only $10{,}612$ observations, whereas the experiment includes $36{,}880$ applications, approximately three and a half times the calculated minimum. The sample size therefore provides more than the minimum estimated power for detecting the planned group-level differences in callback rates.\footnote{To verify this, we conducted a Monte Carlo simulation using a linear probability model. The simulation included six demographic groups and six job categories, with a baseline callback rate of 0.150 in every group-by-category cell except one, for which a single target cell was set to 0.225. Across 1,000 replications, a joint F-test of all group-by-job-category interaction terms rejected the null of no heterogeneity in 92.0 percent of samples. This omnibus figure does not measure the power of the single-coefficient tests reported in the callback tables. The simulation was not designed to assess the credential triple interactions in Section \ref{section:credentials}, and we make no ex ante power claim for those tests.}

\section{Econometric Methodology} \label{section:test_discrim}
      
    \noindent We use linear probability models to estimate callback discrimination overall and its heterogeneity across occupation groups and task environments. Our baseline specification is

        \begin{equation}
            callback_{i,j} = \alpha + \beta D_{i} + X_{i}'\Theta + \phi_{j} + \epsilon_{i,j},
            \label{eq_baseline}
        \end{equation}

    \noindent where subscripts $i$ and $j$ index applicants and job advertisements, respectively. The outcome variable, $callback_{i,j}$, is an indicator equal to one if an employer responds positively to an application and zero otherwise. $D_{i}$ denotes a vector of indicator variables for race/ethnicity--gender categories, specifically White women, Black men, Black women, Hispanic men, and Hispanic women, with White men as the reference category. $X_{i}$ contains all r\'{e}sum\'{e} characteristics randomly assigned to applicants. $\phi_{j}$ denotes advertisement fixed effects that control for observed and unobserved advertisement-specific factors common to all applicants, such as local labor market conditions and overall applicant volume. $\epsilon_{i,j}$ captures residual determinants of callbacks. Random assignment ensures that $D_i$ and $X_i$ are orthogonal to $\epsilon_{i,j}$, so the coefficients in $\beta$ and $\Theta$ have a causal interpretation.\footnote{We assess the independence of assignment using correlation coefficients across r\'{e}sum\'{e} attributes. Table \ref{table:corr-coeff} reports the correlations between race/ethnicity indicators, gender indicators, and other characteristics. All coefficients are smaller than 0.1 in absolute value, consistent with independent assignment.} 
        
    To analyze how discrimination varies across occupational contexts and task requirements, we extend equation \eqref{eq_baseline} to include interactions between the demographic indicators, $D_{i}$, and occupation-specific characteristics.

        \begin{equation}
            callback_{i,j} = \alpha + \beta D_{i} + \gamma D_{i} \times G_{o(j)} + X_{i}'\Theta + \phi_{j} + \epsilon_{i,j},
            \label{eq_interaction}
        \end{equation}

    \noindent where $o(j)$ denotes the occupation assigned to each job advertisement, and $G_{o(j)}$ is a set of indicators for the relevant occupation or task categories. We characterize the job environment in four ways. First, we use four major occupation groups (management, business and financial operations, sales, and office and administrative support), which cover approximately 94 percent of our audit sample. Second, we apply K-means clustering to occupation-level task measures to form clusters of advertisements that differ in task content. Third, we construct the discretion index $E_j^\ast$, which captures how far hiring rests on subjective rather than objective evaluation. Fourth, we use this index and its underlying components to examine how subjective evaluation and objective precision differentially shape callback gaps. The parameters of interest are $\beta$, $\beta + \gamma$, and $\gamma$. For a given demographic group, $\beta$ measures the callback gap in the omitted job category, $\beta+\gamma$ measures the gap in the comparison category, and $\gamma$ measures the difference between them.\footnote{When analyzing discrimination by major occupation groups, $o(j)$ corresponds to the 2-digit Standard Occupational Classification (SOC) code linked to the job advertisement. In analyses focusing on task measures, $o(j)$ corresponds to the detailed 6-digit SOC occupation code linked to each advertisement.}
   
    Random assignment of r\'{e}sum\'{e} characteristics ensures that applicant attributes are orthogonal to job characteristics, so the baseline callback gaps in equation \eqref{eq_baseline} have a causal interpretation. Within each occupation or task environment, the interaction specification continues to compare applicants whose demographic attributes were randomly assigned. The resulting group-specific callback gaps therefore have a causal interpretation within each environment. Task content is a structural feature of the job determined by market demand, not a treatment that can be experimentally assigned. The relationship between these causal callback gaps and task content is therefore observational. The design embeds randomized applicant attributes within the task structure generated by the labor market. This is analogous to estimating treatment-effect heterogeneity across pre-treatment characteristics without randomly assigning those characteristics.

    The standard concern with treatment-effect heterogeneity is that the moderating variable may correlate with omitted factors. Jobs that differ in task content may also differ in pay, firm type, applicant pool composition, screening technology, and selectivity. Job-advertisement fixed effects absorb all such characteristics when estimating within-ad callback gaps, but they do not make the cross-job relationship between callback gaps and task content causal.\footnote{Using within-firm variation in task content instead would require firms with multiple job advertisements. Approximately 71 percent of firms in the sample posted only one advertisement during the study period, and among firms with multiple advertisements, 23 percent posted jobs in the same occupation, contributing no task variation. Job-ad fixed effects draw on both within-firm and between-firm task variation, providing more identifying variation for the task-based specifications.} Three features of the evidence help discipline the interpretation of this relationship. First, the gradient is robust across multiple alternative classifications of task content, including major occupation groups, K-means clusters with varying $k$, one-way through three-way task interactions, and advertisement-level text measures. Second, the credential attenuation test in Section \ref{section:credentials} uses within-job randomization of r\'{e}sum\'{e} credentials. Moreover, the White-male callback rate is identical (0.149) in both halves of the discretion split, so its interpretation is not driven by cross-job differences in baseline callback rates. Third, the sign pattern in the mechanism decomposition, in which subjective noise widens gaps while objective precision compresses them, is inconsistent with a generic account in which job complexity uniformly increases discrimination.\footnote{As a robustness check, we examine whether discrimination varies with the race/ethnicity--gender composition of occupations. We find no evidence that discrimination depends on the relative employment shares of the different demographic groups.}

    Following the r\'{e}sum\'{e} audit literature \citep[e.g.,][]{kline_systemic_2022}, we do not attempt to distinguish between taste-based and information-based sources of discrimination. As \citet{neumark_detecting_2012} emphasizes, audit studies identify differences in employer responses while holding observable qualifications constant, capturing discrimination as defined in legal and empirical contexts. Our estimates therefore reflect the combined effects of level-based and variance-based discrimination.

    This limitation concerns the source of discrimination within a given job. The framework in Section \ref{section:theory} instead asks where discrimination concentrates across jobs. It generates three qualitative predictions that distinguish evaluative discretion from competing explanations for the cross-job pattern. First, subjective noise and objective precision should have opposite-signed effects on callback gaps. A generic screening-difficulty account, in which task complexity of either type widens discrimination, would predict that both components widen gaps. Second, contact should amplify discrimination only in non-routine jobs. A pure customer-prejudice account would predict that contact matters regardless of task content. Third, informative r\'{e}sum\'{e} credentials should narrow non-White-male callback gaps more in low-discretion than in high-discretion jobs. These pieces are not equally decisive. The randomized credential test carries more evidential weight than the other two because it varies verifiable information within the job advertisement. The cross-occupation task and screening-text patterns are descriptive, and we interpret them as corroborating the credential result.

    Because four r\'{e}sum\'{e}s were submitted to each job advertisement, a potential concern is within-ad interference. If employers evaluate applicants relative to one another, callback probabilities may depend not only on an applicant's own characteristics but also on the demographic composition of the other experimental applicants. To assess this, we test whether the callback for r\'{e}sum\'{e} $i$ depends on the demographic composition of the other applicants submitted to the same advertisement. We regress callbacks on the applicant's race--gender indicators interacted with the number of non-White-male or Black co-applicants. The joint $F$-test of all five group-specific peer-composition interactions does not reject the null that the interactions are jointly zero ($p = 0.814$ for non-White-male peers and $p = 0.598$ for Black peers). We further test whether peer composition moderates the discretion gradient by interacting the non-White-male $\times$ High-$E_j^\ast$ term with measures of applicant-pool composition. The resulting triple interaction is small and statistically insignificant ($p = 0.846$). These results provide no evidence that employer responses depend on the composition of the experimental applicant set, supporting the assumption that applications are evaluated independently.

\section{Results} \label{section:results}

\subsection{Discrimination Across Major Occupation Groups} \label{section:results-tasks}

  \noindent Table \ref{table:discrim-occ} reports estimated callback gaps between each demographic group and White men for the full sample (column 1) and separately for the four major occupation groups that account for 94 percent of the audit sample (columns 2 through 5).\footnote{Advertisements linked to occupations outside these four groups are classified as \enquote{other} and included in the regressions for columns 2 through 5, but the corresponding estimates are not reported. This residual category comprises 76 detailed occupations spanning 17 major occupation groups, resulting in a sample too small for precise inference.} Column 1 presents estimates from equation \eqref{eq_baseline}. Columns 2 through 5 present estimates from equation \eqref{eq_interaction}, with management as the omitted category.

  In the full sample, Black applicants face statistically significant callback gaps relative to White men of 2.1 percentage points for Black men and 1.4 percentage points for Black women. Relative to a White-male baseline callback rate of approximately 15 percent, the Black male gap represents a 13 percent proportional reduction in callback probability (Appendix Table \ref{table:occgroup-poisson}). Discrimination varies substantially across major occupation groups. In management occupations (column 2), four of the five non-reference groups face significant callback gaps ranging from 3.3 to 5.1 percentage points. The Black male gap in management (5.1 percentage points, or 28 percent relative to White men) is roughly three times as large as the corresponding gap for office and administrative support roles (1.7 percentage points).\footnote{Because White-male callback rates vary across task groupings, percentage-point gaps are not directly comparable across environments with different baseline callback rates. The routine-cognitive split has a pronounced base-rate asymmetry (0.222 versus 0.076 for White men). The discretion median split that drives the credential attenuation analysis in Section \ref{section:credentials} has identical White-male callback rates in both halves (0.149 versus 0.149), so the credential results are unaffected by base-rate scaling. Appendix Table \ref{table:wm-baserates} reports White-male callback rates for each grouping.} Outside management, estimated gaps are smaller and less pervasive, with statistically significant estimates confined to a subset of demographic groups in each category. A joint $F$-test on the interaction terms $\left(D_{i} \times G_{o(j)}\right)$ rejects the null of equal callback gaps across occupation groups at the 1 percent level.\footnote{Online Appendix Table O2 replaces job-ad fixed effects with firm fixed effects. The baseline discrimination estimates are largely unchanged.} The concentration in management is evident in proportional as well as percentage-point terms. In a Poisson specification, in which each gap is expressed as a callback ratio, callback ratios among the four groups with significant management gaps range from 0.72 for Black men to 0.82 for White women. A joint test rejects equality of proportional callback gaps across occupation groups at the 5 percent level (Appendix Table \ref{table:occgroup-poisson}).

  These gaps do not appear to reflect occupation-specific returns to other r\'{e}sum\'{e} characteristics. When the r\'{e}sum\'{e} characteristics are interacted with major occupation group, the callback gaps in Table \ref{table:discrim-occ} are largely unchanged (Appendix Table \ref{table:robust-gxx}). Hispanic women show no statistically significant gap in the baseline job-ad fixed-effects specification (Table \ref{table:discrim-occ}), a pattern consistent with a near-zero group-specific noise differential. Isolated significant estimates for this group appear in some robustness and subsample specifications, in both directions. We therefore interpret the baseline null as the central tendency rather than evidence that a gap is uniformly absent across specifications and subsamples.\footnote{One candidate explanation is a weak ethnicity signal from the names, which would attenuate attention-based discrimination. The available perception data do not directly support this reading, since Hispanic first names paired with Hispanic surnames, the combination used here, are recognized as Hispanic about 90 percent of the time \citep{gaddis_racialethnic_2017}. A group-specific noise differential near zero, the model's $\pi_g \approx 1$ case, fits the pattern without relying on signal strength.}

  The bottom panel of Table \ref{table:discrim-occ} reports task-intensity percentiles by major occupation group. Management occupations rank high in interpersonal intensity (mean percentile of 91) and low in routine cognitive intensity (mean percentile of 23), compared with business and financial operations (interpersonal mean of 66, routine cognitive mean of 55). These differences in task content motivate the task-based groupings examined next.

\subsection{Discrimination Across Task Profiles}\label{section:results-interactions}

  \noindent Task intensities vary within major occupation groups. Table \ref{table:reg-task-discrim-kmeans4} groups advertisements into four clusters using K-means applied to occupation-level task-intensity percentiles.\footnote{The cluster assignment is subject to uncertainty that the standard errors in Table \ref{table:reg-task-discrim-kmeans4} do not reflect. Re-running the K-means algorithm from 200 random starting values produces partitions whose adjusted Rand index relative to the baseline partition has a median of 0.57, indicating that the four-cluster solution is one of several comparable groupings of the task data. Selecting the partition that minimizes the within-cluster sum of squares across 500 random starts reproduces the concentration of callback gaps (Appendix Table \ref{table:reg-task-discrim-kmeans4-beststart}). Four of the five non-reference groups face significant gaps in the high-discretion cluster, and the Black male gap is wider there and narrower in the low-discretion cluster than in the other clusters. Partitions within 0.1 percent of the minimum still differ from one another (mean adjusted Rand index 0.73), so no starting rule yields a unique grouping. We therefore read the cluster results as a descriptive representation of the gradient. We base statistical inference on the continuous task indices (Table \ref{table:mechanism-contact}) and on direct task-interaction specifications (Online Appendix Tables O7--O12), neither of which depends on a cluster assignment.} The bottom panel reports each cluster's composition and task profile. The clustering distinguishes between jobs that share high analytical demands but differ in routine cognitive content. This contrast allows us to examine whether greater objective evaluation capacity is associated with narrower callback gaps.

  Clusters 1 and 2 isolate this variation. Both rank high in analytical intensity (mean percentiles 90 and 88), but differ substantially in routine cognitive content (mean percentiles 26 and 68). Discrimination is concentrated in Cluster 1, where four of the five non-reference groups face significant callback gaps ranging from 2.3 to 3.6 percentage points. The Black male gap of 3.6 percentage points is equivalent to 42 percent of the White-male callback rate of 8.6 percent in this cluster. Cluster 1 is not composed primarily of management occupations. Business and financial operations advertisements account for 61 percent of its postings. The concentration therefore aligns more closely with the task profile than with the major occupation label alone. In Cluster 2, no group shows a statistically significant gap despite similarly high analytical demands. In Clusters 3 and 4, which combine moderate analytical intensity with higher routine or contact content, callback gaps are smaller and less broadly distributed. Black men face significant gaps in both clusters, Black women face a significant gap in Cluster 4, and White women face a marginally significant gap in Cluster 3.\footnote{Alternative specifications using three, five, six, and seven clusters appear in Online Appendix Tables O3--O6. To assess whether these patterns reflect specific task dimensions rather than clustering choices, Tables O7--O9 present direct task-interaction specifications. We also assess sensitivity to alternative cut-points (40th and 60th percentiles, Tables O10--O12). Callback gaps are largest in jobs with high analytical and interpersonal intensity and smaller when routine cognitive intensity is high. Three-way interaction estimates indicate that callback gaps are concentrated in jobs combining high analytical and interpersonal intensity with either low routine content or high contact intensity, whereas no comparable gaps emerge in jobs low in both analytical and interpersonal intensity. We report joint $F$-tests of the interaction terms in each specification.}

  These results are robust to alternative measures of task content. Online Appendix Table O13 reports results using quintiles of the composite discretion index $E_j^\ast$, and Appendix Table \ref{table:ad-text-mechanism-comparison} reports results using advertisement-level text measures constructed from job-posting language rather than occupation-level O*NET measures. The text-based results replicate the directional pattern. Callback gaps are again largest where subjective demands are high relative to the job's verifiable content.

\subsection{Credential Signal Attenuation} \label{section:credentials}

  \noindent The occupation patterns show where callback gaps concentrate. The randomized credentials provide a more direct test of the mechanism by varying applicants' verifiable information within each job advertisement and measuring whether this information narrows callback gaps more where evaluation rests on structured criteria. Because credentials are randomized within advertisements, this evidence carries greater weight than the cross-occupation comparisons. Table \ref{table:credential-attenuation} uses randomly assigned r\'{e}sum\'{e} credentials to test whether informative signals differentially attenuate non-White-male callback gaps across discretion regimes. We partition the sample at the median of the composite discretion index $E_j^\ast$ and estimate a single regression that includes each credential's main effect, its interaction with the pooled non-White-male indicator, and the triple interaction $\left(\text{Credential}\times\text{Non-White-male}\times\text{High Discretion}\right)$. Job-advertisement fixed effects absorb all advertisement-level variation.

  Column 1 reports average callback returns. Three credentials have positive and significant effects on callbacks. Social internships increase callbacks by 1.1 percentage points, programming and data skills by 1.0 percentage point, and study abroad by 0.8 percentage points. The remaining three credentials (GPA, quantitative internships, and math minor) have returns indistinguishable from zero. Columns 2 through 4 examine how credential returns differ by non-White-male status and discretion regime.

  For credentials with positive returns, the pattern is broadly consistent with the evaluative-discretion prediction. In low-discretion jobs (column 2), the two with the larger callback returns narrow non-White-male callback gaps by 3.7 to 4.1 percentage points. In high-discretion jobs (column 3), the same credentials provide no differential benefit. The interaction estimates are close to zero, and the confidence intervals rule out differential benefits larger than about three percentage points. The corresponding triple interactions (column 4) are negative for all three positive-return credentials, and a joint test rejects the null that they are jointly zero at the 5 percent level. 
  
  The estimated low-discretion interaction effects are larger than the pooled non-White-male callback gap in low-discretion jobs, which is about two percentage points, so the point estimates imply that credentialed non-White-male applicants have callback rates roughly two percentage points higher than those of comparable White men. The model bounds an informative credential's effect at full closure. The credential may reduce the gap to zero but cannot reverse it because an objective signal compresses the variance gap toward zero. The estimates therefore lie beyond the model's bound, but the implied post-credential gap is not statistically distinguishable from zero, making the apparent reversal consistent with sampling variation. Quality-interaction reversals of this kind appear in \citet{nunley_racial_2015}. 
  
  Study abroad is the exception among the positive-return credentials. It raises callbacks by an average return of 0.8 percentage points but shows no attenuation of the non-White-male gap within low-discretion jobs. One interpretation is that study abroad signals family resources rather than competence. Such a signal of socioeconomic status can raise callbacks without reducing employer uncertainty about productivity. The theory predicts attenuation for informative productivity signals and is silent on status signals, so we treat the study-abroad result as a partial exception rather than evidence supporting the mechanism.
  
  The zero-return credentials are a placebo comparison. They show no systematic pattern across discretion regimes. Because the six credentials are each interacted with the same discretion split, the triple interactions form a family of related hypotheses. We therefore adjust for multiple testing across the six coefficients. The social-internship triple, which is the largest of the six, survives a Romano-Wolf stepdown adjustment across all six triples, computed from 999 bootstrap replications clustered by job advertisements. The programming triple is not individually distinguishable after adjustment ($p = 0.116$).\footnote{One of three placebo credentials (quantitative internships) reaches marginal significance in low-discretion jobs despite having no average callback return. The corresponding triple interaction is not statistically significant. A single marginally significant estimate among three placebo credentials is consistent with chance variation at conventional significance levels.}

  In low-discretion positions, where evaluation criteria are more structured, verifiable credentials narrow non-White-male callback gaps. In high-discretion positions, the same credentials provide no differential benefit. Credentials that convey competence therefore appear to improve non-White-male applicants' access to routine-intensive jobs but do not narrow gaps in the analytically and interpersonally demanding roles where discrimination is concentrated.

\subsection{Decomposing the Mechanism} \label{section:mechanism}

  \noindent The credential experiment shows that verifiable information narrows non-White-male gaps where evaluation is structured. We next decompose the task gradient into subjective and objective channels to identify where structured evaluation appears in the task space. Table \ref{table:mechanism-contact} reports estimates from regressions of callbacks on race--gender indicators interacted with the subjective noise proxy ($\hat{B}$, constructed from analytical and interpersonal task intensity) and the objective precision proxy ($\hat{P}$, constructed from routine cognitive task intensity). All task components are standardized to mean zero and unit variance, making the coefficients directly comparable.\footnote{This standardization applies to the composites used in the decomposition regressions, where it makes the coefficients directly comparable. The discretion index instead uses the same underlying composites on their natural positive scale, constructed from min-max normalized task measures plus a constant (Appendix \ref{appendix:functional_forms}). Both are strictly positive in the sample ($\hat{B}_j \in [1.00, 2.72]$, $\hat{P}_j \in [1.00, 2.00]$), so $E_j^\ast = \hat{B}_j/(\hat{B}_j + \hat{P}_j)$ lies in $[0.35, 0.66] \subset [0,1)$, as required by the theory. The two applications use monotone transformations of the same task composites for distinct empirical purposes.} Columns 1--3 report full-sample estimates, while columns 4 and 5 split the sample by contact intensity (bottom and top 40 percent, respectively). We present a pooled non-White-male specification (Panel A) and group-specific estimates (Panel B) that relax the common-gradient assumption. The pooled specification provides greater power for the main $\hat{B} - \hat{P}$ test and matches the pooled specification used in the credential analysis in Section \ref{section:credentials}.\footnote{The pooled indicator includes all five non-White-male groups, including White women, here and in the pooled specifications of Sections \ref{section:credentials} and \ref{section:screening}. Appendix Table \ref{table:pool-robustness} re-estimates the pooled mechanism, credential, and screening results with White women excluded, so the pool contains Black and Hispanic applicants only. The point estimates are essentially unchanged, the credential joint test strengthens, and the screening interactions retain about 95 percent of their magnitude with wider confidence intervals from the smaller sample.}

  Panel A pools the five non-reference groups. The clearest pattern operates through objective evaluability. The routine-cognitive proxy $\hat{P}$ is positive and significant at the 10 percent level, indicating that callback gaps are smaller where tasks can be screened against verifiable criteria. The subjective proxy $\hat{B}$ is negative, as the model predicts, but not individually distinguishable from zero. This imprecision is consistent with subjective evaluation varying idiosyncratically across employers rather than systematically with occupation-level task content. The subjective proxy is also built from less precisely measured constructs, so it carries more measurement error and is subject to greater attenuation than the routine proxy. The within-group difference $\hat{B} - \hat{P}$ is $-0.012$ and significant at the 5 percent level. A generic screening-difficulty account, under which all forms of evaluation difficulty widen gaps, would predict a systematic subjective effect that the data do not show.

  Panel B reports group-specific estimates. The $\hat{P}$ coefficients are positive for all five non-reference groups and jointly significant at the 1 percent level, indicating that the objective-evaluability channel is present across groups. The $\hat{B}$ coefficients are negative for four of the five groups but mostly imprecise, as the measurement-error reading above predicts. The within-group difference $\hat{B} - \hat{P}$ is negative for all five groups, with two individually significant at conventional levels.

  Columns 4 and 5 test Proposition \ref{prop:contact_main}'s prediction that contact amplifies the discretion channel. The difference $\hat{B} - \hat{P}$ is close to zero in low-contact jobs but $-0.043$ and significant at the 1 percent level in high-contact jobs.\footnote{We test whether $(\hat{B} - \hat{P})_{\text{low}} = (\hat{B} - \hat{P})_{\text{high}}$ using a three-way interaction model (group $\times$ task $\times$ contact) estimated on the combined low/high sample. The joint test across all five non-reference groups is marginally significant at conventional levels. It loses significance when group-specific callback gaps are allowed to vary by major occupation group, since high-contact jobs concentrate in sales and office occupations and the standard errors roughly double. We therefore read the contact pattern as descriptive support for the amplification prediction rather than an independent test.} In Panel B, all five groups exhibit more negative $\hat{B} - \hat{P}$ differences in high-contact than in low-contact jobs, consistent with contact amplifying discrimination. However, the design does not separate the subjective-evaluation channel from customer prejudice. The high-contact difference is larger for White women ($-0.049$) than for the other groups, indicating that the contact pattern is partly a gender result. Customer prejudice can operate over gender as much as race, so a large White-women effect in customer-facing jobs is consistent with a customer-preference channel rather than against one. We therefore do not treat the contact result as separating evaluative discretion from customer prejudice.

  The magnitudes are large relative to the estimated credential returns. A one-standard-deviation increase in $\hat{P}$ (objective precision) raises callbacks for White women by 1.9 percentage points, more than the return to any r\'{e}sum\'{e} credential in our design. Moving from a low-discretion to a high-discretion job, corresponding to roughly a two-standard-deviation shift in $\hat{B} - \hat{P}$, is associated with a widening of the pooled non-White-male callback gap of about 2.4 percentage points, similar in magnitude to the overall callback gaps reported in Table \ref{table:discrim-occ}. Both figures rest on inference clustered at the job-ad level, an approach that the occupation-level checks below qualify.

  Because the task moderators vary at the occupation level, occupation-level inference warrants separate attention. Re-estimating the decomposition with clustering by the 175 six-digit occupations leaves the standard error on the pooled $\hat{B} - \hat{P}$ estimate unchanged (0.0063, $p = 0.051$), but the conventional cluster-robust calculation may overstate the information in the cross-occupation variation. Partialling the $\hat{B} - \hat{P}$ regressor with respect to the controls and the job-ad fixed effects and measuring each occupation's share of the remaining variation, in the spirit of \citet{carter_asymptotic_2017}, gives an effective cluster count of roughly 11 rather than 175. The five largest occupations contribute 52 percent of the identifying variation. A wild cluster bootstrap at the occupation level, which is robust to this concentration, yields $p = 0.15$.\footnote{Bootstrap $p$-values impose the null and use 9,999 replications \citep{cameron_bootstrap-based_2008, roodman_fast_2019}, with Rademacher and Webb weights agreeing throughout (0.149 and 0.144 here). The same calculation applied to the screening-prevalence interaction of Section \ref{section:screening} gives an effective cluster count of roughly 4, with 68 percent of the identifying variation in the five largest occupations, and bootstrap $p$-values of 0.15 and 0.14 under the two weighting schemes. The credential triple interactions of Section \ref{section:credentials} moderate a within-ad-randomized credential effect by the occupation-level discretion split, so they too draw on cross-occupation variation. Clustered at the occupation level, they survive the same bootstrap (joint $p = 0.082$, with the two largest individual triples at $p = 0.029$ and $p = 0.054$). Dropping one occupation at a time leaves the pooled $\hat{B} - \hat{P}$ estimate negative in all 175 leave-one-out re-estimates, so no single occupation drives the pattern.} We therefore read the cross-occupation decomposition as a sign-consistent pattern whose apparent precision is strongest under the job-ad-level design. We place primary weight on the credential experiment of Section \ref{section:credentials}, whose credential effects are randomized within the ad, though the comparison across discretion regimes there also draws on occupation-level variation.

  A horse race against direct measures of job quality gives the cross-occupation pattern additional empirical content. Occupations with high $\hat{B} - \hat{P}$ are also better-paid and more credentialed (across advertisements, the correlations with log median wage, the college share of employment, and the O*NET job zone range from 0.54 to 0.61), so the gradient could in principle reflect any account in which discrimination rises with job quality rather than evaluative discretion. Appendix Table \ref{table:complexity-horserace} interacts the pooled non-White-male indicator simultaneously with $\hat{B} - \hat{P}$ and with each quality measure, constructed from the American Community Survey and O*NET. The discretion contrast is slightly larger with all three quality interactions included ($-0.014$), while the quality interactions themselves are individually small with mixed signs. Conditioning on the quality measures also absorbs cross-occupation variation unrelated to the mechanism. The occupation-level wild cluster bootstrap $p$-value for the conditional difference is 0.036, compared with 0.15 for the unconditional comparison on the full sample and 0.08 on the common sample with all quality measures present.\footnote{Because the quality measures are missing for six occupations, the conditional columns of Table \ref{table:complexity-horserace} use a slightly smaller sample than the baseline. On the common sample, the unconditional bootstrap $p$-value is 0.08 and the conditional one is 0.036, so the improvement is not an artifact of the sample change. A permutation placebo indicates it is not an artifact of conditioning either. Reassigning the occupation-level discretion profile at random across occupations and re-running the conditional specification, the analytic occupation-clustered $p$-value clears the 5 percent level in 13 percent of 200 placebo draws and falls below the observed analytic value of 0.005 in 4 percent. These figures use the analytic $p$-value, not the wild cluster bootstrap $p$-value of 0.036 reported above. A null contrast therefore does not inherit the conditional significance that the genuine discretion signal shows.} Comparing occupations at similar wage and education levels, these dimensions of job quality do not account for the discretion contrast, and the conditional comparison survives occupation-level inference. The unconditional comparison does not, so we retain the evidentiary weighting described above. The credential experiment carries the mechanism evidence, and the cross-occupation decomposition maps where the pattern appears.

\subsection{Screening Instruments in the Advertisement Text} \label{section:screening}

  \noindent The advertisement text provides an independent measure of the screening environment that does not rely on the O*NET task taxonomy. We searched each posting for references to a verifiable screening instrument, using an \emph{a priori} dictionary of pre-employment tests and assessments, certifications and licenses, background and drug checks, and stated GPA requirements. About 29 percent of advertisements mention at least one instrument, usually a certification or license (Table \ref{table:screening-instruments}). Routine cognitive intensity is the only task dimension significantly associated with their prevalence. Management, where callback gaps are concentrated, names them less often than the other major occupation groups. We use the share of an occupation's advertisements that mention at least one instrument as a proxy for the objective-screening environment. This proxy is only modestly correlated with the discretion index $E_j^\ast$ (correlation $-0.22$ across advertisements), so the screening measure is not simply a restatement of the task-based decomposition. Inspection of the individual matches shows that the dictionary rarely identifies a credential applicants must hold when applying. Most matches describe licenses obtained after starting, often with employer sponsorship, or certifications listed as preferred. The prevalence measure therefore captures the extent to which an occupation's skill standard is defined by a verifiable third-party credential rather than whether the experimental applicants already possess that credential.\footnote{Among advertisements with a certification or license mention, 32 percent describe licensure obtainable after hire (insurance, securities, and mortgage licensing dominate), 17 percent name professional certifications such as the CPA, and 17 percent contain licensed-status language. In addition, 23 percent match a driver's license and 13 percent firm self-description. These categories are not mutually exclusive. Column 3 of Table \ref{table:screening-robustness} shows that the substantive credential mentions account for the screening result, while driver's licenses and other non-substantive matches show no relationship with gaps. This reading matches the design. We applied to no position that required a license in hand (Section \ref{section:design}), so the prevalence measure reflects the occupation's credentialing environment rather than an application-stage requirement imposed on the fictitious applicants.}

  Callback gaps shrink where advertisements more frequently announce standardized screening. Columns 3 and 4 of Table \ref{table:screening-gaps} interact the pooled non-White-male indicator with continuous instrument prevalence. Each one-standard-deviation increase in prevalence narrows the pooled gap by 0.9 percentage points. The predicted gap falls from 0.8 percentage points at mean prevalence to essentially zero one standard deviation above the mean. The interaction is statistically significant with standard errors clustered by job advertisements. It is nearly unchanged when the non-White-male indicator is also interacted with the task intensities, indicating that the screening environment is associated with callback gaps beyond occupation-level task content.

  Columns 1 and 2 illustrate the same pattern using a median split. The pooled non-White-male gap is 1.6 percentage points in low-screening occupations, equivalent to an 8 percent reduction relative to the White-male callback rate in that half of the sample. In high-screening occupations the gap is statistically indistinguishable from zero, and the confidence interval rules out gaps larger than about one percentage point. The difference across the two groups is statistically significant with clustering by job advertisements but only marginally so with clustering by occupations, and it loses precision when the task-intensity interactions are added (column 2). We therefore treat the split as descriptive and base inference on the continuous specification.

  Table \ref{table:screening-robustness} addresses three potential alternative explanations and then examines the sensitivity of the result to occupation-level inference. First, White-male callback rates decline with screening prevalence. The occupation-level correlation is $-0.39$, and callback rates are 0.204 and 0.098 in the low- and high-prevalence halves of the sample, respectively (memo rows, Table \ref{table:screening-gaps}). Because percentage-point gaps tend to be smaller when baseline callback rates are lower, the levels result could arise even if proportional discrimination were constant. A Poisson specification with job-ad fixed effects addresses this concern by expressing the pooled gap as a log callback ratio. The proportional gap also narrows with prevalence, by 0.07 log points per standard deviation, and the estimate remains significant whether standard errors are clustered by job advertisement or by occupation. This result weighs against an explanation based solely on lower baseline callback rates. The levels interaction also remains similar when the model includes an interaction with the occupation's White-male callback rate computed after excluding the focal advertisement. The base-rate interaction itself is small and imprecisely estimated. 
  
  Second, dividing the certification and license mentions by content shows the result is concentrated among substantive credentials. Mentions of driver's licenses, firm self-description and other incidental matches show no relationship with callback gaps. Third, excluding the 406 advertisements flagged in the advertisement-text audit as potentially requiring a license in hand produces slightly larger estimates in both the levels and proportional specifications. 
  
  The occupation-level inference caution discussed in Section \ref{section:mechanism} applies with particular force here. The effective number of occupation clusters underlying the prevalence interaction is only about 4 and a wild cluster bootstrap at the occupation level yields $p = 0.15$. The precision claims therefore rest on clustering at the job-ad level, which matches the level of randomization. We therefore interpret the screening evidence as directionally consistent with the within-ad credential experiment, but not as independently decisive evidence of the mechanism.

  Firm-level studies of test adoption complement our advertisement-based evidence by examining how employers use screening instruments after adoption. \citet{autor_does_2008} find that a national retail chain's standardized job test raised the tenure of its hires without reducing the minority share of hires, and \citet{hoffman_discretion_2018} find that hires made against a job test's recommendation had shorter tenures. Those studies observe how firms use screening instruments after adoption, while our measure captures the instruments an occupation's advertisements announce to applicants. These announcements define the screening environment visible to applicants, and callback gaps are narrower in occupations where they are more prevalent.

\section{Conclusion}

  \noindent Hiring discrimination varies systematically with the task content of jobs. In our r\'{e}sum\'{e} audit of new college graduates, callback gaps relative to White men concentrate in analytically and interpersonally demanding jobs, where employers must assess fit using the r\'{e}sum\'{e}, and shrink in routine jobs that can be screened against verifiable criteria. Randomly assigned credentials that raise callback rates, which provide our central test of the mechanism, narrow those gaps only where evaluation is structured. 

  We develop an evaluative-discretion framework in which discrimination narrows when objective criteria constrain the hiring decision and widens when subjective judgment carries greater weight. Our evidence supports the first prediction more sharply than the second, as expected if subjective bias varies idiosyncratically across employers and is harder to measure. The pattern for Hispanic women, who face no significant callback gap in the baseline specifications, is consistent with the model's case in which subjective evaluation introduces little additional noise relative to White men. In a measure constructed from the advertisement text rather than the O*NET task data, callback gaps also shrink where advertisements announce standardized screening instruments. Taken together, the results are difficult to reconcile with a generic explanation in which job complexity alone raises discrimination.

  The design, however, identifies where discrimination concentrates across job types more cleanly than why it does so. Jobs that differ in task content differ along other dimensions, and we cannot experimentally manipulate tasks, so the results are consistent with evaluative discretion without proving it. We submitted the applications in 2016 and 2017, so the estimates predate much of the subsequent expansion of applicant-tracking systems and algorithmic r\'{e}sum\'{e} screening. Such systems may function as standardized screening instruments in the sense of Section \ref{section:screening}. The framework predicts that their spread will narrow callback gaps where they displace subjective r\'{e}sum\'{e} review but will have less effect where evaluation remains highly discretionary.

  The mid-career workforce observed in the ACS is already sorted along the same gradient. Among mid-career college-educated workers in the same occupations as our applicants, the White-male share rises with a job's discretion index while the Black and Hispanic shares fall. Black men are 9.5 percentage points less likely than White men to hold the high-discretion jobs where callback gaps concentrate, and 7.7 percentage points less likely among workers who studied the same fields as our applicants.\footnote{We compute these shares from the 2015--2019 American Community Survey, restricting the sample to employed workers aged 35 to 54 with a bachelor's degree and assigning each the discretion index of the worker's six-digit occupation, within the 175 occupations in our audit. The 9.5-point figure pools all fields of degree, while the 7.7-point figure restricts the sample to the eight fields of study randomized across our r\'{e}sum\'{e}s.} We do not read this sorting as the causal consequence of callback discrimination, since it reflects schooling, preferences, and employer and worker decisions at many stages beyond the first screen. Even so, discrimination that concentrates in the jobs that build task-specific human capital runs in the same direction as the inequality already visible among incumbents. The consequences of discrimination for careers therefore depend not only on its average magnitude but also on where it occurs in the task distribution.

  The gaps concentrate in the higher-paying analytical and interpersonal jobs within the entry-level fields we sample, so the long-run consequences depend on the value of a foregone callback. If these jobs build task-specific human capital, exclusion at the callback stage may compound into later disparities. If foregone callbacks rarely would have produced offers, the downstream employment effects are smaller. \citet{jarosch_statistical_2019} formalize the second case. In their framework, callbacks are an intermediate stage, and subsequent interviews can partly offset initial screening gaps, so callback gaps may overstate differences in eventual job finding. The offer-stage evidence points the other way on average, with discrimination continuing past the callback \citep{quillian_evidence_2020}. The correction may be more limited in our setting if interviews for high-discretion jobs rely on the same forms of subjective judgment as the initial r\'{e}sum\'{e} screen. More generally, the evidence suggests that verifiable information narrows disparities only when it meaningfully constrains the evaluation process.

\bibliographystyle{chicago}
\singlespacing
\bibliography{references}


\section*{Tables}
\addtolength{\textheight}{5pt}
\FloatBarrier 

\begin{table}[H]
    \centering
    \caption{List of Prevalent Detailed Occupations by Major Occupation Group}
    \label{table:occlist-new}
    \input{tables/AER/table3-occlist}
\end{table}

\clearpage
\begin{table}[H]
    \centering
    \caption{Discrimination Overall and By Major Occupation Group}
    \label{table:discrim-occ}
    \input{tables/AER/table4-discrim-occ}
\end{table}

\clearpage
\input{tables/AER/table5-discrim-tasks-kmeans4}

\clearpage
\input{tables/AER/table-credential-attenuation}

\clearpage
\input{tables/AER/table-mechanism-contact}
\clearpage

\begin{table}[H]
    \centering
    \caption{Non-White-Male Callback Gaps by Occupation-Level Screening-Instrument Prevalence}
    \label{table:screening-gaps}
    \begin{adjustbox}{max width=\textwidth}
    \input{tables/AER/table-screening-gaps}
    \end{adjustbox}
\end{table}
\clearpage



\section*{Figures}

\begin{figure}[H]
\caption{Applications by Commuting Zone}
\label{fig:app-maps}
\begin{center}
\includegraphics[width=1\linewidth]{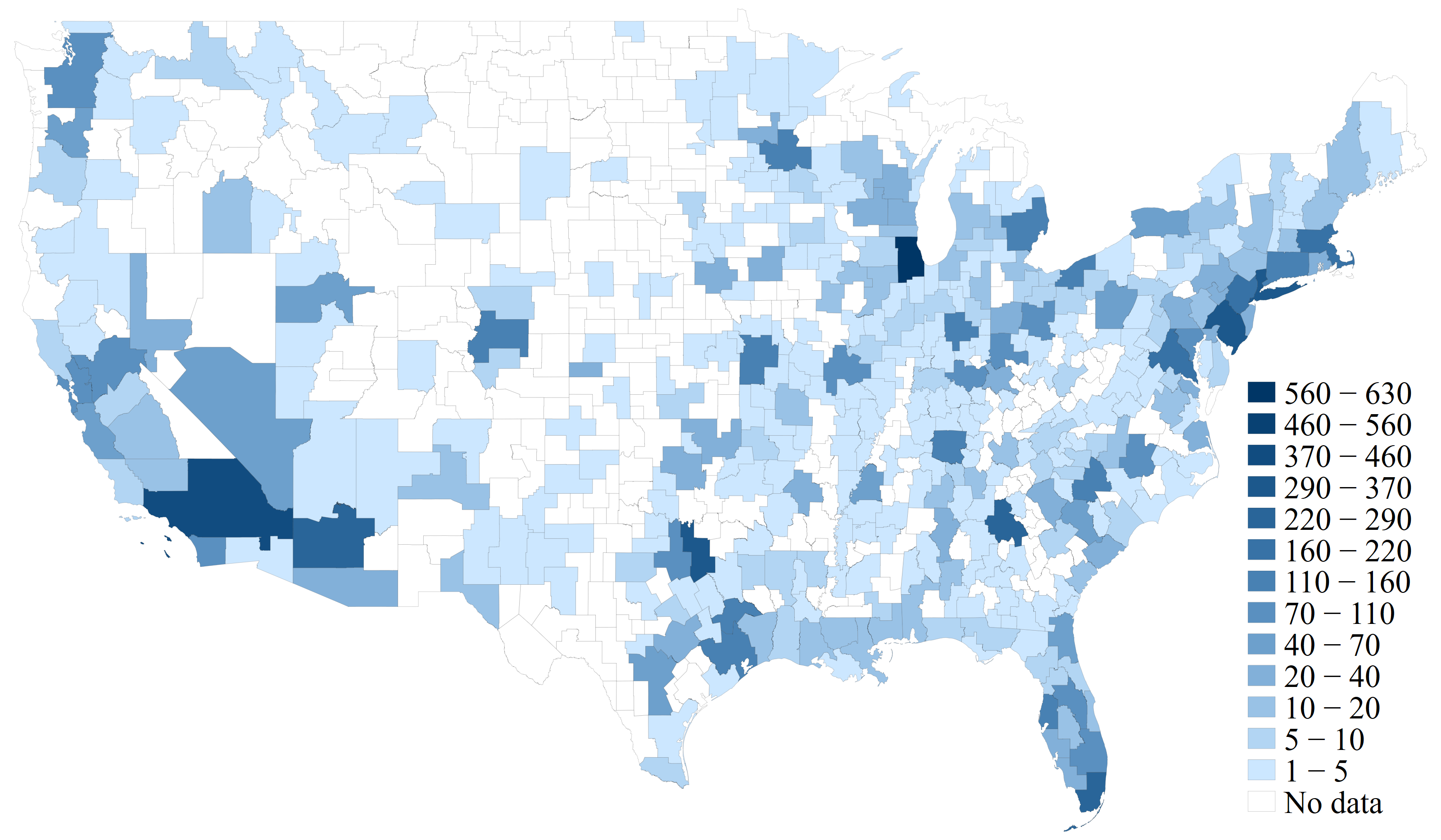}
\end{center}
\scriptsize{\parbox{1\linewidth}{\emph{Notes:} The figure provides a heat map of the number of unique advertisements across commuting zones in the continental United States. We use the advertising firm's geographic location, which is extracted from the ad text, to identify the 1990 commuting zone in which the firm is located.}}
\end{figure}

\clearpage
\begin{figure}[H]
\caption{Distributional Comparisons of Tasks, Occupations in the Audit and All Occupations in the ACS}
\label{fig:task-distributions}
\begin{center}
\includegraphics[width=1\linewidth]{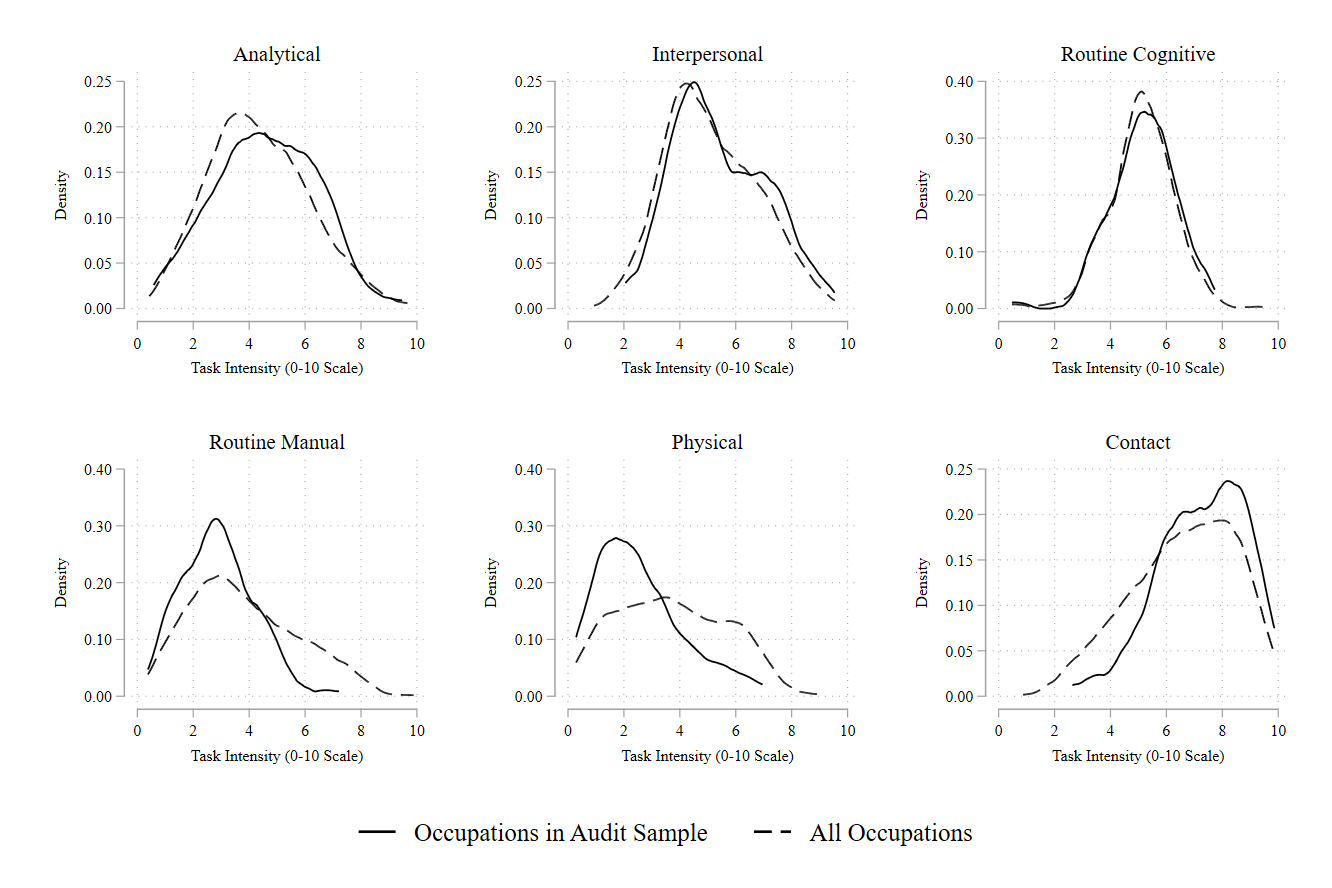}
\end{center}
\scriptsize{\parbox{1\linewidth}{\emph{Notes:} The figure shows kernel density estimates for the occupations in the audit sample (solid line) as well as all occupations in the ACS (dashed line). The estimates are based on a sample of 21--26 year-old college graduates who are employed over the 2015--2018 period from the ACS. This sample is used to compute employment-weighted averages for the task intensity variables by detailed occupation.}}
\end{figure}

\clearpage
\FloatBarrier
\appendix
\section{Appendix Tables}
\renewcommand{\thetable}{A\arabic{table}}  
\renewcommand{\thefigure}{A\arabic{figure}}  
\setcounter{table}{0}  
\setcounter{figure}{0}  

\begin{table}[htbp]
    \centering
    \caption{Correlation Coefficients Between Race/Ethnicity-Gender Indicators and R\'{e}sum\'{e} Attributes}
    \label{table:corr-coeff}
    \input{tables/AER/z-table2-corr-resumeXs}
\end{table}

\clearpage
\input{tables/AER/table-ad-text-mechanism-comparison}

\clearpage
\begin{table}[htbp]
    \centering
    \caption{Discrimination by Major Occupation Group, Adding Occupation-Group by R\'{e}sum\'{e}-Characteristic Interactions}
    \label{table:robust-gxx}
    \input{tables/AER/table-robust-group-x-resume}
\end{table}

\clearpage
\begin{table}[htbp]
    \centering
    \caption{Employer Demand for Judgment-Based Attributes by Task Content}
    \label{table:judgment-demand}
    \input{tables/AER/table-judgment-demand}
\end{table}

\clearpage
\begin{table}[htbp]
    \centering
    \caption{Objective Screening Instruments in Advertisement Text by Task Content}
    \label{table:screening-instruments}
    \input{tables/AER/table-screening-instruments}
\end{table}

\clearpage
\begin{table}[htbp]
    \centering
    \caption{Robustness of Callback Gaps by Screening Environment}
    \label{table:screening-robustness}
    \begin{adjustbox}{max width=\textwidth}
    \input{tables/AER/table-screening-robustness}
    \end{adjustbox}
\end{table}

\clearpage
\begin{table}[htbp]
    \centering
    \caption{Pooled Results With White Women Excluded from the Pool}
    \label{table:pool-robustness}
    \input{tables/AER/table-pool-robustness}
\end{table}

\clearpage
\begin{table}[htbp]
    \centering
    \caption{Discrimination by Major Occupation Group, Poisson Estimates of Proportional Gaps}
    \label{table:occgroup-poisson}
    \begin{adjustbox}{max width=\textwidth}
    \input{tables/AER/table-occgroup-poisson}
    \end{adjustbox}
\end{table}

\clearpage
\input{tables/AER/table5-discrim-tasks-kmeans4-beststart}

\clearpage
\begin{table}[htbp]
    \centering
    \caption{Evaluative Discretion Versus Direct Measures of Job Quality}
    \label{table:complexity-horserace}
    \begin{adjustbox}{max width=\textwidth}
    \input{tables/AER/table-complexity-horserace}
    \end{adjustbox}
\end{table}

\clearpage
\input{tables/AER/table-wm-baserates}


\clearpage
\FloatBarrier

\begin{spacing}{1.5}
\section{Theoretical Appendix}
\label{appendix:theory}

\noindent This appendix provides the formal model and derivations underlying the theoretical framework in Section \ref{section:theory}. It is self-contained and can be read independently of the main text.

\subsection{Environment and Functional Forms}
\label{appendix:functional_forms}

\paragraph{Signal structure.}
Workers from two demographic groups $g \in \{M, m\}$ (majority, minority) apply to jobs. Productivity $\theta \sim N(0,1)$ is identically distributed across groups. For each applicant, firms observe a subjective screening signal $s^s$ and an objective screening signal $s^o$, conditionally independent given productivity. The signals are
\[
s^{s}_{ij} = \theta_i + \varepsilon^{s}_{ij}, \qquad s^{o}_{ij} = \theta_i + \varepsilon^{o}_{ij}.
\]
Noise terms are independent of $\theta$ and of each other. For majority applicants, $\varepsilon^{s}_{M} \sim N(0,B_j)$. For group-$g$ minority applicants, $\varepsilon^{s}_{g} \sim N(0, B_j/\pi_g)$ with $\pi_g \in (0,1)$, so subjective evaluation is less precise for minorities. The objective signal has equal precision across groups, with $\varepsilon^{o} \sim N(0,U_j)$.

\paragraph{Signal noise and objective precision.}
We parameterize subjective noise and objective precision as
\begin{equation}
\label{eq:BU_param}
    B_j = 1 + \alpha a_j + \beta p_j, \qquad P_j = 1 + \delta r_j,
\end{equation}
with $\alpha, \beta, \delta > 0$. The objective noise variance is $U_j = 1/P_j$. Our results require only these sign conditions. The linear forms are for concreteness. The employer places weight $E_j^\ast \in [0,1)$ on the subjective signal, with $E_j^\ast$ increasing in analytical and interpersonal intensity ($a_j, p_j$) and decreasing in routine intensity ($r_j$). We construct its empirical counterpart from the rescaled, strictly positive task composites $\hat{B}_j$ and $\hat{P}_j$,
\begin{equation}
\label{eq:Estar_appendix}
    E_j^\ast = \frac{\hat{B}_j}{\hat{B}_j + \hat{P}_j} \in [0, 1),
\end{equation}
a dimensionless monotone index of the job's reliance on subjective assessment. Because $\hat{B}_j$ and $\hat{P}_j$ are rescaled to a common positive scale, the index is a proper fraction rather than a ratio of a variance to a precision.

\paragraph{Group-level noise differential.}
Define the group-specific noise penalty
\begin{equation}
\label{eq:Delta_def}
    \Delta_g \;\equiv\; \frac{1}{\pi_g} - 1 \;>\; 0.
\end{equation}
A smaller $\pi_g$ implies a larger $\Delta_g$.

\subsection{Composite Evaluation and Variance Gap}
\label{appendix:discrim}

\noindent We model employers as placing weight $E_j^\ast$ on the subjective signal and weight $1 - E_j^\ast$ on the objective signal, so that task demands rather than signal precision govern the combination (see Section \ref{section:theory} for motivation). The composite signal for a group-$g$ applicant is
\[
    \hat{s}_{gj} = E_j^\ast \, s^s_{gj} + (1 - E_j^\ast) \, s^o_{gj},
\]
with variance
\[
    V_{gj} = {E_j^\ast}^2 \, \tau^2_{gj} + (1 - E_j^\ast)^2 \, U_j,
\]
where $\tau^2_{Mj} = B_j$ and $\tau^2_{gj} = B_j / \pi_g$.

\begin{lemma}[Variance Gap]
\label{lem:var_gap}
The composite evaluation variance gap between group-$g$ minority and majority applicants is
\[
    V_{gj} - V_{Mj} \;=\; {E_j^\ast}^2 \, B_j \, \Delta_g.
\]
This gap is positive for $E_j^\ast > 0$ and increasing in $E_j^\ast$.
\end{lemma}

\begin{proof}
\begin{align*}
    V_{gj} - V_{Mj} &= {E_j^\ast}^2 \left(\frac{B_j}{\pi_g} - B_j\right) \\
    &= {E_j^\ast}^2 \, B_j \left(\frac{1}{\pi_g} - 1\right) \\
    &= {E_j^\ast}^2 \, B_j \, \Delta_g.
\end{align*}
For $E_j^\ast > 0$, with $B_j, \Delta_g > 0$, the gap is positive. It is increasing in $E_j^\ast$ by inspection. The gap is also increasing in $B_j$, since $\partial({E_j^\ast}^2 B_j)/\partial B_j = {E_j^\ast}^2 + 2 E_j^\ast B_j P_j/(B_j+P_j)^2 > 0$.
\end{proof}

\subsection{Callbacks Decrease in Evaluation Noise}
\label{appendix:callback}

\begin{lemma}[Callback Monotonicity]
\label{lem:callback}
Suppose employers call back applicants when posterior expected productivity exceeds a threshold $\bar{\theta} > 0$. With Gaussian signals and prior $\theta \sim N(0,1)$, the callback probability is strictly decreasing in the composite evaluation noise variance $V_{gj}$.
\end{lemma}

\begin{proof}
The composite assessment $\hat{s}_{gj} = \theta_i + \text{noise}$ has unconditional distribution $N(0, 1 + V_{gj})$. The posterior mean given the composite signal is $E[\theta \mid \hat{s}_{gj}] = \hat{s}_{gj}/(1 + V_{gj})$. The employer calls back when $E[\theta \mid \hat{s}_{gj}] > \bar{\theta}$, equivalently when $\hat{s}_{gj} > \bar{\theta}(1 + V_{gj})$. The unconditional callback probability is therefore
\[
    c_{gj} \;=\; 1 - \Phi\!\left(\bar{\theta}\sqrt{1 + V_{gj}}\right),
\]
which is strictly decreasing in $V_{gj}$ since $\bar{\theta} > 0$.\footnote{The behavioral-weights assumption governs how the employer \emph{combines} subjective and objective signals into $\hat{s}_{gj}$. The callback decision still uses the posterior mean, reflecting the employer's recognition that the composite assessment is noisy. This is a weaker rationality requirement than full Bayesian signal extraction, which would also adjust the combination weights.}
\end{proof}

Since callbacks are decreasing in $V_{gj}$ and $V_{gj} > V_{Mj}$ for any minority group $g$, the callback gap is
\begin{align}
    D_{gj} \;\equiv\; c_{Mj} - c_{gj} &\approx\; -c'(V_{Mj}) \cdot (V_{gj} - V_{Mj}) \nonumber \\
    &=\; \phi\!\left(\bar{\theta}\sqrt{1+V_{Mj}}\right) \cdot \frac{\bar{\theta}}{2\sqrt{1+V_{Mj}}} \cdot (V_{gj} - V_{Mj}). \label{eq:taylor_callback}
\end{align}
Substituting the variance gap from Lemma \ref{lem:var_gap} yields
\begin{equation}
\label{eq:callback_gap_full}
    D_{gj} \;\approx\; K_{0j} \, {E_j^\ast}^2 \, B_j \, \Delta_g,
\end{equation}
where
\[
    K_{0j} \;=\; \phi\!\left(\bar{\theta}\sqrt{1+V_{Mj}}\right) \cdot \frac{\bar{\theta}}{2\sqrt{1+V_{Mj}}} \;>\; 0.
\]
Because $c(V)$ is convex, $K_{0j}$ is decreasing in $V_{Mj}$, and the first-order term overstates the exact gap when the variance gap is large. The approximation is accurate for comparisons across jobs with similar baseline evaluation variance and modest variance gaps.

\subsection{Monotonic Representation}
\label{appendix:sufficient_stats}

\noindent The model contains parameters ($B_j$, $\Delta_g$, $\bar{\theta}$) that a correspondence audit cannot separately identify. We isolate a single focal object governing the testable predictions and collect the remaining parameters into a scale factor whose sign is unambiguous.

Equation \eqref{eq:callback_gap_full} shows that the callback gap depends jointly on $E_j^\ast$, $B_j$, and $\Delta_g$. Because $E_j^\ast = B_j/(B_j + P_j)$, these are not independently varying. Define the composite scale factor
\[
    K_{gj} \;\equiv\; K_{0j} \, E_j^\ast \, B_j \, \Delta_g \;>\; 0.
\]
The callback gap can then be written as
\begin{equation}
\label{eq:callback_sufficient}
    D_{gj} \;\approx\; K_{gj} \, E_j^\ast,
\end{equation}
which is equation \eqref{eq:main_result} in the main text. The scale factor $K_{gj}$ depends on $E_j^\ast$ itself, so this is not a decomposition into independent components. Three properties justify the representation. First, $K_{gj} > 0$ for all feasible parameter values, so the variance gap is unambiguously increasing in $E_j^\ast$. Second, the empirical specifications cannot separately identify $B_j$ from $E_j^\ast$, so the decomposed expression in equation \eqref{eq:callback_gap_full} offers no additional empirical leverage beyond equation \eqref{eq:callback_sufficient}. Third, discrimination vanishes when $E_j^\ast = 0$ and increases with evaluative discretion, which is the central economic content.

\paragraph{Robustness.}
The variance-gap prediction does not depend on the specific weighting rule. Any composite $\hat{s} = f(\omega_j) \, s^s + [1 - f(\omega_j)] \, s^o$ with $f$ strictly increasing and $f(\omega) > 0$ for $\omega > 0$ yields a variance gap that is increasing in $E_j^\ast$.

\subsection{Proof of Proposition \ref{prop:task_main} (Task Bundles)}
\label{appendix:proof_task}

\noindent Under the parameterization \eqref{eq:BU_param}, $B_j = 1 + \alpha a_j + \beta p_j$ and $P_j = 1 + \delta r_j$. From Lemma \ref{lem:var_gap}, the variance gap is
\[
    V_{gj} - V_{Mj} = \Delta_g \frac{B_j^3}{(B_j + P_j)^2}.
\]
The partial derivatives of the variance gap with respect to each task intensity are
\[
    \frac{\partial}{\partial a_j}(V_{gj} - V_{Mj}) = \Delta_g\alpha \frac{B_j^2(B_j + 3P_j)}{(B_j + P_j)^3} > 0,
\]
\[
    \frac{\partial}{\partial p_j}(V_{gj} - V_{Mj}) = \Delta_g\beta \frac{B_j^2(B_j + 3P_j)}{(B_j + P_j)^3} > 0,
\]
\[
    \frac{\partial}{\partial r_j}(V_{gj} - V_{Mj}) = -2\Delta_g\delta \frac{B_j^3}{(B_j + P_j)^3} < 0.
\]
Analytical and interpersonal tasks raise the variance gap. Routine tasks lower it. Management-type occupations have high analytical and interpersonal content ($B_j$ large) and low routine content ($P_j$ small), producing a large variance gap. Office and administrative occupations have lower non-routine content ($B_j$ smaller) and higher routine content ($P_j$ larger), producing a smaller gap.

The callback gap is $D_{gj} \approx -c'(V_{Mj}) \cdot (V_{gj} - V_{Mj})$, where $-c'(V_{Mj}) > 0$ (Appendix \ref{appendix:callback}). For comparisons across jobs over which the local Taylor slope $-c'(V_{Mj})$ is approximately constant, the variance-gap ordering carries through to callback gaps. The prediction is about task \emph{bundles}, not individual tasks in isolation. \qed

\subsection{Proof of Proposition \ref{prop:contact_main} (Contact Complementarity)}
\label{appendix:proof_contact}

\noindent Extend the model to allow contact intensity $k_j$ to amplify subjective noise in non-routine jobs.
\[
    B_j(k) = 1 + (\alpha a_j + \beta p_j) \cdot g(k_j),
\]
where $g\colon \mathbb{R}_+ \to \mathbb{R}_+$ is increasing with $g(0) = 1$.

\textbf{Case 1: Routine jobs ($a_j = p_j = 0$).} Then $B_j(k) = 1$ for all $k_j$, so $\partial B_j / \partial k_j = 0$ and hence $\partial E_j^\ast / \partial k_j = 0$. Contact has no effect on the variance gap or on callback gaps.

\textbf{Case 2: Non-routine jobs ($a_j + p_j > 0$).} Then $\partial B_j / \partial k_j = (\alpha a_j + \beta p_j) g'(k_j) > 0$. Since $E_j^\ast$ is increasing in $B_j$, this raises $E_j^\ast$ and increases the variance gap.

The cross-partials
\begin{align*}
    \frac{\partial^2 B_j}{\partial k_j \, \partial a_j} &= \alpha \, g'(k_j) > 0, \\[4pt]
    \frac{\partial^2 B_j}{\partial k_j \, \partial p_j} &= \beta \, g'(k_j) > 0,
\end{align*}
verify the complementarity. Contact affects the variance gap only in non-routine jobs. Callback gaps inherit this pattern for comparisons over which the local Taylor slope $-c'(V_{Mj})$ is approximately constant. If customer animus operated independently of task content, contact would raise gaps in all jobs, while a customer-prejudice channel that works through non-routine demands would mimic the pattern, so the contrast is suggestive rather than decisive. \qed
\end{spacing}


\end{document}

%% file: tables/AER/table3-occlist.tex
\def\sym#1{\ifmmode^{#1}\else\(^{#1}\)\fi}
\fontsize{14pt}{16pt} \selectfont
\begin{adjustbox}{width=\textwidth}
\begin{threeparttable}
\begin{tabular}{p{12.5cm}ccccccccc}
\hline\hline
\noalign{\vskip 0.4em}
&&&&\multicolumn{6}{c}{Task Intensity Centiles} \\
\noalign{\vskip 0.3em}
\cline{5-10}
\noalign{\vskip 0.3em}
&\multicolumn{1}{c}{\# of} &\multicolumn{1}{c}{Share of }&& & & & \\
&\multicolumn{1}{c}{Unique Ads}&\multicolumn{1}{c}{Total Sample}&&\multicolumn{1}{c}{Analyt.} &\multicolumn{1}{c}{Interp.} &\multicolumn{1}{c}{R. Cog.} &\multicolumn{1}{c}{R. Man.} &\multicolumn{1}{c}{Phys.} &\multicolumn{1}{c}{Contact}  \\
\noalign{\vskip 0.3em}
Detailed Occupation&\multicolumn{1}{c}{(1)}&\multicolumn{1}{c}{(2)}&&\multicolumn{1}{c}{(3)}&\multicolumn{1}{c}{(4)}&\multicolumn{1}{c}{(5)}&\multicolumn{1}{c}{(6)} &\multicolumn{1}{c}{(7)} &\multicolumn{1}{c}{(8)}\\
\noalign{\vskip 0.3em}
\hline
\noalign{\vskip 0.3em}
\multicolumn{10}{l}{\emph{Panel A: Management}} \\
Financial Managers&220&\hspace{0.15cm}2.39\%&&89&92&45&9&7&50 \\
Sales Managers&210&\hspace{0.15cm}2.28\%&&80&99&13&4&22&83 \\
Marketing Managers&208&\hspace{0.15cm}2.26\%&&74&91&9&8&9&77 \\
Advertising and Promotions Managers&\hspace{0.15cm}75&\hspace{0.15cm}0.81\%&&64&67&21&37&23&69 \\
General and Operations Managers&\hspace{0.15cm}39&\hspace{0.15cm}0.42\%&&61&90&20&33&31&83 \\
Personal Service Managers, All Other&\hspace{0.15cm}23&\hspace{0.15cm}0.25\%&&82&92&27&23&28&58 \\
&&&&&&&&& \\
\multicolumn{10}{l}{\emph{Panel B: Business and Financial Operations}} \\
Loan Officers&335&\hspace{0.15cm}3.63\%&&47&51&66&13&9&86 \\
Market Research Analysts and Marketing Specialists&322&\hspace{0.15cm}3.49\%&&92&80&22&16&13&20 \\
Financial and Investment Analysts&289&\hspace{0.15cm}3.13\%&&95&53&51&15&8&32 \\
Accountants and Auditors&284&\hspace{0.15cm}3.08\%&&86&83&87&15&12&42 \\
Claims Adjusters, Examiners, and Investigators&148&\hspace{0.15cm}1.61\%&&50&34&80&45&32&95 \\
Personal Financial Advisors&111&\hspace{0.15cm}1.20\%&&91&90&59&9&8&59 \\
Credit Analysts&107&\hspace{0.15cm}1.16\%&&46&44&68&25&22&45 \\
Management Analysts&\hspace{0.15cm}64&\hspace{0.15cm}0.69\%&&95&98&9&4&3&35 \\
Insurance Underwriters&\hspace{0.15cm}57&\hspace{0.15cm}0.62\%&&74&57&72&28&16&47 \\
Business Operations Specialists, All Other&\hspace{0.15cm}56&\hspace{0.15cm}0.61\%&&80&84&9&9&12&55 \\
Human Resources Specialists&\hspace{0.15cm}37&\hspace{0.15cm}0.40\%&&51&73&23&18&23&47 \\
Compensation, Benefits, and Job Analysis Specialists&\hspace{0.15cm}25&\hspace{0.15cm}0.27\%&&78&60&42&14&15&22 \\
Financial Specialists, All Other&\hspace{0.15cm}25&\hspace{0.15cm}0.27\%&&97&84&71&10&14&21 \\
&&&&&&&&& \\
\multicolumn{10}{l}{\emph{Panel C: Sales}} \\
Sales Representatives of Services, Except Advertising, Insurance, Financial Services, and Travel&686&\hspace{0.15cm}7.44\%&&57&9&42&8&11&92 \\
Insurance Sales Agents&681&\hspace{0.15cm}7.39\%&&52&48&49&8&3&89 \\
Sales Representatives, Wholesale and Manufacturing, Except Technical and Scientific Products&589&\hspace{0.15cm}6.39\%&&49&36&11&9&20&82 \\
First-Line Supervisors of Non-Retail Sales Workers&439&\hspace{0.15cm}4.76\%&&92&100&8&6&9&81 \\
Securities, Commodities, and Financial Services Sales Agents&246&\hspace{0.15cm}2.67\%&&74&46&39&11&14&71 \\
Advertising Sales Agents&223&\hspace{0.15cm}2.42\%&&58&21&21&3&8&100 \\
Retail Salespersons&201&\hspace{0.15cm}2.18\%&&29&39&42&25&34&100 \\
First-Line Supervisors of Retail Sales Workers&100&\hspace{0.15cm}1.08\%&&36&77&23&42&48&96 \\
Demonstrators and Product Promoters&\hspace{0.15cm}61&\hspace{0.15cm}0.66\%&&4&4&5&15&19&39 \\
Telemarketers&\hspace{0.15cm}52&\hspace{0.15cm}0.56\%&&5&6&68&44&25&82 \\
Real Estate Sales Agents&\hspace{0.15cm}25&\hspace{0.15cm}0.27\%&&62&64&7&5&26&91 \\
&&&&&&&&& \\
\multicolumn{10}{l}{\emph{Panel D: Office and Administration}} \\
Customer Service Representatives&978&10.61\%&&57&61&87&40&19&82 \\
Loan Interviewers and Clerks&265&\hspace{0.15cm}2.87\%&&79&47&89&20&5&98 \\
Bookkeeping, Accounting, and Auditing Clerks&255&\hspace{0.15cm}2.77\%&&38&27&97&39&28&36 \\
First-Line Supervisors of Office and Administrative Support Workers&207&\hspace{0.15cm}2.25\%&&66&87&29&36&18&75 \\
Bill and Account Collectors&142&\hspace{0.15cm}1.54\%&&27&20&80&45&25&70 \\
Tellers&114&\hspace{0.15cm}1.24\%&&25&26&99&63&40&93 \\
Secretaries and Administrative Assistants, Except Legal, Medical, and Executive&113&\hspace{0.15cm}1.23\%&&34&37&66&40&27&67 \\
Billing and Posting Clerks&\hspace{0.15cm}63&\hspace{0.15cm}0.68\%&&34&25&93&46&22&62 \\
Insurance Claims and Policy Processing Clerks&\hspace{0.15cm}55&\hspace{0.15cm}0.60\%&&21&14&90&54&19&57 \\
Medical Secretaries and Administrative Assistants&\hspace{0.15cm}50&\hspace{0.15cm}0.54\%&&23&24&84&52&21&88 \\
Executive Secretaries and Executive Administrative Assistants&\hspace{0.15cm}41&\hspace{0.15cm}0.44\%&&43&52&77&26&14&76 \\
Office Clerks, General&\hspace{0.15cm}38&\hspace{0.15cm}0.41\%&&15&33&80&39&28&53 \\
New Accounts Clerks&\hspace{0.15cm}37&\hspace{0.15cm}0.40\%&&48&71&84&46&37&97 \\
Office and Administrative Support Workers, All Other&\hspace{0.15cm}28&\hspace{0.15cm}0.30\%&&59&63&70&69&63&52 \\
Receptionists and Information Clerks&\hspace{0.15cm}25&\hspace{0.15cm}0.27\%&&22&46&86&31&20&99 \\
&&&&&&&&& \\
\multicolumn{10}{l}{\emph{Panel E: Other}} \\
Public Relations Specialists&281&\hspace{0.15cm}3.05\%&&86&94&15&7&4&67 \\
Computer User Support Specialists&\hspace{0.15cm}48&\hspace{0.15cm}0.52\%&&65&52&60&48&41&58 \\
Medical Records Specialists&\hspace{0.15cm}21&\hspace{0.15cm}0.23\%&&41&27&96&42&30&23 \\
&&&&&&&&& \\
\noalign{\vskip 0.5em}
\hline\hline
\noalign{\vskip 0.5em}
        \end{tabular}
    \begin{tablenotes}
    \item \fontsize{10pt}{12pt} \selectfont
\emph{Notes:} The table lists detailed occupations with at least 20 unique ads, the percentage of the total sample, and task-intensity centiles for each occupation. Centiles are computed using 21--26 year-old, college-educated workers from the 2015--2018 ACS aggregated to the detailed occupation level. Panels A--D correspond to major occupation groups (2-digit SOC). Panel E (Other) pools the remaining detailed occupations, which span 17 major occupation groups.
    \end{tablenotes}
    \end{threeparttable}
    \end{adjustbox}

%% file: tables/AER/table4-discrim-occ.tex
\def\sym#1{\ifmmode^{#1}\else\(^{#1}\)\fi}
\fontsize{14pt}{16pt} \selectfont    
\begin{adjustbox}{width=\textwidth}
\begin{threeparttable}
        \begin{tabular}{l*{6}{D{.}{.}{-1}}}
\hline\hline 
&&&&&& \\
&&&\multicolumn{4}{c}{\hspace{0.5cm}Major Occupation Group} \\
\noalign{\vskip 0.5em}
\cline{4-7}
\noalign{\vskip 0.5em}
&&&&\multicolumn{1}{c}{\hspace{0.5cm}Business and} &&\multicolumn{1}{c}{\hspace{0.5cm}Office and} \\
&\multicolumn{1}{c}{\hspace{0.5cm}Overall}&&\multicolumn{1}{c}{\hspace{0.5cm}Management}&\multicolumn{1}{c}{\hspace{0.5cm}Financial}&\multicolumn{1}{c}{\hspace{0.5cm}Sales}&\multicolumn{1}{c}{\hspace{0.5cm}Administration}  \\
\noalign{\vskip 0.75em}
&\multicolumn{1}{c}{\hspace{0.5cm}(1)}&&\multicolumn{1}{c}{\hspace{0.5cm}(2)}&\multicolumn{1}{c}{\hspace{0.5cm}(3)}&\multicolumn{1}{c}{\hspace{0.5cm}(4)}&\multicolumn{1}{c}{\hspace{0.5cm}(5)} \\
\noalign{\vskip 0.5em}
\hline 
\noalign{\vskip 0.5em}
White Women&-0.004&&-0.033\sym{**}&-0.006&-0.009&0.005 \\
&(0.005)&&(0.013)&(0.007)\sym{\dagger}&(0.010)&(0.008)\sym{\dagger\dagger} \\
Black Women&-0.014\sym{***}&&-0.046\sym{***}&0.003&-0.016&-0.022\sym{***} \\
&(0.005)&&(0.015)&(0.008)\sym{\dagger\dagger\dagger}&(0.010)&(0.008) \\
Hispanic Women&0.004&&-0.022&-0.002&0.009&0.002 \\
&(0.005)&&(0.015)&(0.007)&(0.010)\sym{\dagger}&(0.008) \\
Black Men&-0.021\sym{***}&&-0.051\sym{***}&-0.009&-0.023\sym{**}&-0.017\sym{**} \\
&(0.005)&&(0.013)&(0.007)\sym{\dagger\dagger\dagger}&(0.010)&(0.008)\sym{\dagger\dagger} \\
Hispanic Men&-0.006&&-0.040\sym{***}&-0.015\sym{**}&0.006&-0.008 \\
&(0.005)&&(0.014)&(0.008)&(0.010)\sym{\dagger\dagger\dagger}&(0.009)\sym{\dagger} \\
&&&&&& \\
Centile Means/Medians/Std. Devs.&&&&&& \\
\hspace{0.25cm}Analytical &\multicolumn{1}{l}{\hspace{0.25cm}62/57/22} &&\multicolumn{1}{l}{\hspace{0.4cm}78/80/10} &\multicolumn{1}{l}{\hspace{0.4cm}76/86/20} &\multicolumn{1}{l}{\hspace{0.35cm}56/52/19} &\multicolumn{1}{l}{\hspace{.8cm}50/57/17}\\
\hspace{0.25cm}Interpersonal &\multicolumn{1}{l}{\hspace{0.25cm}55/51/28} &&\multicolumn{1}{l}{\hspace{0.4cm}91/92/ 9} &\multicolumn{1}{l}{\hspace{0.4cm}66/60/19} &\multicolumn{1}{l}{\hspace{0.35cm}42/39/28} &\multicolumn{1}{l}{\hspace{.8cm}49/61/20}\\
\hspace{0.25cm}R.~Cognitive &\multicolumn{1}{l}{\hspace{0.25cm}49/49/30} &&\multicolumn{1}{l}{\hspace{0.4cm}23/13/15} &\multicolumn{1}{l}{\hspace{0.4cm}55/66/25} &\multicolumn{1}{l}{\hspace{0.35cm}31/39/18} &\multicolumn{1}{l}{\hspace{.8cm}82/87/18}\\
\hspace{0.25cm}R.~Manual &\multicolumn{1}{l}{\hspace{0.25cm}21/15/16} &&\multicolumn{1}{l}{\hspace{0.4cm}13/ 9/12} &\multicolumn{1}{l}{\hspace{0.4cm}17/15/ 9} &\multicolumn{1}{l}{\hspace{0.35cm}11/ 8/ 9} &\multicolumn{1}{l}{\hspace{.8cm}40/40/10}\\
\hspace{0.25cm}Physical &\multicolumn{1}{l}{\hspace{0.25cm}16/13/11} &&\multicolumn{1}{l}{\hspace{0.4cm}15/ 9/ 9} &\multicolumn{1}{l}{\hspace{0.4cm}13/12/ 7} &\multicolumn{1}{l}{\hspace{0.35cm}14/11/10} &\multicolumn{1}{l}{\hspace{.8cm}21/19/ 9}\\
\hspace{0.25cm}Contact &\multicolumn{1}{l}{\hspace{0.25cm}73/82/22} &&\multicolumn{1}{l}{\hspace{0.4cm}69/77/15} &\multicolumn{1}{l}{\hspace{0.4cm}50/42/25} &\multicolumn{1}{l}{\hspace{0.35cm}87/89/10} &\multicolumn{1}{l}{\hspace{.8cm}76/82/18}\\
\noalign{\vskip 0.5em}
\(N\) in Ad Group      &       \multicolumn{1}{c}{\hspace{0.25cm}36,880}    &               &       \multicolumn{1}{c}{\hspace{0.5cm}3,488}      &       \multicolumn{1}{c}{\hspace{0.5cm}7,820}      &       \multicolumn{1}{c}{\hspace{0.35cm}13,372}     &       \multicolumn{1}{c}{\hspace{0.75cm}9,888}      \\
\noalign{\vskip 0.5em}
\hline\hline 
\noalign{\vskip 0.5em}
        \end{tabular}
    \begin{tablenotes}
    \item \fontsize{12pt}{13pt} \selectfont    
  \emph{Notes:} The table presents percentage point differences in callback rates between White men and the other race/ethnicity-gender groups, which include of White women, Black women, Hispanic women, Black men, and Hispanic men, for the full sample of ads as well as subgroups of ads linked to the management, business and financial operations, sales, and office and administration major occupation groups. The full sample of 36,880 observations is used to estimate two separate regression specifications. In column 1, the baseline discrimination estimates are presented, which follows equation (\ref{eq_baseline}). In columns 2-5, we present discrimination estimates from a regression specification with interaction terms between \(D_{i}\) and indicator variables for the major occupation group to which the ad is classified, which follows equation (\ref{eq_interaction}). For columns 2-5, management is the base category. The estimates shown in columns 3, 4, and 5 are, therefore, linear combinations of parameters (the coefficient on \(D_{i}\) plus the coefficient on the interaction between \(D_{i}\) and the major occupation group identifier). The computation of the test statistics for the linear combinations uses the delta method (see STATA's \texttt{lincom} command). The regression specifications include the full set of r\'esum\'e controls as well as ad fixed effects. Standard errors with clustering on job ads are in parentheses. \sym{*}, \sym{**}, and \sym{***} indicate statistical significance at the 10, five, and one percent levels, respectively. Using \sym{\dagger}, \sym{\dagger\dagger}, and \sym{\dagger\dagger\dagger}, we indicate whether the callback gaps in jobs classified as business and financial operations, sales, and office and administration occupations differ statistically from those in the management category at the 10, five, and one percent levels, respectively. In the lower portion of the table, we present the average, median, and standard deviation of the centile estimates for the detailed occupations that comprise the full sample of ads (column 1) and then separately for management, business and financial operations, sales, and office and administration occupations (columns 2--5), respectively.
    \end{tablenotes}
    \end{threeparttable}
    \end{adjustbox}

%% file: tables/AER/table5-discrim-tasks-kmeans4.tex
\begin{table}[htbp]
    \def\sym#1{\ifmmode^{#1}\else\(^{#1}\)\fi}
    \fontsize{8pt}{9pt} \selectfont
    \centering
    \caption{Discrimination Across Clusters of Jobs with Different Task Content}
    \label{table:reg-task-discrim-kmeans4}
    \begin{adjustbox}{width=\textwidth}
    \begin{threeparttable}
        \begin{tabular}{l*{4}{D{.}{.}{-1}}}
\hline\hline
\noalign{\vskip 0.5em}
&\multicolumn{1}{c}{\hspace{0.5cm}(1)}&\multicolumn{1}{c}{\hspace{0.5cm}(2)}&\multicolumn{1}{c}{\hspace{0.5cm}(3)}&\multicolumn{1}{c}{\hspace{0.5cm}(4)} \\
\noalign{\vskip 0.5em}
\hline
\noalign{\vskip 0.5em}
White Women&-0.031\sym{**}&0.003&-0.017\sym{*}&0.006 \\
&(0.013)&(0.010)\sym{\dagger\dagger}&(0.010)&(0.007)\sym{\dagger\dagger} \\
Black Women&-0.023\sym{*}&0.008&-0.012&-0.017\sym{**} \\
&(0.013)&(0.012)\sym{\dagger}&(0.011)&(0.007) \\
Hispanic Women&-0.006&-0.008&0.005&0.007 \\
&(0.012)&(0.009)&(0.011)&(0.007) \\
Black Men&-0.036\sym{***}&-0.009&-0.025\sym{**}&-0.017\sym{**} \\
&(0.012)&(0.011)\sym{\dagger}&(0.011)&(0.007) \\
Hispanic Men&-0.026\sym{**}&-0.012&-0.009&-0.000 \\
&(0.013)&(0.010)&(0.010)&(0.007)\sym{\dagger} \\
&&&& \\
Obs. in Cluster&\multicolumn{1}{c}{\hspace{0.45cm}3,080}&\multicolumn{1}{c}{\hspace{0.45cm}3,556}&\multicolumn{1}{c}{\hspace{0.275cm}11,216}&\multicolumn{1}{c}{\hspace{0.275cm}19,028} \\
\hspace{0.25cm}\% Management&0.352&0.006&0.209&0.002 \\
\hspace{0.25cm}\% Business and Finance&0.610&0.946&0.018&0.125 \\
\hspace{0.25cm}\% Sales&0.000&0.000&0.583&0.359 \\
\hspace{0.25cm}\% Office and Admin.&0.000&0.000&0.074&0.476 \\
WM Callback Rate&0.086&0.022&0.216&0.144 \\
&&&& \\
Centile Mean/Median/Std. Dev.&&&& \\
\hspace{0.25cm}Analytical&\multicolumn{1}{c}{\hspace{-.15cm}90/92/\hspace{0.14cm}4}&\multicolumn{1}{c}{\hspace{-.15cm}88/90/\hspace{0.14cm}7}&\multicolumn{1}{c}{\hspace{-.15cm}69/74/17}&\multicolumn{1}{c}{\hspace{-.15cm}48/52/16} \\
\hspace{0.25cm}Interpersonal&\multicolumn{1}{c}{\hspace{-.15cm}86/84/\hspace{0.14cm}7}&\multicolumn{1}{c}{\hspace{-.15cm}70/81/16}&\multicolumn{1}{c}{\hspace{-.15cm}70/87/29}&\multicolumn{1}{c}{\hspace{-.15cm}39/47/19} \\
\hspace{0.25cm}R. Cognitive&\multicolumn{1}{c}{\hspace{-.15cm}26/22/13}&\multicolumn{1}{c}{\hspace{-.15cm}68/62/16}&\multicolumn{1}{c}{\hspace{-.15cm}17/13/10}&\multicolumn{1}{c}{\hspace{-.15cm}69/70/21} \\
\hspace{0.25cm}R. Manual&\multicolumn{1}{c}{\hspace{-.15cm}12/10/\hspace{0.14cm}5}&\multicolumn{1}{c}{\hspace{-.15cm}17/15/\hspace{0.14cm}8}&\multicolumn{1}{c}{\hspace{-.15cm}12/\hspace{0.14cm}8/11}&\multicolumn{1}{c}{\hspace{-.15cm}28/26/17} \\
\hspace{0.25cm}Physical&\multicolumn{1}{c}{\hspace{-.15cm}11/12/\hspace{0.14cm}6}&\multicolumn{1}{c}{\hspace{-.15cm}12/12/\hspace{0.14cm}8}&\multicolumn{1}{c}{\hspace{-.15cm}16/14/10}&\multicolumn{1}{c}{\hspace{-.15cm}18/19/13} \\
\hspace{0.25cm}Contact&\multicolumn{1}{c}{\hspace{-.15cm}35/35/15}&\multicolumn{1}{c}{\hspace{-.15cm}40/42/11}&\multicolumn{1}{c}{\hspace{-.15cm}80/81/10}&\multicolumn{1}{c}{\hspace{-.15cm}81/86/18} \\
\noalign{\vskip 0.5em}
\hline\hline
\noalign{\vskip 0.5em}
        \end{tabular}
    \begin{tablenotes}
    \item \fontsize{7pt}{8pt} \selectfont  \emph{Notes:} The table presents percentage point differences in callback rates between White men and the other race/ethnicity-gender groups, which include of White women, Black women, Hispanic women, Black men, and Hispanic men, across 4 different clusters of job ads. 
In the upper portion of the table, we present these estimates using the full sample of 36,880 observations to estimate one regression specification, in which indicators for each of the 4 clusters are captured by \(G_{o(j)}\) in equation (\ref{eq_interaction}). In column 1, discrimination estimates are presented for the base group, which provides callback gaps between the different demographic groups and White men for the first cluster of ads. For the estimates shown in columns 2--4, we compute linear combinations of parameters (the coefficient on 
\(D_{i}\) plus the coefficient on the interaction between \(D_{i}\) and \(G_{o(j)}\)). The computation of these linear combinations and their standard errors rely on STATA's \texttt{lincom} command, which sums the coefficient estimates of interest and uses the delta method to compute test statistics. The regression specifications include the full set of 
r\'{e}sum\'{e} controls as well as ad fixed effects. Standard errors with clustering on job ads are in parentheses. \sym{*}, \sym{**}, and \sym{***} indicate statistical significance at the 10, five, and one percent levels, respectively. Similarly, \sym{\dagger}, \sym{\dagger\dagger}, and \sym{\dagger\dagger\dagger} indicate whether the estimated callback gap in columns 2 through 4 differ statistically from the estimates shown in column 1 at the 10, five, and one percent levels, respectively. The middle and lower portions of the table provide information on the characteristics of the ads in each cluster. The shares of observations linked to the four major occupation groups that comprise the vast majority of the audit sample are provided in the middle portion, and the averages, medians, and standard deviations of the task-intensity centiles (1--100) computed for each occupation in the ad cluster are presented in the lower portion of the table. 
The middle portion also reports the White-male callback rate in each cluster, which is the denominator for the relative callback gaps discussed in the text.
    \end{tablenotes}
    \end{threeparttable}
    \end{adjustbox}
\end{table}

%% file: tables/AER/table-credential-attenuation.tex
\begin{table}[htbp]
    \def\sym#1{\ifmmode^{#1}\else\(^{#1}\)\fi}
    \centering
    \caption{Credential Signals and Discrimination by Evaluative Discretion}
    \label{table:credential-attenuation}
    \begin{adjustbox}{max width=\textwidth}
    \begin{threeparttable}
    \begin{tabular}{l*{4}{D{.}{.}{-1}}}
    \hline\hline
    \noalign{\vskip 0.75em}
    &\multicolumn{1}{c}{Callback}&\multicolumn{1}{c}{Low}&\multicolumn{1}{c}{High}&\multicolumn{1}{c}{} \\
    &\multicolumn{1}{c}{Return}&\multicolumn{1}{c}{Discretion}&\multicolumn{1}{c}{Discretion}&\multicolumn{1}{c}{Difference} \\
    &\multicolumn{1}{c}{(ABNS)}&\multicolumn{1}{c}{Cred \(\times\) NWM}&\multicolumn{1}{c}{Cred \(\times\) NWM}&\multicolumn{1}{c}{(3) \(-\) (2)} \\
    &\multicolumn{1}{c}{(1)}&\multicolumn{1}{c}{(2)}&\multicolumn{1}{c}{(3)}&\multicolumn{1}{c}{(4)} \\
    \noalign{\vskip 0.5em}
    \hline
    \noalign{\vskip 0.75em}
    \multicolumn{5}{l}{\emph{Positive-Return Credentials}} \\
    \noalign{\vskip 0.5em}
    Social Intern&0.011\sym{***}&0.041\sym{***}&-0.001&-0.042\sym{**} \\
    &(0.004)&(0.015)&(0.013)&(0.019) \\
    \noalign{\vskip 0.25em}
    Prog.\ + Data&0.010\sym{**}&0.037\sym{**}&-0.002&-0.039 \\
    &(0.004)&(0.017)&(0.016)&(0.024) \\
    \noalign{\vskip 0.25em}
    Study Abroad&0.008\sym{***}&0.004&-0.002&-0.006 \\
    &(0.003)&(0.014)&(0.013)&(0.019) \\
    \noalign{\vskip 0.25em}
    Joint \(F\)-test, Col.\ 4 (\(p\)-value)&&&&0.048 \\
    \noalign{\vskip 0.5em}
    \hline
    \noalign{\vskip 0.5em}
    \multicolumn{5}{l}{\emph{Zero-Return Credentials (Placebo)}} \\
    \noalign{\vskip 0.5em}
    GPA Listed&0.001&-0.011&-0.006&0.005 \\
    &(0.003)&(0.014)&(0.013)&(0.018) \\
    \noalign{\vskip 0.25em}
    Quant.\ Intern&-0.000&0.025\sym{*}&0.008&-0.017 \\
    &(0.003)&(0.014)&(0.014)&(0.020) \\
    \noalign{\vskip 0.25em}
    Math Minor&-0.001&0.011&-0.005&-0.016 \\
    &(0.004)&(0.013)&(0.012)&(0.018) \\
    \noalign{\vskip 0.25em}
    Joint \(F\)-test, Col.\ 4 (\(p\)-value)&&&&0.675 \\
    \noalign{\vskip 0.5em}
    \hline
    \noalign{\vskip 0.5em}
    \(N\) (applications)&\multicolumn{1}{c}{36,880}&&\multicolumn{1}{c}{36,880}& \\
    \noalign{\vskip 0.75em}
    \hline\hline
    \noalign{\vskip 0.75em}
    \end{tabular}
    \begin{tablenotes}
    \item \fontsize{9pt}{11pt} \selectfont 
    \emph{Notes.}
    All six randomized r\'esum\'e credentials are shown, sorted by their callback return. Positive-return credentials have statistically significant positive effects on callbacks (\(p < 0.10\)). Zero-return credentials do not, and they serve as placebo tests.
    Column 1 reports the main effect of each credential on callback, a result established in \citet{arellano-bover_unbundling_2026} and re-estimated here on the same audit sample.
    Columns 2--4 are from a single regression that interacts all six credentials with the pooled non-White-male indicator (NWM, equal to one for the five race-gender groups other than White men) and a high-discretion indicator (\(E^\ast\) above median).
    Column 2 reports Credential \(\times\) NWM in low-discretion jobs. Column 3 reports the same interaction in high-discretion jobs (sum of the two-way and triple interactions). Column 4 reports the triple interaction (Column 3 minus Column 2).
    The evaluative-discretion theory predicts that credentials attenuate discrimination in low-discretion jobs (Column 2 \(> 0\)) but not in high-discretion jobs (Column 3 \(\approx 0\)), yielding a negative triple (Column 4 \(< 0\)), but only for credentials that actually predict callbacks. Placebo credentials should show no pattern.
    The positive-return triples survive a wild cluster bootstrap at the occupation level (joint \(p = 0.082\)), and the social-internship triple survives a Romano-Wolf stepdown adjustment across all six triples (\(p = 0.016\)) (Section \ref{section:credentials}).
    All regressions include job-ad fixed effects, r\'esum\'e controls, and interactions of race-gender with manual task intensity (\(\hat{M}\)) and contact (\(\hat{C}\)). Standard errors clustered at the job-ad level are in parentheses. \sym{*} \(p<0.10\), \sym{**} \(p<0.05\), \sym{***} \(p<0.01\).
    \end{tablenotes}
    \end{threeparttable}
    \end{adjustbox}
\end{table}

%% file: tables/AER/table-mechanism-contact.tex
\begin{table}[htbp]
    \def\sym#1{\ifmmode^{#1}\else\(^{#1}\)\fi}
    \centering
    \caption{Decomposing the Discretion Channel by Contact Intensity}
    \label{table:mechanism-contact}
    \begin{adjustbox}{max width=\textwidth}
    \begin{threeparttable}
    \fontsize{10pt}{12pt}\selectfont
    \begin{tabular}{l*{3}{D{.}{.}{-1}}c*{2}{D{.}{.}{-1}}}
    \hline\hline
    \noalign{\vskip 0.75em}
    &\multicolumn{3}{c}{Full Sample}&&\multicolumn{2}{c}{Subsamples} \\
\cline{2-4}\cline{6-7}
    \noalign{\vskip 0.5em}
    &\multicolumn{1}{c}{Subjective}&\multicolumn{1}{c}{Objective}&\multicolumn{1}{c}{Subjective}&&\multicolumn{1}{c}{Low}&\multicolumn{1}{c}{High} \\
    &\multicolumn{1}{c}{Noise}&\multicolumn{1}{c}{Precision}&\multicolumn{1}{c}{vs.\ Precision}&&\multicolumn{1}{c}{Contact}&\multicolumn{1}{c}{Contact} \\
    &\multicolumn{1}{c}{\(\hat{B}\)}&\multicolumn{1}{c}{\(\hat{P}\)}&\multicolumn{1}{c}{\(\hat{B} - \hat{P}\)}&&\multicolumn{1}{c}{\(\hat{B} - \hat{P}\)}&\multicolumn{1}{c}{\(\hat{B} - \hat{P}\)} \\
    &\multicolumn{1}{c}{(1)}&\multicolumn{1}{c}{(2)}&\multicolumn{1}{c}{(3)}&&\multicolumn{1}{c}{(4)}&\multicolumn{1}{c}{(5)} \\
    \noalign{\vskip 0.5em}
    \hline
    \noalign{\vskip 0.75em}
    \multicolumn{7}{l}{\emph{Panel A. Pooled Across Non-White-Male Groups}} \\
    \noalign{\vskip 0.5em}
    Non-White-male&-0.005&0.007\sym{*}&-0.012\sym{**}&&-0.001&-0.043\sym{***} \\
    &(0.005)&(0.004)&(0.006)&&(0.008)&(0.013) \\
    \noalign{\vskip 0.5em}
    \hline
    \noalign{\vskip 0.5em}
    \multicolumn{7}{l}{\emph{Panel B. Group-Specific Estimates}} \\
    \noalign{\vskip 0.5em}
    White Women&0.001&0.019\sym{***}&-0.018\sym{**}&&0.006&-0.049\sym{***} \\
    &(0.007)&(0.005)&(0.008)&&(0.010)&(0.015) \\
    \noalign{\vskip 0.25em}
    Black Women&-0.008&0.004&-0.012&&0.008&-0.036\sym{**} \\
    &(0.007)&(0.006)&(0.008)&&(0.011)&(0.017) \\
    \noalign{\vskip 0.25em}
    Hispanic Women&-0.000&0.003&-0.003&&0.001&-0.040\sym{**} \\
    &(0.007)&(0.006)&(0.008)&&(0.009)&(0.017) \\
    \noalign{\vskip 0.25em}
    Black Men&-0.007&0.009&-0.016\sym{**}&&-0.011&-0.041\sym{***} \\
    &(0.007)&(0.006)&(0.008)&&(0.010)&(0.015) \\
    \noalign{\vskip 0.25em}
    Hispanic Men&-0.010&0.003&-0.013&&-0.005&-0.046\sym{***} \\
    &(0.007)&(0.006)&(0.008)&&(0.010)&(0.017) \\
    \noalign{\vskip 0.25em}
    \noalign{\vskip 0.25em}
    \noalign{\vskip 0.5em}
    Predicted sign&\multicolumn{1}{c}{\hspace{0.25cm}\(-\)}&\multicolumn{1}{c}{\hspace{0.25cm}\(+\)}&\multicolumn{1}{c}{\hspace{0.25cm}\(-\)}&&\multicolumn{1}{c}{\hspace{0.25cm}\(-\)}&\multicolumn{1}{c}{\hspace{0.25cm}\(-\)} \\
    \noalign{\vskip 0.25em}
    Joint \(F\)-test (\(p\)-value)&0.474&0.005&0.155&&0.408&0.023 \\
    \noalign{\vskip 0.5em}
    \hline
    \noalign{\vskip 0.5em}
    \(N\) (applications)&\multicolumn{3}{c}{36,880}&&\multicolumn{1}{c}{\hspace{0.25cm}15,436}&\multicolumn{1}{c}{\hspace{0.25cm}13,092} \\
    \noalign{\vskip 0.75em}
    \hline\hline
    \noalign{\vskip 0.75em}
    \end{tabular}
    \begin{tablenotes}
    \item \fontsize{9pt}{11pt} \selectfont 
    \emph{Notes.}
    Columns 1--3 are from a single regression on the full sample. Columns 4 and 5 are from separate regressions on subsamples defined by contact intensity, low contact (bottom 40 percent) and high contact (top 40 percent). The pooled indicator in Panel A equals one for the five race-gender groups other than White men.
    \(\hat{B}\) proxies subjective evaluation noise (analytical and interpersonal task intensity). \(\hat{P}\) proxies objective evaluation precision (routine cognitive task intensity).
    The evaluative-discretion theory predicts \(\hat{B} - \hat{P} < 0\). Proposition 2 predicts this gap is larger (more negative) in high-contact jobs.
    Because the task moderators vary at the occupation level, the effective number of occupation clusters behind the column 3 estimate is roughly 11 and a wild cluster bootstrap at the occupation level yields \(p = 0.15\) (Section \ref{section:mechanism}). The same difference conditioned on occupation median wage, college share, and job zone survives that bootstrap (\(p = 0.036\), Appendix Table \ref{table:complexity-horserace}).
    All regressions include job-ad fixed effects, r\'esum\'e controls, and interactions of race-gender with manual task intensity (\(\hat{M}\)) and contact (\(\hat{C}\)). Standard errors clustered at the job-ad level are in parentheses. \sym{*} \(p<0.10\), \sym{**} \(p<0.05\), \sym{***} \(p<0.01\).
    \end{tablenotes}
    \end{threeparttable}
    \end{adjustbox}
\end{table}

%% file: tables/AER/table-screening-gaps.tex
\def\sym#1{\ifmmode^{#1}\else\(^{#1}\)\fi}
\fontsize{12pt}{14pt} \selectfont
\begin{threeparttable}
\begin{tabular}{l D{.}{.}{-1} D{.}{.}{-1} D{.}{.}{-1} D{.}{.}{-1}}
\hline\hline\noalign{\vskip 0.5em}
&\multicolumn{2}{c}{Median split}&\multicolumn{2}{c}{Continuous prevalence} \\
&\multicolumn{1}{c}{(1)}&\multicolumn{1}{c}{(2)}&\multicolumn{1}{c}{(3)}&\multicolumn{1}{c}{(4)} \\
\noalign{\vskip 0.5em}\hline\noalign{\vskip 0.5em}
Non-White-male gap, low-screening&-0.016\sym{***}&-0.016\sym{**}&& \\
&(0.006)&(0.006)&& \\\noalign{\vskip 0.25em}
Non-White-male gap, high-screening&-0.000&-0.000&& \\
&(0.005)&(0.005)&& \\\noalign{\vskip 0.25em}
Difference (high minus low)&0.015\sym{**}&0.015\sym{*}&& \\
&(0.008)&(0.009)&& \\\noalign{\vskip 0.25em}
Non-White-male (gap at mean prevalence)&&&-0.008\sym{**}&-0.008\sym{**} \\
&&&(0.004)&(0.004) \\\noalign{\vskip 0.25em}
Non-White-male \(\times\) screening prevalence (per SD)&&&0.009\sym{**}&0.009\sym{**} \\
&&&(0.004)&(0.004) \\
\noalign{\vskip 0.4em}\hline\noalign{\vskip 0.4em}
White-male callback rate, low-screening&\multicolumn{4}{c}{0.204} \\
White-male callback rate, high-screening&\multicolumn{4}{c}{0.098} \\
\noalign{\vskip 0.25em}
Non-White-male \(\times\) task-intensity controls&\multicolumn{1}{c}{No}&\multicolumn{1}{c}{Yes}&\multicolumn{1}{c}{No}&\multicolumn{1}{c}{Yes} \\
\(N\) (applications)&\multicolumn{1}{c}{36,880}&\multicolumn{1}{c}{36,880}&\multicolumn{1}{c}{36,880}&\multicolumn{1}{c}{36,880} \\
\noalign{\vskip 0.5em}\hline\hline\noalign{\vskip 0.5em}
\end{tabular}
\begin{tablenotes}
\item \fontsize{10pt}{12pt} \selectfont
\emph{Notes.} The table reports pooled callback gaps from regressions of callback on a pooled non-White-male indicator (equal to one for the five race-gender groups other than White men) and its interaction with the occupation-level prevalence of objective screening instruments. Screening prevalence is the share of a six-digit occupation's ads mentioning any instrument (Table \ref{table:screening-instruments}). Columns 1 and 2 split occupations at the median prevalence across ads, and the high-screening gap is the linear combination of the non-White-male coefficient and the interaction, computed with the delta method. Columns 3 and 4 interact the non-White-male indicator with prevalence standardized to mean zero and unit variance across ads. Columns 2 and 4 add interactions between the non-White-male indicator and standardized O*NET analytical, interpersonal, routine cognitive, and contact intensities, so the screening estimates in those columns reflect variation beyond occupation-level task content. Standard errors clustered on job ads are in parentheses. With clustering on the 175 six-digit occupations at which screening prevalence varies, the p-values for the screening estimate are 0.082, 0.131, 0.004, and 0.006 in columns 1 through 4. The effective number of occupation clusters behind the prevalence interaction is roughly 4, and a wild cluster bootstrap at the occupation level yields \(p = 0.15\) (Section \ref{section:screening}). Because White-male callback rates differ across the median split (memo rows), percentage-point gaps are not directly comparable across regimes in relative terms. Across ads, the correlation between the discretion index \(E^\ast\) and screening prevalence is -0.22 for the continuous measure and -0.21 for the median-split indicator, so the screening measure is not a relabeling of the task-based decomposition. All specifications include job-ad fixed effects and the full set of r\'esum\'e controls. \sym{*}, \sym{**}, \sym{***} denote significance at the 10, 5, and 1 percent levels.
\end{tablenotes}\end{threeparttable}

%% file: tables/AER/z-table2-corr-resumeXs.tex
\def\sym#1{\ifmmode^{#1}\else\(^{#1}\)\fi}
\fontsize{7pt}{9pt} \selectfont
\begin{adjustbox}{width=\textwidth}
\begin{threeparttable}
    \begin{tabular}{l*{6}{D{.}{.}{-1}}}
\hline\hline
\noalign{\vskip 0.5em}
&\multicolumn{6}{c}{Demographic Group} \\
\noalign{\vskip 0.25em}
\cline{2-7}
\noalign{\vskip 0.5em}
&\multicolumn{1}{c}{White}&\multicolumn{1}{c}{Black}&\multicolumn{1}{c}{Hispanic}&\multicolumn{1}{c}{White}&\multicolumn{1}{c}{Black}&\multicolumn{1}{c}{Hispanic} \\
&\multicolumn{1}{c}{Women}&\multicolumn{1}{c}{Women}&\multicolumn{1}{c}{Women}&\multicolumn{1}{c}{Men}&\multicolumn{1}{c}{Men}&\multicolumn{1}{c}{Men} \\
\noalign{\vskip 0.5em}
&\multicolumn{1}{c}{(1)}&\multicolumn{1}{c}{(2)}&\multicolumn{1}{c}{(3)}&\multicolumn{1}{c}{(4)}&\multicolumn{1}{c}{(5)}&\multicolumn{1}{c}{(6)} \\
\noalign{\vskip 0.5em}
\hline
\noalign{\vskip 0.5em}
Major-Economics&-0.0081&0.0113&0.0036&-0.0071&0.0025&-0.0021 \\
Major-Finance&-0.0085&0.0093&-0.0007&-0.0033&0.0019&0.0014 \\
Major-Marketing&0.0017&-0.0023&-0.0008&-0.0008&-0.0024&0.0047 \\
Major-Anthropology&0.0009&0.0003&0.0050&0.0041&-0.0035&-0.0068 \\
Major-Philosophy&-0.0007&0.0009&-0.0004&0.0049&-0.0054&0.0006 \\
Major-Chemistry&0.0018&-0.0032&0.0010&-0.0001&-0.0002&0.0006 \\
Major-Biology&0.0008&-0.0045&-0.0011&0.0039&0.0037&-0.0028 \\
Major-Psychology&0.0121&-0.0119&-0.0066&-0.0014&0.0034&0.0044 \\
Minor-Mathematics&0.0099&-0.0029&-0.0120&-0.0020&0.0044&0.0026 \\
Minor-History&-0.0035&0.0027&0.0071&-0.0051&-0.0007&-0.0004 \\
University-Southeast \#1&-0.0636&0.0008&0.0597&0.0237&-0.0365&0.0156 \\
University-Southeast \#2&0.0322&-0.0437&0.0215&0.0589&-0.0722&0.0024 \\
University-Southeast \#3&0.0608&-0.0682&0.0043&0.0097&0.0288&-0.0355 \\
University-West \#1&0.0058&0.0476&-0.0453&0.0157&0.0486&-0.0715 \\
University- West \#2&0.0046&0.0568&-0.0758&-0.0418&0.0179&0.0386 \\
University-Midwest \#1&-0.0397&0.0102&0.0437&-0.0671&0.0053&0.0477 \\
University- Midwest \#2&-0.0696&0.0036&0.0610&0.0411&-0.0378&0.0013 \\
University- Midwest \#3&0.0271&-0.0385&0.0085&0.0611&-0.0701&0.0110 \\
University- Northeast \#1&0.0643&-0.0747&0.0003&0.0112&0.0411&-0.0422 \\
University- Northeast \#2&0.0052&0.0382&-0.0388&0.0022&0.0613&-0.0673 \\
University- Southwest \#1&0.0077&0.0596&-0.0712&-0.0428&0.0115&0.0356 \\
University- Southwest \#2&-0.0352&0.0106&0.0303&-0.0717&0.0033&0.0627 \\
Internship-Analytical&-0.0001&-0.0071&0.0027&-0.0029&0.0085&-0.0011 \\
Internship-Interpersonal&0.0001&-0.0003&-0.0068&0.0055&0.0001&0.0014 \\
Computer-Programming and Data Analysis&0.0049&0.0008&-0.0097&0.0017&0.0087&-0.0064 \\
Computer-Programming&-0.0044&0.0010&0.0060&0.0031&-0.0035&-0.0023 \\
Computer-Data Analysis&0.0066&0.0014&-0.0003&-0.0015&-0.0036&-0.0026 \\
Computer-Basic Computer Skills&-0.0065&0.0009&0.0121&0.0032&-0.0052&-0.0046 \\
Language-Native Fluent&0.0021&-0.0039&-0.0002&-0.0024&0.0079&-0.0034 \\
Language-Native Proficient&-0.0076&0.0036&0.0058&-0.0035&-0.0010&0.0028 \\
Language-Nonnative Fluent&-0.0068&-0.0016&-0.0070&0.0037&0.0090&0.0027 \\
Language-Nonnative Proficient&0.0032&-0.0057&0.0023&-0.0046&0.0039&0.0009 \\
Volunteer Work&0.0037&0.0017&0.0025&-0.0030&0.0019&-0.0068 \\
College Job-Sales&-0.0015&0.0019&0.0013&-0.0115&0.0049&0.0049 \\
College Job-University&-0.0052&0.0031&-0.0024&0.0097&-0.0079&0.0026 \\
College Job-Restaurant&0.0067&-0.0051&0.0010&0.0019&0.0030&-0.0075 \\
GPA-3.8 and 4.0&0.0031&-0.0025&-0.0043&-0.0020&0.0043&0.0015 \\
GPA-3.4 and 3.6&-0.0036&0.0028&-0.0043&0.0048&-0.0095&0.0098 \\
GPA-3.0 and 3.2&-0.0031&-0.0033&0.0080&0.0022&-0.0032&-0.0007 \\
Cover Letter&0.0049&-0.0042&-0.0070&-0.0046&0.0038&0.0071 \\
Study Abroad Scholarship&0.0041&-0.0074&0.0034&-0.0025&0.0053&-0.0029 \\
\noalign{\vskip 0.5em}
\hline\hline
\noalign{\vskip 0.5em}
    \end{tabular}
\begin{tablenotes}
\item \fontsize{6pt}{8pt} \selectfont 
\emph{Notes:} The table presents correlation coefficients between the race/ethnicity-gender indicator variables and the other resume characteristics. The estimates are based on the full sample of 36,880 observations.
\end{tablenotes}
\end{threeparttable}
\end{adjustbox}

%% file: tables/AER/table-ad-text-mechanism-comparison.tex
\begin{table}[htbp]
    \def\sym#1{\ifmmode^{#1}\else\(^{#1}\)\fi}
    \centering
    \caption{Ad-Text Mechanism Decomposition}
    \label{table:ad-text-mechanism-comparison}
    \begin{adjustbox}{max width=\textwidth}
    \begin{threeparttable}
    \begin{tabular}{l*{3}{D{.}{.}{-1}}}
    \hline\hline
    \noalign{\vskip 0.75em}
    &\multicolumn{1}{c}{\(\hat{B}_{text}\)}&\multicolumn{1}{c}{\(\hat{P}_{text}\)}&\multicolumn{1}{c}{\(B - P\)} \\
    &\multicolumn{1}{c}{(1)}&\multicolumn{1}{c}{(2)}&\multicolumn{1}{c}{(3)} \\
    \noalign{\vskip 0.25em}
    \emph{Predicted sign}&\multicolumn{1}{c}{\(-\)}&\multicolumn{1}{c}{\(+\)}&\multicolumn{1}{c}{\(-\)} \\
    \noalign{\vskip 0.5em}
    \hline
    \noalign{\vskip 0.75em}
    \multicolumn{4}{l}{\emph{Panel A: Pooled Across Non-White-Male Groups}} \\
    \noalign{\vskip 0.5em}
    Non-White-male&-0.006\sym{**}&0.003&-0.009\sym{**} \\
    &(0.003)&(0.003)&(0.005) \\
    \noalign{\vskip 0.5em}
    \hline
    \noalign{\vskip 0.5em}
    \multicolumn{4}{l}{\emph{Panel B: Group-Specific Estimates}} \\
    \noalign{\vskip 0.5em}
    White Women&-0.006&0.001&-0.007 \\
    &(0.004)&(0.004)&(0.006) \\
    \noalign{\vskip 0.25em}
    Black Women&-0.005&0.003&-0.008 \\
    &(0.004)&(0.004)&(0.006) \\
    \noalign{\vskip 0.25em}
    Hispanic Women&-0.009\sym{**}&0.000&-0.009\sym{*} \\
    &(0.004)&(0.004)&(0.006) \\
    \noalign{\vskip 0.25em}
    Black Men&-0.005&0.007\sym{*}&-0.012\sym{**} \\
    &(0.004)&(0.004)&(0.005) \\
    \noalign{\vskip 0.25em}
    Hispanic Men&-0.007&0.005&-0.012\sym{**} \\
    &(0.004)&(0.004)&(0.006) \\
    \noalign{\vskip 0.25em}
    \noalign{\vskip 0.25em}
    \hline
    \noalign{\vskip 0.5em}
    \multicolumn{4}{l}{\emph{Diagnostics (Panel B)}} \\
    \noalign{\vskip 0.5em}
    All group coefs \(= 0\) (\(p\)-value)&0.292&0.461& \\
    \noalign{\vskip 0.25em}
    \(B_{text} = P_{text}\) for all groups (\(p\)-value)&&&0.323 \\
    \noalign{\vskip 0.25em}
    Consistent sign (of 5 groups)&5/5&4/5&5/5 \\
    \noalign{\vskip 0.75em}
    \hline\hline
    \noalign{\vskip 0.75em}
    \end{tabular}
    \begin{tablenotes}
    \item \fontsize{10pt}{12pt} \selectfont
    \emph{Notes.}
    The table reports coefficients from a single regression of callback on race-gender indicators interacted with standardized (mean zero, unit variance) ad-level text measures and r\'esum\'e controls.
    \(\hat{B}_{text}\) is the rate of analytical and interpersonal task language per 100 words and \(\hat{P}_{text}\) is the rate of routine cognitive task language, measured from the posting text with an a priori task dictionary.
    Job-ad fixed effects absorb all ad-level characteristics. Column~3 reports \(\hat{B}_{text} - \hat{P}_{text}\) with standard errors from \texttt{lincom}.
    The evaluative-discretion theory predicts \(\hat{B}_{text} - \hat{P}_{text} < 0\).
    Panel~A constrains all non-White-male groups to share a common task gradient and Panel~B reports group-specific estimates. Diagnostics refer to Panel~B. The first row tests whether all five group-specific coefficients jointly equal zero and the second whether \(B_{text} = P_{text}\) for all five groups.
    Standard errors clustered at the job-ad level are in parentheses. \sym{*} \(p<0.10\), \sym{**} \(p<0.05\), \sym{***} \(p<0.01\).
    \end{tablenotes}
    \end{threeparttable}
    \end{adjustbox}
\end{table}

%% file: tables/AER/table-robust-group-x-resume.tex
\def\sym#1{\ifmmode^{#1}\else\(^{#1}\)\fi}
\fontsize{12pt}{14pt} \selectfont
\begin{adjustbox}{max width=\textwidth}
\begin{threeparttable}
\begin{tabular}{l*{4}{D{.}{.}{-1}}}
\hline\hline
\noalign{\vskip 0.5em}
&\multicolumn{4}{c}{Major Occupation Group} \\
\cline{2-5}
\noalign{\vskip 0.5em}
&\multicolumn{1}{c}{Management}&\multicolumn{1}{c}{Business and}&\multicolumn{1}{c}{Sales}&\multicolumn{1}{c}{Office and} \\
&\multicolumn{1}{c}{}&\multicolumn{1}{c}{Financial}&\multicolumn{1}{c}{}&\multicolumn{1}{c}{Administration} \\
\noalign{\vskip 0.5em}
&\multicolumn{1}{c}{(1)}&\multicolumn{1}{c}{(2)}&\multicolumn{1}{c}{(3)}&\multicolumn{1}{c}{(4)} \\
\noalign{\vskip 0.5em}
\hline
\noalign{\vskip 0.5em}
White Women&-0.037\sym{***}&0.000&-0.013&0.010 \\
&(0.013)&(0.007)\sym{\dagger\dagger}&(0.010)&(0.008)\sym{\dagger\dagger\dagger} \\
Black Women&-0.051\sym{***}&0.009&-0.022\sym{**}&-0.015\sym{*} \\
&(0.016)&(0.008)\sym{\dagger\dagger\dagger}&(0.010)&(0.008)\sym{\dagger\dagger} \\
Hispanic Women&-0.028\sym{*}&-0.000&0.009&0.003 \\
&(0.015)&(0.007)&(0.010)\sym{\dagger\dagger}&(0.008)\sym{\dagger} \\
Black Men&-0.056\sym{***}&-0.001&-0.030\sym{***}&-0.010 \\
&(0.014)&(0.007)\sym{\dagger\dagger\dagger}&(0.010)&(0.008)\sym{\dagger\dagger\dagger} \\
Hispanic Men&-0.042\sym{***}&-0.009&0.003&-0.001 \\
&(0.015)&(0.008)\sym{\dagger\dagger}&(0.010)\sym{\dagger\dagger}&(0.009)\sym{\dagger\dagger} \\
\noalign{\vskip 0.5em}
\hline
\noalign{\vskip 0.5em}
\(p\)-value, Joint \(D\times G\)&\multicolumn{4}{c}{0.000} \\
Occ.\(\times\)Resume \(X\) interactions&\multicolumn{1}{c}{Yes}&\multicolumn{1}{c}{Yes}&\multicolumn{1}{c}{Yes}&\multicolumn{1}{c}{Yes} \\
\noalign{\vskip 0.4em}
\(N\) in Ad Group&\multicolumn{1}{c}{3,488}&\multicolumn{1}{c}{7,820}&\multicolumn{1}{c}{13,372}&\multicolumn{1}{c}{9,888} \\
\noalign{\vskip 0.5em}
\hline\hline
\noalign{\vskip 0.5em}
\end{tabular}
\begin{tablenotes}
\item \fontsize{10pt}{12pt} \selectfont
\emph{Notes:} The table re-estimates the by-occupation callback gaps of Table \ref{table:discrim-occ} (columns 2--5) after adding the full set of major-occupation-group by resume-characteristic interactions, \(G_{o(j)}\times X_i\). Each cell is the callback gap between the listed group and White men in that occupation group, expressed in percentage points. Management is the base category, so columns 2--4 are linear combinations of the own-group coefficient and its interaction with the occupation-group indicator, computed by the delta method (\texttt{lincom}). All specifications include the full set of resume controls and job-ad fixed effects, with standard errors clustered on job ads in parentheses. \sym{*}, \sym{**}, and \sym{***} denote significance at the 10, 5, and 1 percent levels. \sym{\dagger}, \sym{\dagger\dagger}, and \sym{\dagger\dagger\dagger} denote that the gap differs from the management gap at the 10, 5, and 1 percent levels. The reported \(p\)-value is for the joint test that all race/ethnicity-gender by occupation-group interactions are zero. Comparison with Table \ref{table:discrim-occ} shows that the discrimination gradient is unchanged when occupational differences in the returns to resume characteristics are absorbed.
\end{tablenotes}
\end{threeparttable}
\end{adjustbox}

%% file: tables/AER/table-judgment-demand.tex
\def\sym#1{\ifmmode^{#1}\else\(^{#1}\)\fi}
\fontsize{12pt}{14pt} \selectfont
\begin{threeparttable}
\begin{tabular}{l D{.}{.}{-1}}
\hline\hline\noalign{\vskip 0.5em}
&\multicolumn{1}{c}{Judgment demand} \\
\noalign{\vskip 0.5em}\hline\noalign{\vskip 0.5em}
\multicolumn{2}{l}{\emph{Panel A. Ad-level OLS on standardized O*NET task intensity}} \\
\noalign{\vskip 0.4em}
Analytical&0.030\sym{**} \\
&(0.015) \\\noalign{\vskip 0.25em}
Interpersonal&0.053\sym{***} \\
&(0.015) \\\noalign{\vskip 0.25em}
Routine cognitive&-0.046\sym{***} \\
&(0.011) \\\noalign{\vskip 0.25em}
Contact&-0.065\sym{***} \\
&(0.012) \\\noalign{\vskip 0.25em}
Log ad length&0.241\sym{***} \\
&(0.010) \\
\noalign{\vskip 0.4em}\hline\noalign{\vskip 0.4em}
\multicolumn{2}{l}{\emph{Panel B. Mean judgment-demand count by major occupation group}} \\
\noalign{\vskip 0.4em}
\hspace{0.25cm}Management&\multicolumn{1}{c}{1.16} \\
\hspace{0.25cm}Business and financial&\multicolumn{1}{c}{0.94} \\
\hspace{0.25cm}Sales&\multicolumn{1}{c}{0.74} \\
\hspace{0.25cm}Office and administration&\multicolumn{1}{c}{0.68} \\
\noalign{\vskip 0.4em}\hline\noalign{\vskip 0.4em}
Job ads (\(N\))&\multicolumn{1}{c}{9,220} \\
\noalign{\vskip 0.5em}\hline\hline\noalign{\vskip 0.5em}
\end{tabular}
\begin{tablenotes}
\item \fontsize{10pt}{12pt} \selectfont
\emph{Notes.} Judgment demand counts ad-text matches to a screened subset of the Leadership/decision-making and Problem-solving job-task phrase dictionaries of \citet{brencic_measuring_2026}, the attributes the model treats as subjectively assessed. Panel A regresses the standardized count on the standardized O*NET task intensities and standardized log ad length, with standard errors clustered on job ads in parentheses, so the task gradient is net of ad length. Panel B reports the mean count by major occupation group. The management group, where callback gaps are largest (Table \ref{table:discrim-occ}), shows the highest demand for judgment-based attributes. The measure captures employer demand for these attributes, not the evaluation process itself. \sym{*}, \sym{**}, \sym{***} denote significance at the 10, 5, and 1 percent levels.
\end{tablenotes}\end{threeparttable}

%% file: tables/AER/table-screening-instruments.tex
\def\sym#1{\ifmmode^{#1}\else\(^{#1}\)\fi}
\fontsize{12pt}{14pt} \selectfont
\begin{threeparttable}
\begin{tabular}{l D{.}{.}{-1}}
\hline\hline\noalign{\vskip 0.5em}
&\multicolumn{1}{c}{Objective screening instruments} \\
\noalign{\vskip 0.5em}\hline\noalign{\vskip 0.5em}
\multicolumn{2}{l}{\emph{Panel A. Share of ads mentioning each instrument (percent)}} \\
\noalign{\vskip 0.4em}
\hspace{0.25cm}Pre-employment test or assessment&\multicolumn{1}{c}{ 2.3} \\
\hspace{0.25cm}Certification or license&\multicolumn{1}{c}{23.0} \\
\hspace{0.25cm}Background or drug check&\multicolumn{1}{c}{ 7.8} \\
\hspace{0.25cm}GPA requirement&\multicolumn{1}{c}{ 0.2} \\
\hspace{0.25cm}Any of the above&\multicolumn{1}{c}{28.6} \\
\hspace{0.25cm}\emph{Memo:} experience threshold&\multicolumn{1}{c}{29.2} \\
\noalign{\vskip 0.4em}\hline\noalign{\vskip 0.4em}
\multicolumn{2}{l}{\emph{Panel B. Ad-level LPM of any-instrument on standardized O*NET task intensity}} \\
\noalign{\vskip 0.4em}
Analytical&-0.037 \\
&(0.030) \\\noalign{\vskip 0.25em}
Interpersonal&0.016 \\
&(0.040) \\\noalign{\vskip 0.25em}
Routine cognitive&0.051\sym{***} \\
&(0.018) \\\noalign{\vskip 0.25em}
Contact&0.026 \\
&(0.028) \\\noalign{\vskip 0.25em}
Log ad length&0.060\sym{***} \\
&(0.011) \\
\noalign{\vskip 0.4em}\hline\noalign{\vskip 0.4em}
\multicolumn{2}{l}{\emph{Panel C. Share of ads mentioning any instrument, by occupation group (percent)}} \\
\noalign{\vskip 0.4em}
\hspace{0.25cm}Management&\multicolumn{1}{c}{23.9} \\
\hspace{0.25cm}Business and financial&\multicolumn{1}{c}{28.7} \\
\hspace{0.25cm}Sales&\multicolumn{1}{c}{32.0} \\
\hspace{0.25cm}Office and administration&\multicolumn{1}{c}{27.7} \\
\noalign{\vskip 0.4em}\hline\noalign{\vskip 0.4em}
Job ads (\(N\))&\multicolumn{1}{c}{9,220} \\
\noalign{\vskip 0.5em}\hline\hline\noalign{\vskip 0.5em}
\end{tabular}
\begin{tablenotes}
\item \fontsize{10pt}{12pt} \selectfont
\emph{Notes.} Each instrument category counts ads whose text matches an a priori phrase dictionary naming a verifiable, standardized screening instrument, in the spirit of the job-task phrase dictionaries of \citet{brencic_measuring_2026}. Pre-employment tests and assessments name the instrument or the act of testing applicants. Bare mentions of \emph{assessment} are excluded because they commonly describe a job task. Degree requirements are excluded because the audit holds the degree constant by design. Experience thresholds are shown as a memo row and excluded from the composite because they are a posting requirement rather than a screening instrument. Panel B regresses an indicator for any instrument mention on standardized O*NET task intensities and standardized log ad length at the ad level, with standard errors clustered on 175 six-digit occupations in parentheses. Panel C reports the share of ads mentioning any instrument by major occupation group. Management, where callback gaps are concentrated (Table \ref{table:discrim-occ}), mentions these instruments less often than the other three groups. \sym{*}, \sym{**}, \sym{***} denote significance at the 10, 5, and 1 percent levels.
\end{tablenotes}\end{threeparttable}

%% file: tables/AER/table-screening-robustness.tex
\def\sym#1{\ifmmode^{#1}\else\(^{#1}\)\fi}
\fontsize{11pt}{13pt} \selectfont
\begin{threeparttable}
\begin{tabular}{l D{.}{.}{-1} D{.}{.}{-1} D{.}{.}{-1} D{.}{.}{-1} D{.}{.}{-1}}
\hline\hline\noalign{\vskip 0.5em}
&\multicolumn{1}{c}{Base rate}&\multicolumn{1}{c}{Proportional}&\multicolumn{1}{c}{Mention content}&\multicolumn{2}{c}{Excl.\ flagged ads} \\
&\multicolumn{1}{c}{OLS}&\multicolumn{1}{c}{Poisson}&\multicolumn{1}{c}{OLS}&\multicolumn{1}{c}{OLS}&\multicolumn{1}{c}{Poisson} \\
&\multicolumn{1}{c}{(1)}&\multicolumn{1}{c}{(2)}&\multicolumn{1}{c}{(3)}&\multicolumn{1}{c}{(4)}&\multicolumn{1}{c}{(5)} \\
\noalign{\vskip 0.5em}\hline\noalign{\vskip 0.5em}
Non-White-male (gap at mean of moderators)&-0.0080\sym{**}&-0.0244&-0.0079\sym{**}&-0.0077\sym{*}&-0.0166 \\
&(0.0039)&(0.0308)&(0.0039)&(0.0040)&(0.0332) \\\noalign{\vskip 0.25em}
Non-White-male \(\times\) screening prevalence (per SD)&0.0075\sym{**}&0.0682\sym{**}&&0.0109\sym{**}&0.0852\sym{**} \\
&(0.0038)&(0.0338)&&(0.0042)&(0.0383) \\\noalign{\vskip 0.25em}
Non-White-male \(\times\) White-male callback rate (leave-out, per SD)&-0.0034&&&& \\
&(0.0044)&&&& \\\noalign{\vskip 0.25em}
Non-White-male \(\times\) substantive credential prevalence (per SD)&&&0.0093\sym{**}&& \\
&&&(0.0038)&& \\\noalign{\vskip 0.25em}
Non-White-male \(\times\) non-substantive mention prevalence (per SD)&&&0.0050&& \\
&&&(0.0037)&& \\
\noalign{\vskip 0.4em}\hline\noalign{\vskip 0.4em}
Sample&\multicolumn{1}{c}{Full}&\multicolumn{1}{c}{Full}&\multicolumn{1}{c}{Full}&\multicolumn{1}{c}{Excl.\ flagged}&\multicolumn{1}{c}{Excl.\ flagged} \\
\(N\) (applications)&\multicolumn{1}{c}{36,536}&\multicolumn{1}{c}{9,960}&\multicolumn{1}{c}{36,880}&\multicolumn{1}{c}{35,256}&\multicolumn{1}{c}{9,612} \\
\noalign{\vskip 0.5em}\hline\hline\noalign{\vskip 0.5em}
\end{tabular}
\begin{tablenotes}
\item \fontsize{10pt}{12pt} \selectfont
\emph{Notes.} All columns regress callback on a pooled non-White-male indicator and the listed interactions with job-ad fixed effects and the full set of r\'esum\'e controls, with standard errors clustered on job ads in parentheses. Screening prevalence is the share of a six-digit occupation's ads mentioning any objective screening instrument (Table \ref{table:screening-instruments}), standardized across ads. Column 1 adds an interaction with the occupation's White-male callback rate computed excluding the focal job ad (86 single-White-male-application occupations drop). Columns 2 and 5 are Poisson pseudo-maximum-likelihood estimates (job ads without callback variation drop), so the minority coefficient is a log callback ratio and the interaction tests proportional rather than level compression. Column 3 splits certification and license mentions into substantive credentials (licenses obtainable after hire, professional certifications, licensed-status language) and non-substantive matches (driver's licenses, task and product false positives, firm self-description, bare section headers), each aggregated to occupation-level prevalence and standardized. Columns 4 and 5 exclude the 406 ads whose text contains language consistent with a license requirement at application, a deliberately over-inclusive flag from the advertisement-text audit. With clustering on six-digit occupations, the p-values for the headline coefficient are 0.026, 0.026, 0.000, 0.007, and 0.035 in columns 1 through 5. \sym{*}, \sym{**}, \sym{***} denote significance at the 10, 5, and 1 percent levels.
\end{tablenotes}\end{threeparttable}

%% file: tables/AER/table-pool-robustness.tex
\def\sym#1{\ifmmode^{#1}\else\(^{#1}\)\fi}
\fontsize{12pt}{14pt} \selectfont
\begin{threeparttable}
\begin{tabular}{l D{.}{.}{-1} D{.}{.}{-1}}
\hline\hline\noalign{\vskip 0.5em}
&\multicolumn{1}{c}{All five groups}&\multicolumn{1}{c}{Black and Hispanic only} \\
&\multicolumn{1}{c}{(1)}&\multicolumn{1}{c}{(2)} \\
\noalign{\vskip 0.5em}\hline\noalign{\vskip 0.5em}
\multicolumn{3}{l}{\emph{Panel A. Mechanism decomposition, \(\hat{B}-\hat{P}\) difference (Table \ref{table:mechanism-contact})}} \\
\noalign{\vskip 0.4em}
\hspace{0.25cm}Full sample&-0.0124\sym{**}&-0.0127\sym{*} \\
&(0.0063)&(0.0066) \\\noalign{\vskip 0.25em}
\hspace{0.25cm}High-contact jobs&-0.0427\sym{***}&-0.0459\sym{***} \\
&(0.0126)&(0.0136) \\
\noalign{\vskip 0.4em}\hline\noalign{\vskip 0.4em}
\multicolumn{3}{l}{\emph{Panel B. Credential triple interactions, joint \(p\)-values (Table \ref{table:credential-attenuation})}} \\
\noalign{\vskip 0.4em}
\hspace{0.25cm}Positive-return credentials&\multicolumn{1}{c}{0.048}&\multicolumn{1}{c}{0.018} \\
\hspace{0.25cm}Placebo credentials&\multicolumn{1}{c}{0.675}&\multicolumn{1}{c}{0.669} \\
\noalign{\vskip 0.4em}\hline\noalign{\vskip 0.4em}
\multicolumn{3}{l}{\emph{Panel C. Non-White-male \(\times\) screening prevalence (Table \ref{table:screening-gaps})}} \\
\noalign{\vskip 0.4em}
\hspace{0.25cm}Levels&0.0092\sym{**}&0.0087\sym{**} \\
&(0.0038)&(0.0041) \\\noalign{\vskip 0.25em}
\hspace{0.25cm}Poisson&0.0682\sym{**}&0.0607\sym{*} \\
&(0.0338)&(0.0363) \\
\noalign{\vskip 0.4em}\hline\noalign{\vskip 0.4em}
\(N\) (applications, audit sample)&\multicolumn{1}{c}{36,880}&\multicolumn{1}{c}{30,705} \\
\noalign{\vskip 0.5em}\hline\hline\noalign{\vskip 0.5em}
\end{tabular}
\begin{tablenotes}
\item \fontsize{10pt}{12pt} \selectfont
\emph{Notes.} Column 1 pools all five non-White-male groups, the paper's definition. Column 2 drops White women, so the pool contains Black and Hispanic applicants only, with White men remaining the reference group. Panel A reports the pooled \(\hat{B}-\hat{P}\) difference from the specification of Table \ref{table:mechanism-contact}, in the full sample and in the top 40 percent of contact intensity. Panel B reports joint \(p\)-values for the credential triple interactions of Table \ref{table:credential-attenuation}, separately for the three positive-return credentials and the three placebo credentials. Panel C reports the Non-White-male \(\times\) screening prevalence interaction of Table \ref{table:screening-gaps} in levels and from the Poisson specification of Table \ref{table:screening-robustness}. Standard errors clustered on job ads are in parentheses. With clustering on six-digit occupations, the full-sample Panel A \(p\)-values are 0.051 and 0.061, and the Panel C levels \(p\)-values are 0.004 and 0.004, in columns 1 and 2. \sym{*}, \sym{**}, \sym{***} denote significance at the 10, 5, and 1 percent levels.
\end{tablenotes}\end{threeparttable}

%% file: tables/AER/table-occgroup-poisson.tex
\def\sym#1{\ifmmode^{#1}\else\(^{#1}\)\fi}
\fontsize{12pt}{14pt} \selectfont
\begin{threeparttable}
\begin{tabular}{l D{.}{.}{-1} D{.}{.}{-1} D{.}{.}{-1} D{.}{.}{-1} D{.}{.}{-1}}
\hline\hline\noalign{\vskip 0.5em}
&\multicolumn{1}{c}{Full sample}&\multicolumn{1}{c}{Management}&\multicolumn{1}{c}{Bus.\ and Fin.}&\multicolumn{1}{c}{Sales}&\multicolumn{1}{c}{Office/Admin.} \\
&\multicolumn{1}{c}{(1)}&\multicolumn{1}{c}{(2)}&\multicolumn{1}{c}{(3)}&\multicolumn{1}{c}{(4)}&\multicolumn{1}{c}{(5)} \\
\noalign{\vskip 0.5em}\hline\noalign{\vskip 0.5em}
White Women&-0.031&-0.203\sym{**}&-0.035&-0.070\sym{*}&0.071 \\
&(0.032)&(0.092)&(0.160)&(0.041)&(0.079) \\\noalign{\vskip 0.25em}
Black Women&-0.095\sym{***}&-0.271\sym{**}&0.152&-0.114\sym{***}&-0.186\sym{**} \\
&(0.035)&(0.120)&(0.152)&(0.043)&(0.084) \\\noalign{\vskip 0.25em}
Hispanic Women&0.027&-0.122&0.018&0.013&0.046 \\
&(0.033)&(0.107)&(0.138)&(0.043)&(0.083) \\\noalign{\vskip 0.25em}
Black Men&-0.141\sym{***}&-0.328\sym{***}&-0.057&-0.147\sym{***}&-0.128 \\
&(0.034)&(0.099)&(0.152)&(0.043)&(0.083) \\\noalign{\vskip 0.25em}
Hispanic Men&-0.014&-0.266\sym{**}&-0.123&0.017&-0.007 \\
&(0.035)&(0.114)&(0.168)&(0.043)&(0.087) \\\noalign{\vskip 0.25em}
\noalign{\vskip 0.15em}\hline\noalign{\vskip 0.4em}
White-male callback rate&\multicolumn{1}{c}{0.149}&\multicolumn{1}{c}{0.154}&\multicolumn{1}{c}{0.043}&\multicolumn{1}{c}{0.245}&\multicolumn{1}{c}{0.101} \\
\(N\) (applications)&\multicolumn{5}{c}{9,960} \\
\noalign{\vskip 0.5em}\hline\hline\noalign{\vskip 0.5em}
\end{tabular}
\begin{tablenotes}
\item \fontsize{10pt}{12pt} \selectfont
\emph{Notes.} Poisson pseudo-maximum-likelihood (PPML) analogue of Table \ref{table:discrim-occ}. Each cell reports the log callback ratio between the row group and White men, so the estimates are proportional gaps and do not depend on group base rates. Column 1 uses the full sample. Columns 2 through 5 come from a single regression interacting group indicators with major occupation group, with management as the omitted category and the other columns computed as linear combinations. All specifications include job-ad fixed effects and the full set of r\'esum\'e controls. Job ads without callback variation drop from the Poisson sample. Standard errors clustered on job ads are in parentheses, with \(p\)-values from the normal distribution. A joint test of all group \(\times\) occupation-group interactions rejects equal proportional gaps across occupation groups (\(p =\) 0.027). \sym{*}, \sym{**}, \sym{***} denote significance at the 10, 5, and 1 percent levels.
\end{tablenotes}\end{threeparttable}

%% file: tables/AER/table5-discrim-tasks-kmeans4-beststart.tex
\begin{table}[htbp]
    \def\sym#1{\ifmmode^{#1}\else\(^{#1}\)\fi}
    \fontsize{8pt}{9pt} \selectfont
    \centering
    \caption{Discrimination Across Clusters of Jobs with Different Task Content (Best-of-500-Starts K-Means)}
    \label{table:reg-task-discrim-kmeans4-beststart}
    \begin{adjustbox}{width=\textwidth}
    \begin{threeparttable}
        \begin{tabular}{l*{4}{D{.}{.}{-1}}}
\hline\hline
\noalign{\vskip 0.5em}
&\multicolumn{1}{c}{\hspace{0.5cm}(1)}&\multicolumn{1}{c}{\hspace{0.5cm}(2)}&\multicolumn{1}{c}{\hspace{0.5cm}(3)}&\multicolumn{1}{c}{\hspace{0.5cm}(4)} \\
\noalign{\vskip 0.5em}
\hline
\noalign{\vskip 0.5em}
White Women&-0.015\sym{*}&-0.001&0.009&-0.011 \\
&(0.009)&(0.010)&(0.009)\sym{\dagger}&(0.009) \\
Black Women&-0.020\sym{**}&-0.008&-0.012&-0.017\sym{*} \\
&(0.009)&(0.011)&(0.009)&(0.009) \\
Hispanic Women&-0.001&0.011&0.004&-0.001 \\
&(0.009)&(0.011)&(0.008)&(0.009) \\
Black Men&-0.035\sym{***}&-0.015&-0.019\sym{**}&-0.006 \\
&(0.009)&(0.011)&(0.009)&(0.010)\sym{\dagger\dagger} \\
Hispanic Men&-0.018\sym{**}&0.007&-0.011&-0.002 \\
&(0.009)&(0.011)\sym{\dagger}&(0.009)&(0.010) \\
&&&& \\
Obs. in Cluster&\multicolumn{1}{c}{\hspace{0.275cm}11,360}&\multicolumn{1}{c}{\hspace{0.275cm}11,336}&\multicolumn{1}{c}{\hspace{0.45cm}8,580}&\multicolumn{1}{c}{\hspace{0.45cm}5,604} \\
\hspace{0.25cm}\% Management&0.304&0.000&0.000&0.006 \\
\hspace{0.25cm}\% Business and Finance&0.342&0.000&0.338&0.185 \\
\hspace{0.25cm}\% Sales&0.155&0.997&0.000&0.056 \\
\hspace{0.25cm}\% Office and Admin.&0.073&0.000&0.620&0.667 \\
&&&& \\
Centile Mean/Median/Std. Dev.&&&& \\
\hspace{0.25cm}Analytical&\multicolumn{1}{c}{\hspace{-.15cm}85/89/11}&\multicolumn{1}{c}{\hspace{-.15cm}52/52/13}&\multicolumn{1}{c}{\hspace{-.15cm}63/57/14}&\multicolumn{1}{c}{\hspace{-.15cm}33/34/13} \\
\hspace{0.25cm}Interpersonal&\multicolumn{1}{c}{\hspace{-.15cm}86/91/14}&\multicolumn{1}{c}{\hspace{-.15cm}33/36/18}&\multicolumn{1}{c}{\hspace{-.15cm}60/61/11}&\multicolumn{1}{c}{\hspace{-.15cm}29/27/12} \\
\hspace{0.25cm}R. Cognitive&\multicolumn{1}{c}{\hspace{-.15cm}25/21/17}&\multicolumn{1}{c}{\hspace{-.15cm}33/42/15}&\multicolumn{1}{c}{\hspace{-.15cm}82/87/\hspace{0.14cm}9}&\multicolumn{1}{c}{\hspace{-.15cm}82/80/14} \\
\hspace{0.25cm}R. Manual&\multicolumn{1}{c}{\hspace{-.15cm}13/\hspace{0.14cm}9/10}&\multicolumn{1}{c}{\hspace{-.15cm}11/\hspace{0.14cm}8/\hspace{0.14cm}8}&\multicolumn{1}{c}{\hspace{-.15cm}29/40/12}&\multicolumn{1}{c}{\hspace{-.15cm}45/45/11} \\
\hspace{0.25cm}Physical&\multicolumn{1}{c}{\hspace{-.15cm}12/\hspace{0.14cm}9/\hspace{0.14cm}8}&\multicolumn{1}{c}{\hspace{-.15cm}14/11/11}&\multicolumn{1}{c}{\hspace{-.15cm}16/19/\hspace{0.14cm}8}&\multicolumn{1}{c}{\hspace{-.15cm}31/28/13} \\
\hspace{0.25cm}Contact&\multicolumn{1}{c}{\hspace{-.15cm}58/67/23}&\multicolumn{1}{c}{\hspace{-.15cm}88/89/11}&\multicolumn{1}{c}{\hspace{-.15cm}77/82/18}&\multicolumn{1}{c}{\hspace{-.15cm}64/67/23} \\
\noalign{\vskip 0.5em}
\hline\hline
\noalign{\vskip 0.5em}
        \end{tabular}
    \begin{tablenotes}
    \item \fontsize{7pt}{8pt} \selectfont  \emph{Notes:} Identical to Table \ref{table:reg-task-discrim-kmeans4} except that the K-means partition is selected as the minimum within-cluster sum of squares across 500 random starting values rather than a single random start, and clusters are ordered by descending analytical-plus-interpersonal minus routine-cognitive profile. The table presents percentage point differences in callback rates between White men and the other race/ethnicity-gender groups across 4 clusters of job ads, from one regression in which cluster indicators are captured by \(G_{o(j)}\) in equation (\ref{eq_interaction}). Column 1 reports base-group estimates, and columns 2--4 report linear combinations computed with the delta method. Specifications include the full set of r\'{e}sum\'{e} controls and ad fixed effects. Standard errors with clustering on job ads are in parentheses. \sym{*}, \sym{**}, and \sym{***} indicate statistical significance at the 10, five, and one percent levels, respectively. \sym{\dagger}, \sym{\dagger\dagger}, and \sym{\dagger\dagger\dagger} indicate whether the estimates in columns 2, 3, and 4 differ statistically from column 1 at the 10, five, and one percent levels, respectively. The middle and lower portions report occupation-group shares and task-intensity centile statistics for each cluster.
    \end{tablenotes}
    \end{threeparttable}
    \end{adjustbox}
\end{table}

%% file: tables/AER/table-complexity-horserace.tex
\def\sym#1{\ifmmode^{#1}\else\(^{#1}\)\fi}
\fontsize{12pt}{14pt} \selectfont
\begin{threeparttable}
\begin{tabular}{l D{.}{.}{-1} D{.}{.}{-1} D{.}{.}{-1} D{.}{.}{-1} D{.}{.}{-1}}
\hline\hline\noalign{\vskip 0.5em}
&\multicolumn{1}{c}{Baseline}&\multicolumn{1}{c}{Wage}&\multicolumn{1}{c}{Education}&\multicolumn{1}{c}{Job zone}&\multicolumn{1}{c}{All} \\
&\multicolumn{1}{c}{(1)}&\multicolumn{1}{c}{(2)}&\multicolumn{1}{c}{(3)}&\multicolumn{1}{c}{(4)}&\multicolumn{1}{c}{(5)} \\
\noalign{\vskip 0.5em}\hline\noalign{\vskip 0.5em}
Non-White-male \(\times\) (\(\hat{B} - \hat{P}\)) difference&-0.012\sym{**}&-0.013\sym{*}&-0.015\sym{**}&-0.016\sym{**}&-0.014\sym{*} \\
&(0.006)&(0.007)&(0.007)&(0.007)&(0.007) \\\noalign{\vskip 0.25em}
Non-White-male \(\times\) log median wage&&-0.004&&&-0.017\sym{*} \\
&&(0.006)&&&(0.009) \\\noalign{\vskip 0.25em}
Non-White-male \(\times\) college share&&&0.003&&0.011 \\
&&&(0.006)&&(0.008) \\\noalign{\vskip 0.25em}
Non-White-male \(\times\) job zone&&&&0.011&0.012 \\
&&&&(0.007)&(0.008) \\
\noalign{\vskip 0.4em}\hline\noalign{\vskip 0.4em}
Occupation-clustered \(p\), \(\hat{B} - \hat{P}\)&\multicolumn{1}{c}{0.051}&\multicolumn{1}{c}{0.016}&\multicolumn{1}{c}{0.007}&\multicolumn{1}{c}{0.024}&\multicolumn{1}{c}{0.005} \\
Wild cluster bootstrap \(p\), \(\hat{B} - \hat{P}\)&\multicolumn{1}{c}{0.150}&\multicolumn{1}{c}{0.102}&\multicolumn{1}{c}{0.075}&\multicolumn{1}{c}{0.102}&\multicolumn{1}{c}{0.036} \\
\(N\) (applications)&\multicolumn{1}{c}{36,880}&\multicolumn{1}{c}{34,420}&\multicolumn{1}{c}{34,420}&\multicolumn{1}{c}{36,868}&\multicolumn{1}{c}{34,420} \\
\noalign{\vskip 0.5em}\hline\hline\noalign{\vskip 0.5em}
\end{tabular}
\begin{tablenotes}
\item \fontsize{10pt}{12pt} \selectfont
\emph{Notes.} The table interacts the pooled non-White-male indicator simultaneously with the standardized discretion components and with occupation-level measures of job quality. The first row reports the difference between the Non-White-male \(\times\) \(\hat{B}\) and Non-White-male \(\times\) \(\hat{P}\) interactions, computed with the delta method. Median wage and the college share of employment are computed from the 2015--2018 American Community Survey for employed workers aged 18 to 65, and the job zone is the O*NET preparation-level scale (database 22.0). Each quality measure is standardized to mean zero and unit variance across applications. Across job advertisements, the correlations between \(\hat{B} - \hat{P}\) and the three quality measures range from 0.54 to 0.61. All specifications include job-ad fixed effects, the full set of r\'esum\'e controls, and minority interactions with manual task intensity and contact. Standard errors clustered on job ads are in parentheses. The memo rows report the \(p\)-value for the \(\hat{B} - \hat{P}\) difference with clustering on the six-digit occupations and from a wild cluster bootstrap at the occupation level (Rademacher weights, 9,999 replications imposing the null). Six occupation codes lack ACS wage and education measures and one lacks a job zone, so applications in those occupations drop from the corresponding columns. \sym{*}, \sym{**}, \sym{***} denote significance at the 10, 5, and 1 percent levels.
\end{tablenotes}\end{threeparttable}

%% file: tables/AER/table-wm-baserates.tex
\begin{table}[htbp]
    \centering
    \caption{White-Male Callback Rates by Task Group}
    \label{table:wm-baserates}
    \begin{threeparttable}
    \begin{tabular}{lcc}
    \hline\hline
    \noalign{\vskip 0.75em}
    Task Group&\multicolumn{1}{c}{WM Callback Rate}&\multicolumn{1}{c}{N (WM)} \\
    \noalign{\vskip 0.5em}
    \hline
    \noalign{\vskip 0.75em}
    Overall&0.149& \\
    \noalign{\vskip 0.5em}
    \emph{Analytical}&& \\
    \hspace{0.25cm} Low&0.150&2794 \\
    \hspace{0.25cm} High&0.148&3394 \\
    \emph{Interpersonal}&& \\
    \hspace{0.25cm} Low&0.157&2873 \\
    \hspace{0.25cm} High&0.143&3315 \\
    \emph{R. Cognitive}&& \\
    \hspace{0.25cm} Low&0.222&3110 \\
    \hspace{0.25cm} High&0.076&3078 \\
    \emph{Contact}&& \\
    \hspace{0.25cm} Low&0.123&2901 \\
    \hspace{0.25cm} High&0.172&3287 \\
    \noalign{\vskip 0.5em}
    \emph{Discretion (\(E^\ast\) Median Split)}&& \\
    \hspace{0.25cm} Low Discretion&0.149&2638 \\
    \hspace{0.25cm} High Discretion&0.149&3550 \\
    \noalign{\vskip 0.75em}
    \hline\hline
    \noalign{\vskip 0.75em}
    \end{tabular}
    \begin{tablenotes}
    \item \fontsize{10pt}{12pt} \selectfont 
    \emph{Notes:} White-male callback rates by task-intensity group. ``Low'' and ``High'' refer to below- and at-or-above-median task intensity, respectively. ``High Discretion'' indicates jobs with evaluative discretion (\(E^\ast = \hat{B}/(\hat{B}+\hat{P})\)) above the sample median. These rates serve as group-specific denominators for computing relative callback gaps.
    \end{tablenotes}
    \end{threeparttable}
\end{table}